\documentclass[aps,pra,twocolumn,amssymb,amsmath,amsfonts,showkeys,nofootinbib,floatfix,superscriptaddress,longbibliography]{revtex4-2}

\usepackage[english]{babel}
\usepackage{amsthm}
\usepackage[toc,page]{appendix}
\usepackage[colorlinks=true,citecolor=blue,linkcolor=magenta]{hyperref}
\usepackage{bbm}
\usepackage{graphicx}
\usepackage{stackengine,xcolor}
\usepackage{tcolorbox}
\usepackage{nicematrix}
\usepackage{cleveref}
\usepackage{mathtools}
\usepackage{quantikz}
\usepackage{adjustbox}
\usepackage{mathrsfs}  
\newcommand{\poly}{\operatorname{poly}}

\newcommand{\Tr}{{\rm Tr}}

\newcommand{\g}{\mathfrak{g}}
\newcommand{\su}{\mathfrak{su}}
\newcommand{\so}{\mathfrak{so}}

\newcommand{\SO}{\mathbb{SO}}
\newcommand{\SU}{\mathbb{SU}}

\renewcommand{\vec}[1]{\boldsymbol{#1}}

\newcommand{\id}{\openone}

\newcommand{\Eig}{\text{Eig}}

\newcommand{\GC}{\mathcal{G}}
\newcommand{\HC}{\mathcal{H}}

\newcommand{\OC}{\mathcal{O}}

\newcommand{\ZBB}{\mathbb{Z}}

\newcommand{\Par}{\mathscr{P}}

\newtheorem{theorem}{Theorem}
\newtheorem{problem}{Problem}
\newtheorem{lemma}{Lemma}
\newtheorem{corollary}{Corollary}

\newtheorem{proposition}{Proposition}
\newtheorem{example}{Example}
\newtheorem{definition}{Definition}

\begin{document}

\title{Matchgate synthesis via Clifford matchgates and $T$ gates}

\author{Berta Casas}
\affiliation{Theoretical Division, Los Alamos National Laboratory, Los Alamos, New Mexico 87545, USA}
\affiliation{Barcelona Supercomputing Center, Plaça Eusebi G\"uell, 1-3, 08034 Barcelona, Spain}
\affiliation{Universitat de Barcelona, 08007 Barcelona, Spain}

\author{Paolo Braccia}
\affiliation{Theoretical Division, Los Alamos National Laboratory, Los Alamos, New Mexico 87545, USA}

\author{Élie Gouzien}
\affiliation{Alice \& Bob, 53 boulevard du Général Martial Valin, 75\,015 Paris, France}
	
\author{M. Cerezo}
\thanks{cerezo@lanl.gov}
\affiliation{Information Sciences, Los Alamos National Laboratory, Los Alamos, New Mexico 87545, USA}

\author{Diego Garc\'ia-Mart\'in}
\affiliation{Information Sciences, Los Alamos National Laboratory, Los Alamos, New Mexico 87545, USA}
\affiliation{Department for Quantum Information and Computation at Kepler (QUICK),\\ Johannes Kepler University, Linz, Austria }

\begin{abstract}
    Matchgate unitaries are ubiquitous in quantum computation due to their relation to non-interacting fermions and because they can be used to benchmark quantum computers. Implementing such unitaries on fault-tolerant devices requires first compiling them into a discrete universal gate set, typically Clifford$+T$. Here, we propose a different approach for their synthesis: compile matchgate unitaries using only matchgate gates. To this end, we first show that the matchgate-Clifford group (the intersection of the matchgate and Clifford groups) plus the $\overline{T}$ gate (a $T$ unitary up to a phase) is universal for the matchgate group. Our approach leverages the connection between $n$-qubit matchgate circuits and the standard representation of $\mathbb{SO}(2n)$, which reduces the compilation from $2^n\times 2^n$ unitaries to $2n\times2n$ ones, thus reducing exponentially the size of the target matrix. Moreover, we rigorously show that this scheme is efficient, as an approximation error $\varepsilon_{\mathbb{SO}(2n)}$ incurred in this smaller-dimensional representation translates at most into an $\mathcal{O}(n \,\varepsilon_{\mathbb{SO}(2n)})$ error in the exponentially large unitary. 
    In addition, we study the exact version of the matchgate synthesis problem, and we prove that all matchgate unitaries $U$ such that $U\otimes U^*$ has entries in the ring $\mathbb{Z}\big[1/\sqrt 2,i\big]$ can be exactly synthesized by a finite sequence of gates from  the matchgate-Clifford$+\overline{T}$ set, without ancillas. We then use this insight to map optimal exact  matchgate synthesis to Boolean satisfiability, and compile the circuits that diagonalize the free-fermionic $XX$ Hamiltonian on $n=4,\,8$ qubits.
\end{abstract}

\maketitle

\section{Introduction}

Despite tremendous advancements in quantum technologies, errors continue to be one of the main bottlenecks for solving large-scale problems in quantum computers.  In this setting, error mitigation techniques alone are insufficient, as their cost typically grows rapidly with the number of qubits~\cite{wang2021can,endo2021hybrid}. Instead, scalable architectures rely on quantum error-correcting (QEC) codes to protect logical information against physical noise~\cite{gottesman2009introduction}. Crucially, these codes support only a finite set of logical gates that can be implemented fault-tolerantly~\cite{eastin2009restrictions}, and  general unitaries must be compiled or \emph{synthesized} into sequences over such discrete sets~\cite{kitaev1997quantum,kitaev2002classical,dawson2005solovay}.  

For most leading QEC codes, e.g., the surface code~\cite{fowler2012surface, campbell2017roads}, Clifford gates constitute the ``easy'' part of the logical toolbox~\cite{litinski2019game}, as they can be implemented transversally or with modest overhead~\cite{fowler2018low,litinski2019game}. In contrast, non-Clifford gates such as the $T$ gate require more involved protocols, including magic-state distillation or cultivation~\cite{bravyi2005universal,gidney2024magic,chamberland2019fault}, and code switching~\cite{paetznick2013universal,bombin2015gauge}. These protocols typically dominate the space–time complexity of fault-tolerant quantum computation~\cite{campbell2017roads, beverland2022assessing}. Accordingly, the number of $T$ gates (the \emph{$T$-count}) and their sequential structure (the \emph{$T$-depth}) emerge as central resources that must be carefully optimized. This has motivated a large body of work on approximate and exact synthesis over Clifford+$T$ gate sets, both at the single-qubit~\cite{kliuchnikov2013fast, ross2014optimal,kim2025catalytic} and multi-qubit level~\cite{giles2013exact}, including measurement-assisted probabilistic methods~\cite{bocharov2015efficient}, as well as increasingly sharp bounds on the minimal $T$-count required to implement generic unitaries~\cite{kliuchnikov2023shorter, morisaki2025optimal}.

\begin{figure*}[t]
    \centering
    \includegraphics[width=\linewidth]{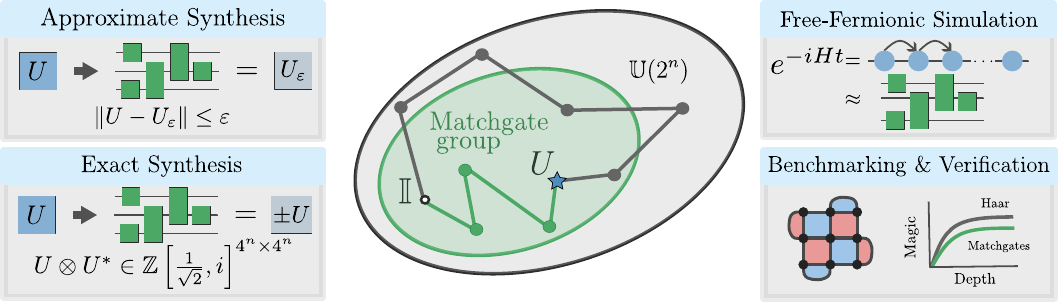}
    \caption{{\bf Summary of our main results.} We propose synthesizing a target matchgate unitary $U$--either approximately or exactly--using only matchgates, as indicated by the green path within the matchgate group. In contrast to the standard approach, where a universal gate set for the full unitary group $\mathbb{U}(2^n)$ (such as Clifford$+T$)  is employed (grey path), our strategy allows us to work with $2n\times 2n$ matrices, instead of $2^n\times 2^n$ ones. Our results may find broad applicability in the simulation of fermionic systems on quantum computers, as well as in benchmarking and verification protocols thereof.}\label{fig:fig1}
\end{figure*}

In parallel, there is a well-developed literature on the compilation of Clifford circuits themselves~\cite{bravyi2021clifford, yang2024harnessing,peham2023depth, webster2025heuristic}. Although not computationally universal, Clifford circuits play a central role in QEC and randomized protocols such as classical shadows~\cite{huang2020predicting,west2024real}. Furthermore, when acting on stabilizer states, they can be represented by matrices of polynomial size in the number of qubits $n$, and simulated efficiently  on a classical computer~\cite{aaronson2004improved,gidney2021stim}. 
This has motivated algorithms for exact  depth-optimal synthesis of Clifford circuits~\cite{peham2023depth} with exponentially improved scaling compared to general unitary synthesis ones~\cite{gouzien2025provably}.
The applications range from fault-tolerant gadget optimization~\cite{chamberland2019fault} to randomized benchmarking~\cite{knill2008randomized}, and more. These efforts suggest that compiling within a structurally or physically motivated subgroup of the full unitary group can be practically relevant. 

Within this landscape, the matchgate group occupies a special position. It is a well-studied subgroup of the unitary group~\cite{valiant2001quantum,knill2001fermionic,terhal2002classical,divincenzo2005fermionic,bravyi2004lagrangian,jozsa2008matchgates,jozsa2010matchgate,ramelow2010matchgate,brod2011extending,brod2014computational,brod2016efficient,oszmaniec2017universal,zhao2021fermionic,wan2022matchgate,helsen2022matchgate,diaz2023showcasing,mele2024efficient,stroeks2024solving}, 
and matchgate circuits are ubiquitous in quantum simulation primitives because of their connection with free-fermionic (Gaussian) evolutions~\cite{terhal2002classical, divincenzo2005fermionic}. Moreover, they can be succinctly described in terms of polynomial-size matrices from the special orthogonal group $\SO(2n)$ acting on Majorana modes~\cite{bravyi2004lagrangian}. In addition, matchgate circuits are efficiently simulable on a classical computer for computational-basis input states and measurements~\cite{valiant2001quantum}. Unlike Clifford circuits, they can generate magic~\cite{veitch2014resource,howard2017application,chitambar2019quantum, leone2022stabilizer} and thus be used in scalable benchmarking protocols that probe fault-tolerant regimes. Not only that, matchgate-type interactions such as $R^{xx}(\frac{\pi}{2})$ gates are native operations on several prominent hardware platforms, including trapped-ion~\cite{chen2024benchmarking} and neutral-atom architectures~\cite{wurtz2023aquila}.  These properties render the matchgate group a natural target for specialized compilation and benchmarking.

In this work, we establish a framework for the synthesis of matchgate circuits using a discrete, fault-tolerant matchgate set. We start by identifying such a gate set, proving that the matchgate-Clifford group (i.e., the intersection of the matchgate and Clifford groups) plus the $\overline{T}$ gate--the usual $T$ gate up to a global phase--is universal for matchgate computation. That is, we find that the matchgate-Clifford$+\overline{T}$ set can approximate any matchgate unitary to arbitrary precision. Through the well-known isomorphism between matchgate circuits and $\SO(2n)$, we further characterize our universal matchgate set as a set of $2n\times 2n$ matrices in the standard representation of $\SO(2n)$. Compiling in this representation has the clear advantage that the matrices to be synthesized are of polynomial size in $n$, instead of exponential. However, to ensure the validity of this approach, we quantify how approximation errors in $\SO(2n)$ lift to errors in the large $2^n\times 2^n$ unitaries. We prove that the errors amplify at most linearly with $n$, which introduces a $\OC\left(\log^c n\right)$ overhead (for some constant $0<c<2$), thus guaranteeing an overall favorable scaling.

We then address the exact synthesis problem for matchgate unitaries, and we show that every matchgate unitary $U$ such that $U\otimes U^*$ has entries in the ring $\mathbb{Z}\left[\frac{1}{\sqrt{2}},i\right]$ can be exactly synthesized by a finite sequence of gates from our universal matchgate set, without ancillas (see Sec.~\ref{sec:exact_synthetis} for the definition of the ring). Our proof technique is a classical synthesis algorithm whose runtime is quartic in the number of qubits and linear in the least denominator exponent~\cite{giles2013exact} of the corresponding $\SO(2n)$ matrix (a quantity that roughly measures how many bits are required to specify the entries of the matrix, and that is directly related to the $\overline{T}$-count).  Furthermore, the analysis of this synthesis algorithm provides us with explicit upper bounds on the number of $\overline{T}$ gates and Clifford gates required for exact matchgate synthesis.

Finally, we explore global exact compilation strategies. In particular, we use the previous insights to map the decision version of the exact matchgate synthesis problem to Boolean satisfiability (SAT), following the procedure in Ref.~\cite{gouzien2025provably}. Then, by performing a binary search on the depth, we obtain optimal- or near-optimal-depth circuits~\footnote{Optimal at the logical level, without taking into account the specifics of error-correcting codes.}. Besides, we employ the optimization version of SAT (namely MAX--SAT, which imposes soft constraints) to search for circuits with a reduced $\overline{T}$-count, among those of optimal depth. We showcase this SAT-based method by compiling circuits that diagonalize the free-fermionic $XX$ Hamiltonian on $n = 4,\, 8$ qubits~\cite{verstraete2009quantum}. We stress that for general $2^n\times 2^n$ unitaries, the SAT solver runs in time doubly-exponential in the depth of the circuit~\cite{gouzien2025provably}, whereas our matchgate-specific approach reduces this to exponential complexity. Thus, finding such circuit for $n=8$ qubits (see Fig.~\ref{fig:xx_model}) is completely out of reach using e.g., standard Clifford$+T$ compilation methods. Our main results are overall summarized in Fig.~\ref{fig:fig1}.

\section{Preliminaries}\label{sec:pre}
In this section, we briefly review the general unitary synthesis problem (universal gate sets and approximate compilation in~\ref{sec:pre_gate_sets}, and exact synthesis in~\ref{sec:pre_exact_synthesis}), together with the matchgate group in~\ref{sec:pre_matchgates}. This will provide the necessary context for our results.

\subsection{Approximate unitary synthesis}\label{sec:pre_gate_sets}

 Let us begin by setting the notation. We denote the usual Pauli matrices as
	\begin{equation}
		X= \begin{pmatrix}
			0 & 1 \\
			1 & 0
		\end{pmatrix}\,,\quad Y= \begin{pmatrix}
		0 & -i \\
		i &0
	\end{pmatrix}\,,\quad  Z = \begin{pmatrix}
	1 & 0 \\
	0 & -1
\end{pmatrix}\,.
	\end{equation}
We will furthermore consider the following single-qubit unitary gates: the Hadamard gate,
	\begin{equation}\label{eq:Hadamard}
		H= \frac{1}{\sqrt{2}} \begin{pmatrix}
			1 & 1 \\ 
			1 & -1
		\end{pmatrix} = \frac{1}{\sqrt{2}} \left(X+Z\right)\,,
	\end{equation}
the $W$ gate~\cite{quantumgates},
\begin{equation}\label{eq:W}
	W= \frac{1}{\sqrt{2}} \begin{pmatrix}
		1 & -i \\ 
		i & -1
	\end{pmatrix} = \frac{1}{\sqrt{2}} \left(Y+Z\right)\,,
\end{equation}
the $S$ gate,
\begin{equation}\label{eq:S_gate}
		S=  \begin{pmatrix}
			1 & 0 \\
			0 & i
		\end{pmatrix}\,,
	\end{equation}
and the $T$ gate,
	\begin{equation}\label{eq:T_gate}
		T=  \begin{pmatrix}
			1 & 0 \\
			0 & e^{i\pi/4}
		\end{pmatrix}\,,
	\end{equation}
together with the two-qubit $CNOT$ gate,
\begin{equation}\label{eq:CNOT}
	CNOT = \begin{pmatrix}
		1 & 0 & 0 & 0 \\
		0 & 1 & 0 & 0 \\
		0 & 0 & 0 & 1 \\
		0 & 0 & 1 & 0 
	\end{pmatrix}\,.
\end{equation}

The main goal of approximate unitary synthesis is: given a target unitary
$U\in \mathbb{U}(2^n)$ and a finite gate set $\GC$, produce a circuit $V$ over $\GC$ that
approximates $U$ within a prescribed accuracy $\varepsilon>0$. The approximation error $\varepsilon$ is usually quantified using the operator norm. For a bounded linear operator $A$, this norm is defined as
\begin{equation}\label{eq:operator_norm}
\|A\| := \sup_{\|\ket{\psi}\|_2 = 1} \|A\ket{\psi}\|_2
= \sqrt{\lambda_{\max}(A^\dagger A)}\,,
\end{equation}
where $\lambda_{\max}(\cdot)$ denotes the largest eigenvalue, and this norm induces a distance $\|U-V\|$ between unitaries. Importantly, this metric has an operational
interpretation, in the sense that if $U_\varepsilon$ satisfies $\|U-U_\varepsilon\|\le \varepsilon$, then
for any input state $\ket{\psi}$ and any positive-valued operator measure (POVM), the corresponding outcome probabilities obey
\begin{equation}\label{eq:op_meaining}
    \left| P_U-P_{U_\varepsilon}\right|\leq 2\varepsilon\,,
\end{equation}
where  $P_U$ and $P_{U_\varepsilon}$ are the probabilities of an arbitrary measurement outcome $M$ in the POVM. 
An important subtlety is that the distance $\|U-V\|$ is sensitive to global phases, meaning that while acting on a quantum state $U$ and $e^{i\phi}U$ generate an unimportant global phase, they may be far in
operator norm and thus sill be distinguished. Hence, in some settings it is convenient to work in the adjoint (superoperator)
representation, and use the phase-insensitive distance
\begin{equation}\label{eq:adjoint_distance}
    d(U,V) := \big\| U\otimes U^* - V\otimes V^* \big\|\,.
\end{equation}
One can verify that $\|U-V\|\le\varepsilon$
implies $d(U,V)\le 2\varepsilon$.

A well-known fundamental result is that the set $\{H,\,T,\,CNOT\}$~\footnote{Note that the set $\{H,\,T,\,CNOT\}$ contains the $H$ and $T$ gates acting on arbitrary qubits,  and the $CNOT$ gate acting on arbitrary pairs of qubits.}, or Clifford$+T$, is universal~\cite{boykin2000new}.  
More precisely, the group generated by these gates is dense in the projective unitary group on $n$ qubits, meaning that for every $U\in \mathbb{U}(2^n)$ and every $\varepsilon >0$ there exists a circuit over this set approximating $U$ within error at most $\varepsilon$ in operator norm, up to an overall global phase.

The key idea to prove this universality result is to show that $iH$ and $e^{-i\frac{\pi}{8}}T$ are dense in the single-qubit special unitary group $\mathbb{SU}(2)$, i.e., they can be used to approximate every $U\in\mathbb{SU}(2)$ to within precision $\varepsilon$, for any $\varepsilon>0$. A completely analogous, less standard, result is that if we substitute the $H$ gate with the $W$ gate in the previous set, we also obtain a universal gate set. The proof of this result immediately follows by noting that the relation between $H$ and $W$ is just a relabeling of the axes $X\leftrightarrow Y$, but we include this result as a technical lemma below  as it will be instrumental for our purposes (see  Appendix~\ref{ap:su2_dense} for a proof).  

\begin{lemma}\label{lem:W,T}
	The set $\{iW,\,e^{-i\frac{\pi}{8}}T\}$ is dense in $\mathbb{SU}(2)$.
\end{lemma}

Universality for $n$-qubit unitaries is achieved by combining density in $\SU(2)$ with the result in Ref.~\cite{barenco1995elementary} which shows that arbitrary single-qubit rotations plus the $CNOT$ gate can approximate any unitary in $\mathbb{U}(2^n)$ up to a phase. The previous ensure that universal quantum computation can be realized using only this reduced number of building blocks. 
Importantly, while no single error-correcting code can implement every gate in $\{H,\,T,\,CNOT\}$ transversally~\cite{eastin2009restrictions}, fault-tolerant constructions for all of them are well established~\cite{gottesman1999demonstrating,zhou2000methodology,knill2004fault,bravyi2005universal}.

Next, we note that universality ensures that approximations exist, but it does not prescribe the synthesis nor quantify the number of gates in the compilation. A second fundamental result that provides such a guarantee is the celebrated Solovay--Kitaev (SK) theorem~\cite{kitaev1997quantum,kitaev2002classical,dawson2005solovay,ross2014optimal,kuperberg2023breaking},
which characterizes the computational complexity and convergence rate of gate sequences from a dense generating set $\GC$.

\begin{lemma}[Solovay--Kitaev theorem]\label{lem:Solovay-Kitaev}
    Let $\mathcal G\subset \mathbb{SU}(2^n)$ be a finite  set of unitary gates generating a dense subgroup of $\mathbb{SU}(2^n)$. Then there exists a constant $0 <c< 4$ such that for any target gate $U\in \mathbb{SU}(2^n)$, there exists a classical algorithm that in time $\OC\left(\log^c\left(\varepsilon^{-1}\right)\right)$ outputs a sequence of gates from $\GC$ of length $\mathcal O\left(\log^c\left(\varepsilon^{-1}\right)\right)$ that approximates $U$ to precision $\varepsilon> 0$ in operator norm. 
\end{lemma}
\noindent Crucially, the theorem is constructive, in the sense that it provides a classical algorithm for generating those approximations.

 The proof of the SK theorem is more involved than that of Lemma~\ref{lem:W,T}, and an excellent account of it can be found in Ref.~\cite{dawson2005solovay}.  Importantly, we note that while Lemma~\ref{lem:Solovay-Kitaev} refers to the scaling with $\varepsilon$, it does not explicitly show the scaling with the number of qubits $n$. Indeed, the explicit scaling of SK in terms of $n$  is exponential in general, as $\Omega\left(2^n \log\left(\varepsilon^{-1}\right)/\log(n)\right)$ and $\OC\left(n^2 4^n \log\left(n^2 4^n \varepsilon^{-1}\right)\right)$ gate operations are needed to approximate arbitrary unitaries~\cite{nielsen2000quantum}.  
 Essentially, the SK theorem only ensures that the cost in terms of the number of gates from a universal set needed to approximate within $\varepsilon$ precision scales very favorably for small, fixed $n$, in particular as $\OC\left(\log^c\left(\varepsilon^{-1}\right)\right)$. While one may wonder if this scaling can be reduced, a volume argument yields a lower bound on the sequence length of $\Omega\left(\log\left(\varepsilon^{-1}\right)\right)$~\cite{harrow2002efficient}, indicating that the SK theorem matches the bound up to poly-logarithmic factors in $\varepsilon^{-1}$. 
 
 It should be stressed here that the SK theorem has been shown to hold for any connected semisimple Lie group and not just $\SU(2^n)$~\cite{kuperberg2023breaking}. Furthermore, sequences of $\OC\left(\log\left(\varepsilon^{-1}\right)\right)$ length are known to exist for all universal gate sets with algebraic entries in compact connected Lie groups~\cite{bourgain2008spectral,bourgain2010spectral,bourgain2012spectral,benoist2016spectral}. This scaling is  within a constant factor of the optimal result, but an algorithm to efficiently find such sequences is unavailable~\cite{harrow2002efficient}. 

An important caveat to highlight is the fact that the generating set in the SK theorem  usually needs to contain inverse gates for $c$ to be constant. If one lifts this requirement, the value $c$ found in~\cite{bouland2021efficient} depends on $2^n$, i.e., we have $c=c(2^n)$. However, for all intents and purposes implementing inverses does not require additional effort, meaning that one typically considers $c$ to be constant and such that $0 <c< 4$. Indeed, the original SK theorem and the version reported in Ref.~\cite{dawson2005solovay} provide values $c=3+\delta$ (for any $\delta>0$) and $c\approx3.97$, respectively. The state of the art is $c=\log_2\phi+\delta=1.440\ldots+\delta$ for any $\delta>0$, where $\phi$ is the golden ratio~\cite{kuperberg2023breaking}.

Despite its tremendous importance, the SK theorem is not used for state-of-the-art compiling. Already for $\SU(2)$, it yields gate counts asymptotically larger than optimal and provides no guarantees on the minimal $T$-count. However, by exploiting the number-theoretic structure of specific gate sets, substantially sharper results are known. For example, in Ref.~\cite{kliuchnikov2013asymptotically}, an algorithm was presented to implement arbitrary single-qubit unitaries using $\mathcal{O}(\log(1/\varepsilon))$ $T$ gates and a constant number of ancilla qubits. Subsequently, Ref.~\cite{ross2014optimal} introduced a synthesis algorithm for $z$-rotations with the same asymptotic scaling but without ancilla qubits~\footnote{Interestingly, this algorithm efficiently provides the optimal sequence, given access to a factoring oracle.}. More recently, a family of ancilla-free, number-theoretic algorithms was developed ~\cite{kliuchnikov2023shorter} to approximate arbitrary single-qubit unitaries over discrete gate sets such as Clifford+$T$ or Clifford+$\sqrt{T}$, with near-optimal heuristic $T$-count. Even more recently, Ref.~\cite{morisaki2025optimal} proposed a deterministic Clifford$+T$ synthesis algorithm that is provably optimal with respect to the $T$-count, requiring at most $3 \log_2(1/\varepsilon)$ $T$ gates for most unitaries. Furthermore, probabilistic techniques where a target unitary is approximated by a mixed unitary channel have also been proposed~\cite{campbell2017shorter}, which can reduce the $T$-count by (up to) a factor of two while maintaining the same approximation error.  

Beyond the single-qubit case, several techniques exist to manage non-Clifford resources in $n$-qubit circuits~\cite{dimatteo2016parallelizing}. A central theme is the optimization of space-time trade-offs, since one can often exchange $T$-depth for ancillas, or conversely, reduce qubit overhead at the expense of a larger $T$-depth~\cite{low2024trading}.  Asymptotic bounds have been established for the $T$-count in  general $\mathbb U(2^n)$ unitaries~\cite{tan2025unitary} (and also for multi-qubit state preparation and diagonal-unitary synthesis~\cite{gosset2024quantum}), although the lower bound $\Omega(2^n)$  has not been reached.

\subsection{Exact unitary synthesis}\label{sec:pre_exact_synthesis}

Next, we review the exact unitary synthesis problem, which asks: \textit{Given a target unitary $U\in\mathbb{U}(2^n)$, does there exist a sequence of gates $V=V_1V_2\cdots V_d$ from a generating set $\GC$ (i.e., $V_i \in \GC$ $\forall i$) such that $U=V$}?

Let us first notice that all matrices in the Clifford$+T$ gate set have entries in the ring $\mathbb{Z}\left[\frac{1}{\sqrt{2}},i\right]$, defined as the set of numbers of the form
\begin{equation}
    \left\{\frac{a + bi + \frac{c}{\sqrt{2}} +  \frac{di}{\sqrt{2}}}{\sqrt{2}^k}\quad \Big| \quad a,b,c,d\in\mathbb{Z},\;k\in\mathbb{N} \right\}\,,
\end{equation}
equipped with the usual addition and multiplication of complex numbers (see Eqs.~\eqref{eq:Hadamard},~\eqref{eq:S_gate},~\eqref{eq:T_gate} and~\eqref{eq:CNOT}, and recall that $e^{i\pi/4}=\frac{1+i}{\sqrt{2}}$). Hence, it is obvious that any unitary which can be exactly synthesized using Clifford and $T$ gates must have entries in this ring, as matrix multiplication only employs addition and multiplication of complex numbers.

The converse implication is far less obvious. It was first established in Ref.~\cite{kliuchnikov2013fast} for single-qubit unitaries, then extended to multi-qubit unitaries~\cite{giles2013exact}, and further refined for certain restricted Clifford$+T$ circuits~\cite{amy2020number}. These results imply that all unitary matrices with entries in the ring $\mathbb{Z}\left[\frac{1}{\sqrt{2}},i\right]$ can be exactly synthesized using Clifford and $T$ gates (using at most one ancilla, which is also necessary). In other words, the group of unitaries with entries in the ring $\mathbb{Z}\left[\frac{1}{\sqrt{2}},i\right]$ and the group generated by the Clifford$+T$ gate set with ancilla qubits are identical.

These results yield explicit algorithms to exactly synthesize arbitrary unitaries with entries in $\mathbb{Z}\left[\frac{1}{\sqrt{2}},i\right]$. For instance, Ref.~\cite{kliuchnikov2013synthesis} introduced an exact-synthesis algorithm whose gate count scales exponentially with $n$. Moreover, the same framework can be used for approximate synthesis: one first approximates a target unitary by another unitary with entries in $\mathbb{Z}\left[\frac{1}{\sqrt{2}},i\right]$, and then resorts to exact-synthesis~\cite{kliuchnikov2013synthesis}. 

It is worth contrasting this algebraic viewpoint with the SK theorem (Lemma~\ref{lem:Solovay-Kitaev}). The latter guarantees efficient approximation using any finite universal gate set, but it does not address the exact representability
question, which depends on whether the target unitary lies in the discrete subgroup generated by
$\GC$.

Within exact synthesis one also seeks optimal implementations with respect to fault-tolerant cost metrics such as the $T$-count or $T$-depth. In the single-qubit case this
optimization is essentially solved, since it has been shown that every exactly implementable one-qubit Clifford+$T$ operator admits a unique canonical decomposition~\cite{matsumoto2008representation,giles2013remarks}, which is also optimal in $T$-count among all exact Clifford+$T$ decompositions. As a consequence, exact synthesis and $T$-count minimization for single-qubit Clifford+$T$ unitaries can be carried out efficiently using this normal form. For multi-qubit unitaries, $T$-optimality is substantially more challenging~\cite{gosset2013algorithm,gheorghiu2022t}.

\subsection{The matchgate group}\label{sec:pre_matchgates}

Let us quickly recall how matchgates are defined. The matchgate group consists of all unitaries generated by $R^z(\theta)=e^{i\theta Z/2}$ rotations on arbitrary qubits, together with $R^{xx}(\theta)=e^{i\theta X\otimes X/2}$ rotations acting on nearest-neighbors in an open one-dimensional array of qubits.
It is a well-studied subgroup of the special unitary group, mainly due to its connection with free-fermionic systems~\cite{valiant2001quantum,knill2001fermionic,terhal2002classical,divincenzo2005fermionic,bravyi2004lagrangian,jozsa2008matchgates,jozsa2010matchgate,ramelow2010matchgate,brod2011extending,brod2014computational,brod2016efficient,oszmaniec2017universal,zhao2021fermionic,wan2022matchgate,helsen2022matchgate,diaz2023showcasing,mele2024efficient,stroeks2024solving}. In particular, it is well known that any dynamics generated by free-fermionic Hamiltonians can be exactly realized as a matchgate circuit via the Jordan-Wigner transformation, and vice-versa~\cite{knill2001fermionic,terhal2002classical,divincenzo2005fermionic}. This fact renders matchgates ubiquitous in quantum simulation primitives~\cite{kivlichan2018quantum,arute2020hartree,arrazola2022universal,verstraete2009quantum,kraus2011compressed,cervera2018exact,jiang2018quantum,dallaire2019low,sopena2022algebraic,ruiz2024bethe,ruiz2024efficient,kokcu2022fixed,kokcu2022algebraic}.

From an abstract point of view, matchgate unitaries are a representation of the Lie group $\mathbb{SPIN}(2n)$, which is the double cover of $\SO(2n)$~\cite{guaita2024representation}. Indeed, the adjoint action of a matchgate circuit $U$ on the Majorana operators can be described via the linear map~\cite{jozsa2008matchgates}
\begin{equation}\label{eq:SO-isomorphism}
    U c_l U^\dagger = \sum_m Q_{lm}\,c_m\,,
\end{equation}
where $Q_{lm}$ are the entries of a matrix $Q\in\SO(2n)$, and $c_l$ are the Majorana operators. Under the Jordan--Wigner transformation, these take the form
\begin{align}\nonumber
    c_1&=XI\cdots I,\; c_3= ZXI\cdots I, \;\dots,\; c_{2n-1}=Z\cdots Z X\,, \nonumber\\
        c_2&=YI\cdots I,\; c_4= ZYI\cdots I, \;\dots, \;\; c_{2n}\;\;\;=Z\cdots Z Y\,, \nonumber
\end{align}
and we recall that they satisfy the anti-commutation relations $\{c_l, c_m\} = 2\delta_{lm}$. Furthermore, the Lie algebra of matchgate unitaries is the real vector space spanned by products of two distinct Majorana operators.

While the groups $\mathbb{SPIN}(2n)$ and $\SO(2n)$ are not isomorphic, their Lie algebras are, and an explicit isomorphism between the two is given by the linear map
\begin{equation}\label{eq:algebra-isomorphism}
    \varphi(c_j c_k) = 2L_{jk}\,,
\end{equation}
where the matrices $L_{jk}$ with entries ${(L_{jk})}_{l,m}=\delta_{jl}\delta_{km}-\delta_{jm}\delta_{kl}$ are a basis for the vector space of anti-symmetric matrices, and hence for the Lie algebra $\so(2n)$. The factor $2$ in Eq.~\eqref{eq:algebra-isomorphism} is key to understanding the double-cover property: consider $Q=e^{\theta L_{jk}}\in\SO(2n)$, and $U=e^{\theta \varphi^{-1}(L_{jk})}=e^{\theta c_j c_k/2}$. It is clear that $Q':= e^{(\theta+2\pi) L_{jk}}=Q$ but $U':= e^{(\theta+2\pi) c_j c_k/2}=-U$. Hence, under the map that sends $U=e^{iH}$  to $Q=e^{\varphi^{-1}(iH)}$ (with $H$ a real linear combination of products of two distinct Majoranas), both $U$ and $-U$ are mapped to the same $Q$, which implies that the groups are not isomorphic. However, $\varphi$ induces an isomorphism $\Phi$ at the group level  between the adjoint representation of $\mathbb{SPIN}(2n)$ and $\SO(2n)$, 
since both $U$ and $-U$ have the same action under conjugation. That is, if $U=e^{\sum \alpha_{jk} c_j c_k}\in\mathbb{SPIN}(2n)$ with $\alpha_{jk}\in\mathbb{R}$, then ${\rm Ad}_U(\cdot)=U(\cdot)U^\dagger$ is isomorphic to $Q= e^{\sum 2 \alpha_{jk} L_{jk}}\in\SO(2n)$ as per Eq.~\eqref{eq:SO-isomorphism} (see e.g., Appendix A in~\cite{braccia2025optimal}) via the map:
\begin{equation}\label{eq:group_isomorphism}
   \begin{matrix}  \Phi: \quad \qquad &
     {\rm Ad}_{\mathbb{SPIN}(2n)} & \;\;\longrightarrow\;\; & \!\SO(2n)  \\ \\ &\!\!\!\!\!\!\!\! e^{\sum \alpha_{jk} c_j c_k} \otimes  e^{-\sum \alpha_{jk} c_j c_k} &\;\;\longmapsto\;\;& e^{\sum 2 \alpha_{jk} L_{jk}}\,,
\end{matrix}
\end{equation}
where we used the superoperator form $U\otimes U^*$ of ${\rm Ad}_U(\cdot)$.

From a computational point of view, matchgate circuits are simulable in polynomial time (in the number of qubits) on a classical computer for certain input states and measurements. Specifically, for computational-basis input states, the expectation value of computational-basis measurements can be computed classically to $B$ bits of precision in time $\OC(\poly(n,B))$~\cite{valiant2001quantum,terhal2002classical}. Moreover, arbitrary input product states can be efficiently simulated to $B$ bits of precision if the measurement is restricted to a single-qubit $Z$ measurement~\cite{jozsa2008matchgates}. Notice that the precision of these simulations improves exponentially with the number of classical bits $B$.

It should nonetheless be stressed that for certain easy-to-prepare input states, matchgate circuits output probability distributions that are provably hard to sample from classically~\cite{oszmaniec2022fermion}. Moreover, matchgates become universal for quantum computation when supplemented with SWAP gates (or equivalently, when they act on qubits whose connectivity graph is not a path or a cycle)~\cite{jozsa2008matchgates,brod2014computational}.

Finally, we discuss the matchgate-Clifford group--the subgroup of matchgate circuits that are also Clifford--. This finite group is isomorphic to the group of signed permutation matrices with unit determinant, $\SO(2n)\cap B_{2n}$ (where $B_{2n}$ is the hyperoctahedral group). This easily follows from the fact that Clifford unitaries map Pauli operators to Pauli operators (up to a $\pm1,\,\pm i$ phase). Hence, for this to hold in Eq.~\eqref{eq:SO-isomorphism}, $Q$ must be a signed permutation matrix belonging to $\SO(2n)$, as Majorana operators are Paulis under the Jordan--Wigner transformation. Importantly, it turns out that in general, the matrix $Q\in\SO(2n)$ in Eq.~\eqref{eq:SO-isomorphism} is the Pauli transfer matrix of the matchgate circuit $U$, restricted to the subspace of the Majorana operators.

\section{Results}\label{sec:results}
We now present our results on matchgate synthesis. First, we discuss our findings regarding approximate synthesis:
\begin{itemize}
	\item In Sec.~\ref{sec:matchgates_universal} we show that the generators of the matchgate-Clifford group plus the $\overline T$ gate (the usual $T$ gate up to a global phase) form a discrete universal matchgate set. 
    \item In Sec.~\ref{sec:residual_entanglement} we study the spurious entanglement introduced when synthesizing single-qubit $R^z(\theta)$ rotations using two-qubit matchgate circuits. We quantify this effect via the operator entanglement $E(U_\varepsilon)$~\cite{zanardi2000entangling,zanardi2001entanglement} of a two-qubit approximation $U_\varepsilon$ to $R^z(\theta)$, and prove that $E(U_\varepsilon)\in\mathcal{O}(\varepsilon^2)$ for the operator-norm error $\|U_\varepsilon-R^z(\theta)\|\le \varepsilon$.
    
    \item In Sec.~\ref{sec:SO_error}, using the isomorphism $\Phi$ from Eq.~\eqref{eq:group_isomorphism}, we obtain the Pauli transfer matrices in $\SO(2n)$ associated to the matchgate-Clifford$+\overline T$ gate set. We prove that compiling in the standard representation of $\SO(2n)$ with $\frac{2\varepsilon}{\pi n}$ error translates to at most error $\varepsilon$ in the unitary approximation, while allowing for an exponential reduction in the size of the compiled matrix. We conclude by discussing several implications of these results.
\end{itemize}

    \noindent Then, we focus on exact synthesis:

    \begin{itemize}

	\item  In Sec.~\ref{sec:matchgate_exact} we introduce the exact matchgate synthesis problem, and prove that all unitaries $U\in\mathbb{SPIN}(2n)$ such that $U\otimes U^*$ has entries in the ring $\mathbb Z\left[\frac{1}{\sqrt 2},i\right]$ can be exactly synthesized by a sequence of gates from the matchgate-Clifford$+\overline T$ set  without ancillas. 
    
    \item In Sec.~\ref{sec:T_bounds} we provide upper bounds on the number of $\overline T$ gates and Clifford gates required to exactly synthesize matchgate circuits with the matchgate-Clifford$+\overline T$ set.

    \item Finally, in Sec.~\ref{sec:SAT} we use the previous insights to map the exact matchgate synthesis problem to Boolean satisfiability (concretely, SAT and MAX-SAT). Subsequently, in Sec.~\ref{sec:xx_exact_compilation} we obtain quantum circuits that diagonalize the free-fermionic XX model on $n=4$ and $n=8$ qubits~\cite{verstraete2009quantum} with optimal (for $n=4$) and near-optimal (for $n=8$) depths.

\end{itemize}

\subsection{Approximate matchgate synthesis}
Here, we explore different aspects of the approximate compilation of matchgate circuits. 

\begin{figure}[t]
    \centering
    \includegraphics[width=1\linewidth]{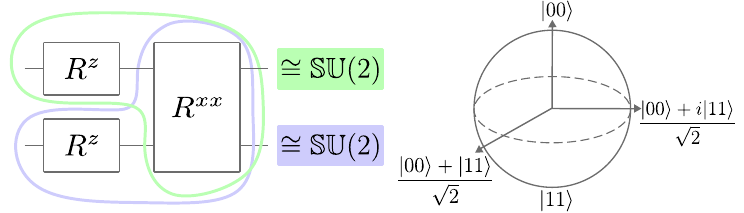}
    \caption{{\bf $\SU(2)$ representations within the two-qubit matchgate group.} We illustrate two $\SU(2)$ representations, each generated by a single-qubit $R^z$ and a two-qubit $R^{xx}$ rotation, as indicated by the colors. When restricted to the fermionic even-parity subspace spanned by $|00\rangle$ and $|11\rangle$, these representations give rise to a Bloch-sphere structure.
    }\label{fig:su2_is_su2}
\end{figure}

\subsubsection{A discrete universal gate set for matchgates}\label{sec:matchgates_universal}
We begin by identifying a discrete universal matchgate set. To do so, let us consider the two-qubit matchgate group, generated by local $Z$ rotations and $XX$ rotations between neighboring qubits on a line. Algebraically,  a single $iZ$ operator (say, on the first qubit) together with $iXX$  and the usual matrix commutator, generate the Lie algebra
\begin{equation}\label{eq:algebra}
    \g = {\rm span}_\mathbb{R}\{ iX\otimes X,\,iY\otimes X,\,iZ\otimes \id\}\,.
\end{equation}
The key realization for constructing a discrete universal matchgate set is that this real vector space is a representation of $\mathfrak{su}(2)$~\cite{kokcu2022algebraic}. Indeed, there exists an isomorphism $\phi: \su(2)\to \g$ defined by
\begin{equation}
\begin{split}\label{eq:su2_isom}
	 &\phi(iX)=iX\otimes X\,,\;\; \phi(iY)=iY\otimes X \,,\; \; \\ &\phi(iZ)=iZ\otimes\id\,,
\end{split}
\end{equation}
which preserves the commutation relations, and thus the structure constants (that is, $\phi([A,B])=[\phi(A),\phi(B)]$ for all $A,B\in \su(2)$). 

The previous implies that rotations $R^z_q(\theta)$ on qubit $q$, obtained via exponentiation of $iZ_q$, plus rotations $R_{q,q+1}^{xx}(\theta)$ on qubits $q,q+1$ (obtained from $iX_q X_{q+1}$), generate a representation of $\SU(2)$~\footnote{This holds because $\SU(2)$ is a simply connected Lie group, and hence the representations of the Lie algebra and Lie group are in one-to-one correspondence~\cite{hall2013lie}.}. Similarly, we can also identify an analogous representation generated by $R^z_{q+1}(\theta)$ together with $R_{q,q+1}^{xx}(\theta)$, as illustrated in Fig.~\ref{fig:su2_is_su2}. In addition, these two $\SU(2)$ representations admit a Bloch-sphere structure when restricted to act on the fermionic even-parity subspace spanned by $\{|00\rangle, |11\rangle\}$ (see Appendix~\ref{ap:bloch_sphere} for further details).

Once we have identified these $\SU(2)$ representations within the matchgate group, we can leverage the fact that the gates $e^{-i\frac{\pi}{8}}T$ and $iW$ are dense in $\SU(2)$ (Lemma~\ref{lem:W,T}). Indeed, we just need to find their representation as matchgates. The $e^{-i\frac{\pi}{8}}T$ gate is represented by itself, as
\begin{equation}
    \overline{T}=e^{\phi\left(i \frac{\pi}{8}Z \right)} = R^z\left(\frac{\pi}{4}\right)\,.
\end{equation}
For the $iW=e^{i\frac{\pi}{2}\frac{Y+Z}{\sqrt{2}}}$ gate, we can identify its matchgate representation with
\begin{align}
   \overline W &= e^{\phi\left(i\frac{\pi}{2}\frac{Y + Z}{\sqrt{2}}\right)}  = e^{i\frac{\pi}{2}\frac{Y\otimes X + Z}{\sqrt{2}}}  \nonumber  \\ &=i\frac{\left(Y\otimes X + Z\right)}{\sqrt{2}} = i\frac{\left(\id + iX\otimes X\right)}{\sqrt{2}}Z  \nonumber  \\   
   &=R^{xx}\left(\frac{\pi}{2}\right)R^z\left(
   \pi\right) \nonumber \\ &=R^{xx}\left(\frac{\pi}{2}\right)\overline {S}^2\,,
\end{align}
where we have used Eq.~\eqref{eq:su2_isom}. Notice that both 
\begin{equation}
    \overline S = R^z\left(\frac{\pi}{2}\right)\,,
\end{equation} 
and
\begin{equation}
	R^{xx}\left(\frac{\pi}{2}\right) = e^{i \frac
{\pi}{4} X\otimes X}  = \frac{1}{\sqrt{2}} \begin{pmatrix}
		1 & 0 & 0 & i \\
		0 & 1 & i & 0 \\
		0 & i & 1 & 0 \\
		i & 0 & 0 & 1 
	\end{pmatrix} \,,
\end{equation}
belong to the Clifford group, and thus $\overline{W}$ is also Clifford. It readily follows that $\overline S=R^{z}\left(\frac{\pi}{2}\right) $ and $R^{xx}\left(\frac{\pi}{2}\right)$ generate the entire matchgate--Clifford group, as all Clifford rotation angles $\left\{\frac{\pi}{2}, \pi, \frac{3\pi}{2}\right\}$ can be synthesized with just these gates. 
\medskip

We can now present our first main result.
\begin{theorem}\label{th:universal_matchgates}
    The gate set $\GC={\left\{\overline  T_q,\,\overline S_q\right\}}_{q=1}^n \cup {\left\{R_{q,q+1}^{xx}\left(\frac{\pi}{2}\right)\right\}}_{q=1}^{n-1}$  is universal for matchgate circuits. 
\end{theorem}

This theorem ensures that the matchgate-Clifford$+\overline T$ set is dense in the  matchgate group (see Appendix~\ref{ap:th_1} for a proof). Thus, the situation for matchgates is completely analogous to the general unitary group, where the Clifford$+T$ set is universal. Importantly, there exist fault-tolerant implementations
of all the gates in $\GC$. For instance, $R^{xx}(\frac{\pi}{2})$ can be exactly decomposed in terms of $CNOT$s, $H$ and $S$ gates as shown in Ref.~\cite{ramelow2010matchgate}.

An immediate corollary of Theorem~\ref{th:universal_matchgates} can be readily obtained from the application of the SK theorem to $\SU(2)$ (Lemma~\ref{lem:Solovay-Kitaev}). Indeed, a nice feature of the isomorphism(s) depicted in Fig.~\ref{fig:su2_is_su2} is that we can import all the machinery and understanding developed for the standard representation of $\SU(2)$ (the paradigmatic single-qubit unitary case) to synthesize matchgate circuits using our universal matchgate set $\GC$. One just needs to find a decomposition of the target unitary in terms of $R^z(\theta)$ and $R^{xx}(\theta)$ rotations, which can be done in polynomial time using, e.g., Hurwitz decomposition~\cite{diaconis2017hurwitz}. Then, the corresponding isomorphism is applied and each $R^z(\theta)$ and $R^{xx}(\theta)$ is synthesized using algorithms developed for  $\SU(2)$, such as \texttt{gridsynth}~\cite{ross2014optimal, gridsynth,pygridsynth}. The subadditivity of the operator norm (i.e., the triangle inequality) guarantees that the total error in the circuit accumulates at most linearly with the number of compiled rotations, ensuring the entire procedure remains efficient. We state this observation formally in the following corollary.

\begin{corollary}\label{co:SK_matchgate}
	Any matchgate circuit consisting of a sequence of $m$ $R^z$ and $R^{xx}$ rotations can be approximated to $\varepsilon$ precision in the operator norm using at most $\OC\left(m \log(\frac{m}{\varepsilon})\right)$ gates from the gate set $\GC$. Such an approximation can be found on a classical computer in time $\OC\left(m \;{\rm poly}\log(\frac{m}{\varepsilon})\right)$.
\end{corollary}

Corollary~\ref{co:SK_matchgate} can be considered a version of the SK theorem for matchgate circuits, as it guarantees that the entire matchgate group can be efficiently approximately compiled using only matchgate-Clifford gates and $\overline T$ gates~\footnote{Let us here recall that any matchgate unitary can be decomposed into at most $n(2n-1)$ $R^z$ and $R^{xx}$ rotations, which is the dimension of the Lie algebra $\so(2n)$~\cite{braccia2025optimal}. This means  that $m\leq n(2n-1)$ in Corollary~\ref{co:SK_matchgate}.}. 
It should nevertheless be stressed that even though each local $R^z$ or $R^{xx}$ gate can be synthesized near-optimally using \texttt{gridsynth}, this does not imply that the compilation of a global matchgate unitary composed of these gates is close to optimal. Instead, this strategy provides an initial compilation with a bounded number of gates which may be further optimized, e.g., using reinforcement learning techniques~\cite{ruiz2025quantum, valcarce2025unitary,bosco2026quantum}.

\subsubsection{Single-qubit $R^z(\theta)$ synthesis with matchgates}
\label{sec:residual_entanglement}
We find it important to remark that the only available single-qubit operations  in the set $\GC$  are the $\overline T$ and $\overline S$ gates, which satisfy $\overline T^2 = \overline S$ and $\overline T^8=\id$. These relations in turn imply that there are only 8 unique combinations of $\overline T$ and $\overline S$, consistent with the fact that $\overline T$ is a $\pi/4$ rotation. As a consequence, we need to use $R^{xx}\left(\frac{\pi}{2}\right)$ to synthesize an arbitrary $R^z(\theta)$ rotation, which forces us to introduce an extra qubit to produce the two-qubit matchgate $U_\varepsilon$ that approximates $R^z(\theta)\otimes \id$ up to $\varepsilon$ error. 

As such, our construction to compile single qubit gates requires the use of two-qubit gates, which in turn can introduce spurious entanglement in the system. We can quantify such correlations through
the operator entanglement \(E(U_\varepsilon)\)~\cite{zanardi2000entangling,zanardi2001entanglement}. This quantity is defined as the linear entropy of the bipartition between the input and output qubits in the vectorized state $\frac{1}{\sqrt{2^n}}|U_\varepsilon\rangle\!\rangle$. That is, given a unitary $U_\varepsilon$ acting on a Hilbert space $\HC$, we define   $\ket{U_\varepsilon}\!\rangle = \sum_{i,j} {(U_\varepsilon)}_{ij} \ket{i} \otimes \ket{j} \in \HC \otimes \HC$, then $E(U_\varepsilon)= 1-4^{-n}\,\Tr\left[{(\ket{U_\varepsilon}\!\rangle \langle\!\bra{U_\varepsilon})}^2\right]$, where the trace is taken over one of the two subsystems. A unitary with larger $E(U_\varepsilon)$ is more capable of generating entanglement when acting on initially separable states~\cite{zanardi2000entangling,zanardi2001entanglement}. In the following proposition, proven in Appendix~\ref{ap:entanglement}, we bound $E(U_\varepsilon)$ in terms of the synthesis error $\varepsilon$.

\begin{proposition}\label{prop:entanglement}
	Let $R^z(\theta)$  be a single-qubit rotation around the $Z$-axis with some $\theta\in [0, 2\pi)$, and let $U_\varepsilon$ be a two-qubit   matchgate circuit approximating  $R^z(\theta)$ such that
	\begin{equation}
		\|U_\varepsilon- R^z(\theta)\|\leq \varepsilon\,,
	\end{equation}
	where $\|\cdot \|$ denotes the operator norm. Then, the entanglement of $U_\varepsilon$, denoted by $ E(U_\varepsilon)$, is bounded by 
	\begin{equation}
		E(U_\varepsilon) \leq 2\varepsilon^2 + \OC(\varepsilon^4)\,.
	\end{equation}
\end{proposition}

\begin{figure}[t]
    \centering
    \includegraphics[width=1\linewidth]{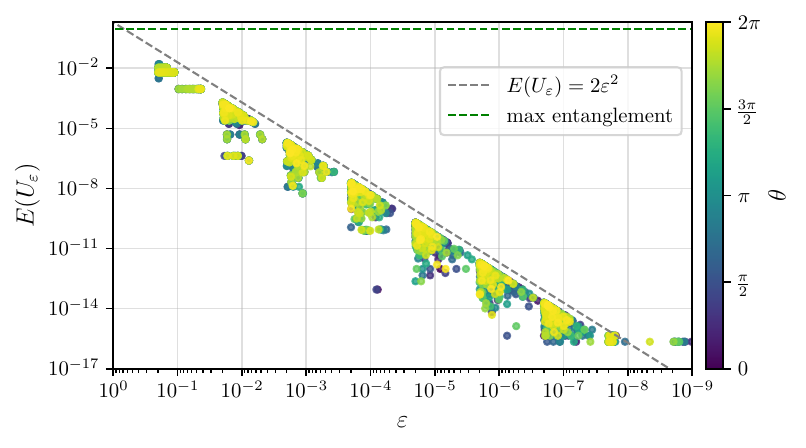}
    \caption{{\bf Residual entanglement in the synthesis of $R^z(\theta)$ with matchgates.} We plot the linear entropy $E(U_\varepsilon)$ of the approximated unitary $U_\varepsilon$ against the precision $\varepsilon$, for values of $\theta$ in the range $\theta\in[0,2\pi)$. The unitaries $U_\varepsilon$ were compiled using the discrete matchgate set $\GC$, via the \texttt{gridsynth} algorithm~\cite{ross2014optimal,gridsynth,pygridsynth} and the $\SU(2)$ isomorphisms depicted in Fig.~\ref{fig:su2_is_su2}. We show the bound on $E(U_\varepsilon)$ from Proposition~\ref{prop:entanglement} (grey dashed line), and the maximal entanglement achievable by any gate in $\SU(4)$ (green dashed line). Points above the bound have reached the numerical-precision limit.}\label{fig:Rz_compilation}
\end{figure}

This proposition demonstrates that the spurious entanglement introduced by the use of an extra qubit to compile an $R^z(\theta)$ gate is upper bounded by (two times) the square of the approximation error. Therefore, if we synthesize  $R^z(\theta)$ to within a small $\varepsilon$, the residual entanglement is also small (as it must be if one is to obtain a faithful approximation of the single-qubit rotation)~\footnote{In fact, notice that it is not uncommon that algorithms synthesizing $\SU(2)$ gates use ancillas~\cite{kliuchnikov2013asymptotically}. This means that  more two-qubit gates are used in practice, presumably exhibiting larger physical error rates than single-qubit gates. This however may not be a big issue in the fault-tolerant regime, where  the crucial resource to be optimized is often considered to be the number of $T$ gates~\cite{litinski2019game}.}. 
The bound in Proposition~\ref{prop:entanglement} is quite favorable, quadratic in $\varepsilon$, and we observe that it is tight in practice. To see this, in Fig.~\ref{fig:Rz_compilation} we plot the actual entanglement as a function of $\varepsilon$ for compilations of $R^z(\theta)$ in the entire range $\theta\in[0,2\pi)$, obtained using the $\texttt{gridsynth}$ algorithm, and compare it to the bound in Proposition~\ref{prop:entanglement}.  We find that even at target precision $\varepsilon \sim 10^{-8}$, which is a relatively modest accuracy scale for state-of-the-art $\SU(2)$ compilers~\cite{gridsynth}, the generated entanglement remains extremely small. In fact, around this target precision our simulations already hit a numerical-precision wall. Hence, we can conclude that in practically relevant regimes the introduced entanglement will essentially be negligible.

Importantly, we note that when compiling a single-qubit matchgate, there is no need to actually introducing dedicated ancilla qubits. Instead, one can use neighboring qubits
already present in the circuit. Since gate approximation errors (including those arising from residual spurious entanglement) accumulate at worst linearly throughout the circuit (by a triangle-inequality/subadditivity argument)~\cite{bernstein1997quantum}, then using neighboring qubits as ancillas does not significantly increase errors. In practice, however, the worst-case bound provided by the triangle inequality seems to be loose, as we show in Fig.~\ref{fig:local_errors} for Haar-random matchgate circuits sampled according to Ref.~\cite{braccia2025optimal}. There, we compare the global approximation error of the compiled circuit, $\varepsilon_{\rm glob}$, with the sum of per-gate synthesis errors (each $R^z$ and $R^{xx}$ rotation was approximated  using \texttt{gridsynth}  with target tolerance $\varepsilon=10^{-3}$).
Moreover, we compute the relative difference between the sum of local errors and the global error, which increases  with  system size and appears to be  largely independent of the local compiler's target precision $\varepsilon$.

We also briefly mention that virtual $R^z$ rotations could potentially be employed in a matchgate circuit architecture to bypass the need to synthesize $R^z$ rotations altogether~\cite{mckay2017efficient}. Indeed, let us recall that
\small
\begin{align}
	&e^{i\theta X_q X_{q+1}/2} e^{i\alpha Z_q /2} \nonumber \\ &= e^{i\alpha Z_q /2} e^{-i\alpha Z_q /2} e^{i\theta X_q X_{q+1}/2} e^{i\alpha Z_q /2} \nonumber \\ & = e^{i\alpha Z_q /2} e^{-i\alpha Z_q /2} \left(\cos\frac{\theta}{2} \id +i \sin\frac{\theta}{2} X_q X_{q+1}\right) e^{i\alpha Z_q /2} \nonumber \\ & =  e^{i\alpha Z_q /2} \left(\cos\frac{\theta}{2} \id +i \sin\frac{\theta}{2} \left(\cos\alpha X_q X_{q+1} + \sin\alpha Y_q X_{q+1}\right)\right) \nonumber \\& =e^{i\alpha Z_q /2} e^{i \theta \left(\cos\alpha X_q X_{q+1} + \sin\alpha Y_q X_{q+1}\right)/2 }\,, \nonumber
\end{align}
\normalsize
which indicates that an $R^z(\alpha)$ gate can be ``commuted through'' an $R^{xx}(\theta)$ gate towards the end of the circuit by rotating the $X$ axis of a qubit into $(\cos\alpha\, X +\sin\alpha \,Y)$. In some platforms, such as those based on superconducting qubits, virtual (non-error-corrected) $R^z$ rotations can be implemented ``in software'' by updating the control--frame phase~\cite{mckay2017efficient}. This requires zero physical pulse duration and hence introduces no error at all. In a fault-tolerant setting, this strategy removes $R^z$ gates by synthesizing $ e^{i \theta \left(\cos\alpha X_q X_{q+1} + \sin\alpha Y_q X_{q+1}\right)/2 }$ instead of $e^{i\theta X_q X_{q+1}/2} $, and we leave for future work to determine if such a procedure could be advantageous in practice.

Finally, another alternative to mitigate spurious entanglement is to merge $R^z$ and $R^{xx}$ gates, and compile directly in $\SO(4)$. For this, we can use the exceptional Lie algebra isomorphism $\so(4)\cong \su(2)\oplus \su(2)$. However, this approach requires an extended  set of matchgates for compilation, including gates like $\sqrt{T}$~\cite{gidney2019efficient} and other non-native gates. We discuss these aspects in detail in Appendix~\ref{ap:so(4)_isomorphism}.

\subsubsection{Error propagation from $\SO(2n)$ to $\mathbb{SPIN}(2n)$}\label{sec:SO_error}
Next, we obtain the image of the matchgate unitaries in the set $\GC$ from Theorem~\ref{th:universal_matchgates} under the group homomorphism $\Phi$ of Eq.~\eqref{eq:group_isomorphism}, corresponding to Pauli transfer matrices in $\SO(2n)$. This mapping allows us to synthesize matchgates in the standard representation of $\SO(2n)$, i.e., as $2n\times 2n$ matrices (rather than the original $2^n\times 2^n$ ones). In particular, the Pauli transfer matrix of a $\overline T$ gate on qubit $q$ is
\begin{align}\label{eq:T_SO}
    \widetilde{T}_{2q-1, 2q} =& \begin{pmatrix}
        \id_{2q-2} & & & \\  & \frac{1}{\sqrt{2}} &  \frac{1}{\sqrt{2}} & \\  &  - \frac{1}{\sqrt{2}} &  \frac{1}{\sqrt{2}} & \\ & & & \id_{2n -2q} 
    \end{pmatrix}\,,
\end{align} 
that of an $\overline S$ gate on qubit $q$ is
\begin{align}\label{eq:S_SO} \widetilde S_{2q-1, 2q}=& \begin{pmatrix}
        \id_{2q-2} & & & \\  & 0 & 1 & \\  &  - 1 & 0 & \\ & & & \id_{2n -2q} 
    \end{pmatrix}\,,
\end{align}
and the one corresponding to an $R^{xx}(\frac{\pi}{2})$ rotation on qubits $q,q+1$ is
\begin{align}\label{eq:Rxx_SO}
 \widetilde R_{2q, 2q+1} =&\begin{pmatrix}
        \id_{2q-1} & & & \\  & 0 & 1 & \\  &  - 1 & 0 & \\ & & & \id_{2n -2q+1} 
    \end{pmatrix}\,.
\end{align}
The indices in the matrices $\widetilde T\,, \, \widetilde S$, and  $\widetilde R$ label the planes in $\mathbb{R}^{2n}$ on which they act as $\SO(2n)$ rotations. 
Here, notice that $\widetilde S$ and $\widetilde T$ (single-qubit gates) act on odd-indexed subspaces in $\mathbb{R}^{2n}$, while $ \widetilde R$ (two-qubit gate) acts on even-indexed ones. 

\begin{figure}[t]
    \centering
    \includegraphics[width=1.\linewidth]{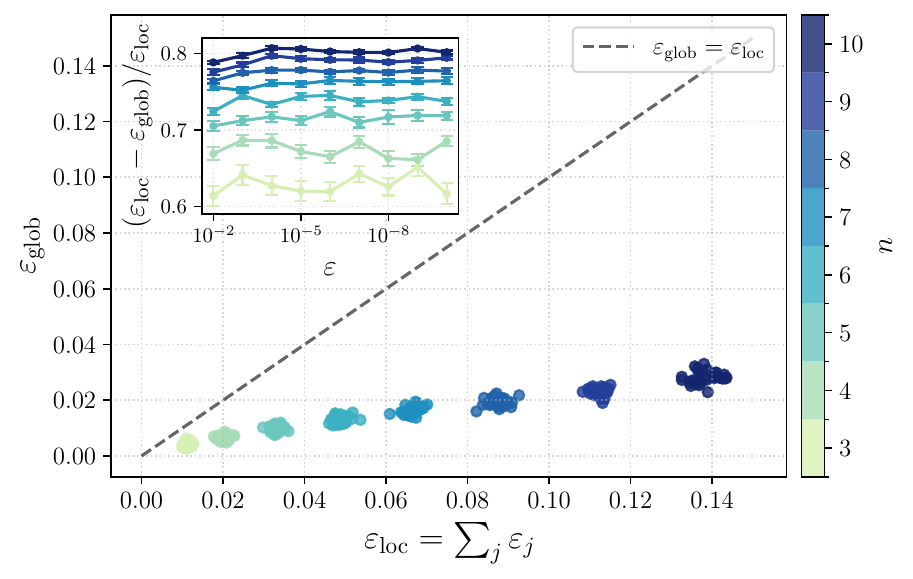}
    \caption{{\bf Accumulation of local synthesis errors in random matchgate circuits.} Global approximation error $\varepsilon_{\rm glob}$ (operator norm distance) as a function of the sum of local errors $\varepsilon_{\rm loc}$, for Haar-random matchgate circuits~\cite{braccia2025optimal} where each $R^z$ and $R^{xx}$ gate was approximated within $\varepsilon = 10^{-3}$ precision using \texttt{gridsynth}. Here, $\varepsilon$ denotes the compiler’s per-gate tolerance, while the local errors $\varepsilon_j$ in the sum $\varepsilon_{\rm loc} = \sum_j \varepsilon_j$ are the per-gate operator-norm distances. Different colors indicate different system sizes $n$, and the dashed line shows the worst-case subadditivity bound. In the inset, we plot the relative error $(\varepsilon_{\rm loc}-\varepsilon_{\rm glob})/\varepsilon_{\rm loc}$ as a function of $\varepsilon$. } 
    \label{fig:local_errors}
\end{figure}

Equations~\eqref{eq:S_SO} and~\eqref{eq:Rxx_SO} imply that the transfer matrices $\tilde S$, and $\tilde R$ are signed permutations. More concretely, they are nearest-neighbors signed transpositions, also called signed inversions. This aligns with the fact that Clifford gates map Paulis to Paulis, or, concomitantly, map Majorana operators to Majorana operators, up to a phase. It is thus clear that the entire matchgate-Clifford group is generated by $\overline{S}$ and $R^{xx}(\frac{\pi}{2})$, as  signed inversions generate the group of signed permutation matrices with unit determinant. 
In contrast, the transfer matrix $\widetilde T$ mixes two neighboring Majoranas in a uniform superposition (with a possible relative phase of $-1$). Crucially, this introduces a factor $\frac{1}{\sqrt{2}}$, and, as we will see, this factor plays a key role in the exact synthesis of matchgate circuits represented by $\SO(2n)$ matrices.

 Since the set $\GC$ is universal for the matchgate group $\mathbb{SPIN}(2n)$ as per Theorem~\ref{th:universal_matchgates}, and the homomorphism $\Phi$ to $\SO(2n)$ is surjective, it follows that the set of gates 
 \begin{equation}\label{eq:SO_gates}
    \widetilde \GC := {\left\{\widetilde T_{2q-1, 2q}, \tilde S_{2q-1, 2q}\right\}}_{q = 1}^n \cup{\left\{\widetilde R_{2q, 2q+1} \right\}}_{q = 1}^{n-1} \subset \SO(2n)\,
 \end{equation}
 is universal in $\SO(2n)$. This implies that  we can compile matchgates in the standard representation of $\SO(2n)$, where the size of the matrices is polynomial in $n$, leveraging an exponential reduction in matrix size compared to the original $2^n\times 2^n$ unitaries. 
 Let us recall that the SK theorem (Lemma~\ref{lem:Solovay-Kitaev}) holds for any connected semisimple Lie group, and not just $\SU(2^n)$, meaning that it also applies to $\SO(2n)$. Hence, one could employ the set $\widetilde\GC$ in Eq.~\eqref{eq:SO_gates} to efficiently approximate $\SO(2n)$ matrices (since the matrices to be compiled have now only a polynomial size). 
 In practice, such  global version of the SK-theorem for $\SO(2n)$ can be coded and implemented using the algorithm in Ref.~\cite{kuperberg2023breaking}, but we leave such an endeavor for future work. We simply remark here that this global synthesis will outperform local strategies based on $\SU(2)$ compilers when the number of qubits is sufficiently large.

\begin{figure}[t]
    \centering
    \includegraphics[width=\linewidth]{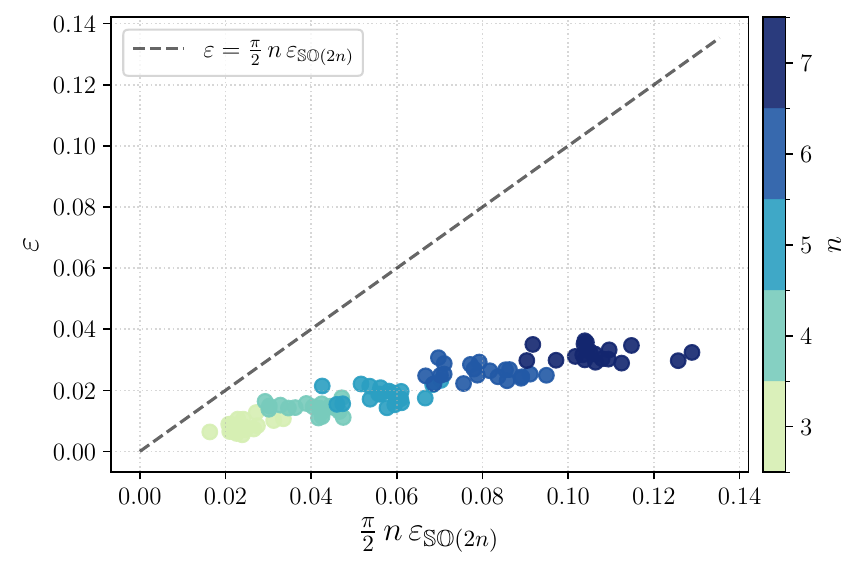}
    \caption{{\bf Error propagation from $\SO(2n)$ to $\mathbb{SPIN}(2n)$.}  We numerically test the bound in Theorem~\ref{th:SO-error} using Haar-random matchgate circuits~\cite{braccia2025optimal}.
    The horizontal axis shows the error incurred in the synthesis of $Q\in\SO(2n)$, scaled up by the number of qubits and a constant factor (i.e., $ \frac{\pi}{2}n \varepsilon_{\SO (2n)}$), while the vertical axis displays the actual error $\varepsilon = \|U\otimes U^*- U_\varepsilon\otimes U_\varepsilon^*\|$ in the adjoint representation of the matchgate circuit.
    Each color correspond to a different system size $n$, and the dashed line represents the theoretical bound provided by Theorem~\ref{th:SO-error}.}\label{fig:SO-errors}
\end{figure}

Crucially, a fundamental question that must be addressed when compiling matchgate circuits in the standard representation of $\SO(2n)$ is how the approximation error $\varepsilon_{\SO(2n)}$ incurred in that representation propagates to the unitary in $\SU(2^n)$.  In Appendix~\ref{ap:th_2}, we show that the error is amplified at most linearly in $n$, which enables efficient compilation schemes. This result is stated formally in the following theorem.

\begin{theorem}\label{th:SO-error}
    Let $Q_{\varepsilon}\in\SO(2n)$ be an approximation to $Q\in\SO(2n)$ such that $||Q-Q_\varepsilon|| = \varepsilon_{\SO(2n)} \ll 1$. Then, $||\Phi^{-1}(Q)-\Phi^{-1}(Q_\varepsilon)||\leq  \frac{\pi}{2}n\, \varepsilon_{\SO(2n)}$.
\end{theorem}

Theorem~\ref{th:SO-error} guarantees that to achieve an approximation error $\varepsilon$ in the matchgate circuit $U\otimes U^*$, an error $\varepsilon_{\SO(2n)}=\frac{2\varepsilon}{\pi n}$ in the synthesis of the corresponding $\SO(2n)$ matrix suffices. Since the  SK algorithm runs in $\OC\left(\log^c \varepsilon_{\SO(2n)}^{-1}\right)$ time, the $\OC(\log^c n)$ overhead is negligible compared to the polynomial versus exponential scaling with system size gained by working in $\SO(2n)$. In Fig.~\ref{fig:SO-errors}, we numerically test the tightness of the bound in Theorem~\ref{th:SO-error} for Haar-random matchgate circuits sampled following Ref.~\cite{braccia2025optimal}. We observe that the bound is loose, meaning that on average, the actual inaccuracy is even smaller than the theorem predicts.
Finally, we recall that a matrix in $\SO(2n)$ only determines a matchgate unitary $U$ up to a global sign. While physically unobservable, this sign becomes a relative phase in controlled operations, and therefore must be considered if $U$ is implemented as a controlled gate. We refer the reader to~\cite{goh2023lie} for an example on how to guarantee that the correct phase is compiled.

\subsection{Exact matchgate synthesis}\label{sec:exact_synthetis}

We now turn our attention to the problem of exact synthesis of matchgate unitaries.  That is: Given a matchgate unitary $U$, does there exist a circuit composed of gates from $\GC$ that exactly synthesizes it?

\subsubsection{Characterization of the matchgate-Clifford$+\overline{T}$ group}\label{sec:matchgate_exact}

To address this question, let us first note that every Pauli transfer matrix in $\widetilde \GC$ has entries in $\left\{ \pm \frac{1}{\sqrt 2}, \pm 1 \right\}$, as per Eqs.~\eqref{eq:T_SO}--\eqref{eq:Rxx_SO}. Consequently, any finite product of such matrices results in a matrix whose entries belong to the ring $\mathbb D\left[\sqrt 2\right]$, with
\begin{equation}
    \mathbb D\left[\sqrt 2\right] := \left\{a + b\sqrt{2} \;\, \Big|\;\, a,b\in \mathbb D \right\}\,.
    \label{eq:D_sqrt2_ring}
\end{equation}
Above, $\mathbb D$ denotes the dyadic ring 
\begin{equation}
    \mathbb D := \left\{\frac{m}{2^l} \;\Big|\; m\in \mathbb Z,\, l\in \mathbb N \right\}\,.
    \label{eq:dyadic_ring}
\end{equation} 
 While the converse statement is far from obvious, we prove that it holds. Namely, we prove that any matrix in $\SO (2n)$ with entries in the ring $\mathbb D\left[\sqrt 2\right]$ can be exactly synthesized by a finite sequence of gates from the set $\widetilde \GC$. Furthermore, every $U\in\mathbb{SPIN}(2n)$ such that $U\otimes U^*$  has entries in $\ZBB\left[\frac{1}{\sqrt{2}},i\right]$ is mapped under the isomorphism $\Phi$ in Eq.~\eqref{eq:group_isomorphism} to a matrix $Q\in\SO(2n)$ with entries in $\mathbb D\left[\sqrt 2\right]$ (and vice versa). Thus, we arrive at the following theorem, whose proof can be found in Appendix~\ref{ap:th_3}.

\begin{theorem}\label{th:exact_synthesis}
    All matchgate unitaries $U\in\mathbb{SPIN}(2n)$  such that $U\otimes U^*$ has entries in the ring $\mathbb Z\left[\frac{1}{\sqrt 2},i\right]$ can be exactly synthesized by a sequence of gates from the set $\GC={\left\{\overline  T_q,\,\overline S_q\right\}}_{q=1}^n \cup {\left\{R_{q,q+1}^{xx}\left(\frac{\pi}{2}\right)\right\}}_{q=1}^{n-1}$, without  ancillas.  
\end{theorem}

Theorem~\ref{th:exact_synthesis} completely characterizes the matchgate-Clifford$+\overline{T}$ group, by identifying it with the group of matchgate unitaries $U$ whose adjoint representation has entries in $\mathbb Z\left[\frac{1}{\sqrt 2},i\right]$. This result resembles its counterpart in $\mathbb{SU}(2^n)$, where the Clifford$+T$ group without ancillas is precisely the group of unitaries with entries in the ring $\mathbb Z\left[\frac{1}{\sqrt 2},i\right]$~\cite{giles2013exact}~\footnote{Notice however the subtlety that for matchgates it is the adjoint representation the one with entries in $\mathbb Z\left[\frac{1}{\sqrt 2},i\right]$, whereas in the unitary group it is the standard representation.}.

At this point we find it important to highlight the fact that our proof of Theorem~\ref{th:exact_synthesis} is constructive. That is, assuming that $Q\in\SO(2n)$ has entries in $\mathbb D\left[\sqrt{2}\right]$, we  provide a classical algorithm that outputs an exact decomposition of $Q$ in terms of gates in $\widetilde \GC$. The algorithm is an adaptation of the exact synthesis method in Ref.~\cite{giles2013exact} to the group $\SO(2n)$. Interestingly, the algorithm for $\SU(2^n)$ runs in a time which is exponential in the size of the matrix (doubly-exponential with the number of qubits)~\cite{kliuchnikov2013synthesis}, whereas the case of $\SO(2n)$ renders a polynomial scaling algorithm with system size. More specifically, the algorithm is based on Gaussian elimination, and it proceeds by sequentially transforming the $j$-th column of the matrix into the canonical-basis vector $\vec{e}_j$, using only operations from $\widetilde{\GC}$. 
Repeating this procedure eventually reduces $Q$ to the identity. Then, inverting the applied sequence yields an exact decomposition of the original matrix $Q$.

The complexity of the synthesis algorithm depends on the number of qubits $n$ and the so-called least denominator exponent (LDE) $k$ of the matrix. Intuitively, $k$ measures how large the denominators in Eq.~\eqref{eq:dyadic_ring} can be, or equivalently, how many classical bits are required to represent the entries of the matrix. More precisely, denoting $\ZBB\left[\sqrt{2}\right] := \left\{a + b\sqrt{2} \; |\; a,b\in \mathbb Z \right\}$, we have the following important definition.

\begin{definition}\label{def:denominator_exponent}
       Given $r\in \mathbb D\left[\sqrt 2\right]$, we say that $k$ is a denominator exponent of $r$ whenever $\sqrt 2^k r\in \mathbb Z\left[\sqrt 2\right]$. When $k$ is a denominator exponent of $r$ and $\sqrt2^{k-1}r \notin \mathbb Z\left[\sqrt 2\right]$, we say that $k$ is the least denominator exponent of $r$. 
   \end{definition}

Our synthesis algorithm runs in time $\mathcal{O}\!\big(n^4 k_{\max}\big)$, where $k_{\max}$ is the maximum least denominator exponent among the entries of $Q$. The algorithm is therefore polynomial in both $n$ and $k_{\max}$. Moreover, the scaling of the algorithm in $k_{\max}$ is optimal, but it remains open whether the $n^4$ factor can be improved, and how the best achievable scaling relates to the dimension of the matchgate Lie algebra, $\dim(\so(2n))=n(2n-1)$.

\subsubsection{Bounds on the number of gates in \\ exact matchgate synthesis}\label{sec:T_bounds}

In the fault-tolerant regime, the dominant cost is typically the number of non-Clifford gates in a circuit, i.e., the $T$-count~\cite{campbell2017roads}. 
Using the Gaussian-elimination strategy from Theorem~\ref{th:exact_synthesis}, we derive upper bounds on the gate counts required for exact synthesis. We present the result in the next theorem. 

\begin{theorem}\label{th:gate_counts}
    Given an $n$-qubit matchgate unitary $U\in\mathbb{SPIN}(2n)$ such that $\Phi(U\otimes U^*)\in\SO(2n)$ has entries in the ring $\mathbb D\left[\sqrt 2\right]$ with maximum least denominator exponent $k_{\max}$, the number of $\overline T$ gates, $N_{T}$, and the number of Clifford gates from $\left\{\overline  S_q\right\}_{q=1}^n \cup {\left\{R_{q,q+1}^{xx}\left(\frac{\pi}{2}\right)\right\}}_{q=1}^{n-1}$, $N_C$, needed to exactly synthesize $U$ are bounded by
    \begin{equation}
    \begin{split}
        N_T\,\leq\;& \frac 23 n^3 k_{\max} + \mathcal O(n^2 k_{\max})\, ,\\
       N_C\,\leq\;&  \frac 43  n^4 k_{\max}+  \mathcal O(n^3 k_{\max} )\,.
        \end{split}
    \end{equation}
\end{theorem}
The proof of this theorem is given in Appendix~\ref{ap:th_4}. The main idea is that the $\overline T$ gates are the only ones that mix Majorana operators, introducing factors of $\frac{1}{\sqrt 2}$ in the Pauli transfer matrix $\Phi(U\otimes U^*)$ (recall Eq.~\eqref{eq:T_SO}). Consequently, the least denominator exponent $k_{\max}$ of $\Phi(U\otimes U^*)$  is directly related to the $\overline T$-count. 

We anticipate that the bounds in Theorem~\ref{th:gate_counts} will be loose in practice, as they assume that all entries have the same least denominator exponent $k_{\max}$, and they also account for the largest possible growth of this denominator exponent at each step of the column-reduction procedure employed in Theorem~\ref{th:exact_synthesis}. Moreover, we expect that global strategies that actively minimize such growth could provide better bounds. 

Besides upper-bounding the total $\overline T$-count, the least denominator exponent also imposes a fundamental lower bound on the $T$-depth of any exact synthesis circuit. Since the Clifford matchgates $\overline S$ and $R^{xx}\left(\frac{\pi}{2}\right)$ have Pauli transfer matrices with integer entries and therefore do not increase the least denominator exponent of $\Phi(U\otimes U^*)$, the $\overline T$ gate is the only generator that introduces factors of $\frac{1}{\sqrt{2}}$.
Consequently, grouping the circuit into layers of $\overline T$ gates interleaved with arbitrary Clifford matchgates, each $\overline{T}$-layer can increase the least denominator exponent by at most one.
Any circuit that exactly synthesizes a target $\Phi(U\otimes U^*)\in\SO(2n)$ with maximum least denominator exponent $k_{\max}$ must therefore have 
\begin{equation}
    T\text{-depth}(U)\ \ge\ k_{\max}\,.
\end{equation}
This bound holds independently of the total $\overline T$-count and highlights that $k_{\max}$ controls not only the minimal non-Clifford gate count but also the minimal non-Clifford circuit depth.

\begin{figure*}[t]
    \centering
    \includegraphics[width=1\linewidth]{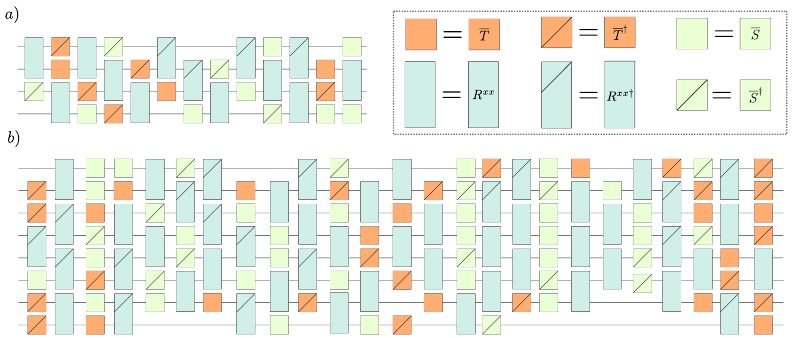}
    \caption{\textbf{Compilation of matchgate circuits}. Circuits $(a)$ and $(b)$ perform the exact diagonalization of the $XX$ Hamiltonian for $n= 4$ and $n= 8$ qubits, respectively (see Ref.~\cite{verstraete2009quantum}). These circuits were compiled into the gate set $\GC$ (including inverses) using the \texttt{kissat} SAT solver~\cite{biere2024cadical}.
     For $n=4$, the compiled circuit has provably optimal depth $13$ with $T$-count $8$. For $n=8$, we obtained a depth-$25$ circuit with $T$-count $34$; for $d < 23$ there does not exist a solution, and for depths $22< d < 25$ our SAT-based searches were not conclusive within the explored time and resource limits.
 }\label{fig:xx_model}
\end{figure*}

\subsubsection{Optimal exact matchgate synthesis}

The previous section answers the existence question, in the sense that  given a matchgate $U\in\mathbb{SPIN}(2n)$ such that $U\otimes U^*$ has entries in $\ZBB\left[\frac{1}{\sqrt{2}},i\right]$, we can always implement it exactly with a finite sequence of gates from $\GC$. We now turn to the question of optimality, formulated as follows.

\begin{problem}[Optimal exact matchgate synthesis]\label{prob:synthesis} Given a target matchgate unitary $U\in\mathbb{SPIN}(2n)$ such that $U\otimes U^*$ has entries in $\ZBB\left[\frac{1}{\sqrt{2}},i\right]$, and a circuit depth $d_{\max}\in\OC\left(\poly(n)\right)$: does there exist a circuit $U_d$ of depth $d\leq d_{\max}$ consisting of gates from $\GC$ such that $U=\pm U_d$?
\end{problem}

Notice that Problem~\ref{prob:synthesis} is a decision problem, with a yes/no answer. To find an actual circuit with minimal depth, one can perform a (e.g., binary) search on $d$, solving Problem~\ref{prob:synthesis} for each depth. This reformulates the optimal synthesis task as a set of decision problems.

Problem~\ref{prob:synthesis} clearly belongs to the NP complexity class: given a candidate circuit $U_d$ for $U$, one can verify its correctness by working  in the standard representation of $\SO(2n)$ and multiplying the corresponding matrices from $\widetilde{\GC}$ in the decomposition of $\Phi(U_d\otimes U^*_d)$, whose dimension grows polynomially with $n$. However, the hardness of the problem is unclear, and determining whether it is NP-complete is left for future work. 

\subsubsection{(MAX-)SAT formulation for optimal exact matchgate synthesis}\label{sec:SAT}

In Problem~\ref{prob:synthesis} we seek a depth-optimal exact decomposition of a target matchgate $U\in\mathbb{SPIN}(2n)$ using the generating set $\GC$ of Theorem~\ref{th:universal_matchgates}.
Theorem~\ref{th:exact_synthesis} guarantees the existence of such decomposition whenever the Pauli transfer matrix $Q = \Phi(U\otimes U^*)\in\SO(2n)$ has entries in the dyadic ring $\mathbb D\left[\sqrt 2\right]$, defined in Eq.~\eqref{eq:D_sqrt2_ring}.  Our goal here is to turn this exact synthesis problem into a sequence of (MAX-)SAT instances, whose satisfying assignments correspond to depth-$d$ matchgate circuits implementing $U\otimes U^*$.

The starting point is the general SAT encoding introduced in Ref.~\cite{gouzien2025provably}~\footnote{Notice that the idea of mapping quantum circuit synthesis to Boolean satisfiability first appeared in Ref.~\cite{peham2023depth}, for Clifford unitaries. However, that approach is specific to Clifford circuits.}, and detailed in Appendix~\ref{ap:SAT_mapping}. A depth-$d$ circuit implementing the $\SO(2n)$ transfer matrix $Q$ is written as
\begin{equation}
    Q = \prod_{i = 1}^d \left[ \sum_{j = 1} ^M x_{ij}G^{(j)} \right]\,,
    \label{eq:SAT_decomposition}
\end{equation}
where each $G^{(j)}$ is chosen from the finite set of generators $\widetilde\GC\subset\SO(2n)$ of size $|\widetilde G| = M$. The Boolean selector variables $x_{ij}\in \{0,1\}$ are constrained such that $\sum_{j}x_{ij} = 1$, ensuring that exactly one generator $G^{(j)}$ is selected at depth $i$. We generalize this constraint to allow multiple commuting generators to be applied within the same layer (this is an improvement over the original encoding in Ref.~\cite{gouzien2025provably}, available at~\cite{SATsynthesis}). 
To do so, we note that each generator $G^{(j)}\in \widetilde \GC$ is sparse, performing a (Givens) rotation within a two-dimensional subspace corresponding to a single pair of rows, while acting as the identity on the rest.
Two such generators commute whenever their supports are disjoint.  This allows us to encode, at each time step, an entire  layer of parallel commuting matchgates by assigning at most one gate to each pair of rows.  The explicit construction and the corresponding Boolean constraints are detailed in Appendix~\ref{ap:parallel_layers}. 

To enforce the matrix equality in Eq.~\eqref{eq:SAT_decomposition}, we `bit-blast’ the algebraic constraints over $\mathbb D\left[\sqrt 2\right]$ into binary clauses~\cite{gouzien2025provably}. These clauses, along with structural constraints imposed on the variables $x_{ij}$, are then passed to a SAT solver (see Appendix~\ref{ap:SAT_mapping} for additional details); in particular, we rely on \texttt{kissat}~\cite{biere2024cadical}.

In addition, to favor solutions with minimal $\overline{T}$-count, we promote the resulting SAT instances to MAX-SAT by adding soft clauses that penalize the use of $\overline{T}$ gates, while keeping all structural and arithmetic constraints as hard clauses.  For a fixed depth $d$, a modern MAX-SAT solver (in our numerics we rely on \texttt{open-wbo}~\cite{martins2014open}) thus returns, if one exists, a depth-$d$ brickwork circuit with the minimal number of $\overline T$ gates.  Combining this with a binary search over $d$ yields circuits that are provably depth-optimal  and, among all depth-optimal circuits, have minimal $\overline{T}$-count.

Finally, in Appendix~\ref{app:gaussian_state_prep} we discuss how this method can be adapted to exactly prepare Gaussian states, provided that their representation in $\SO(2n)$  has entries in the $\mathbb D\left[\sqrt 2\right]$ ring. The number of clauses required for state preparation drops by a factor of two compared to full matchgate synthesis.

\subsubsection{Exact compilation of the $XX$ diagonalizing circuit}\label{sec:xx_exact_compilation}
As a benchmark of our SAT-based synthesis algorithm, we consider the circuits that diagonalize the free-fermionic $XX$ Hamiltonian~\cite{verstraete2009quantum}. In Appendix~\ref{ap:compliation_XX_diagonalizing}, we provide more details on these circuits. 

The target unitary, $U_{\mathrm{dis}}$, has entries in the ring $\mathbb Z\left[\frac{1}{\sqrt{2}}, i\right]$ when the number of qubits is $n=4$ or $n=8$. Hence, its associated Pauli transfer matrix satisfies $Q_{\mathrm{dis}}\in\SO(2n)\cap{\mathbb D\left[\sqrt 2\right]}^{2n\times 2n}$, and, by Theorem~\ref{th:exact_synthesis}, there exists a finite-depth matchgate circuit over the matchgate-Clifford$+\overline T$ set that implements this unitary exactly. 

Our SAT-based approach allows us to construct the circuits while certifying their optimality regarding both  depth and (given the optimal depth)  $\overline T$-count. In particular, we compiled $Q_{\mathrm{dis}}$ with provably optimal depth and associated minimal $\overline{T}$-count for $n=4$, and obtained a depth-$25$ solution for $n=8$.
The corresponding circuits are shown in Fig.~\ref{fig:xx_model}. 
For the $n=4$ instance, we first used \texttt{kissat} to determine the optimal depth, which was found to be 13. We then employed \texttt{open-wbo}~\cite{martins2014open} to search for the optimal $\overline{T}$-count solution at this fixed depth, ultimately identifying the shallowest optimal compilation with a $\overline{T}$-count of 8.

Regarding the $n=8$ instance, at $d=25$ \texttt{kissat} found a solution, while at $d=22$ the solver proved the unsatisfiability of the Boolean constraints, certifying that no depth-$22$, or lower, brickwork matchgate circuit over our gate set can realize the target $Q_{\mathrm{dis}}$. The intermediate depths $d=23,24$ remain to be fully explored, as  convergence was not reached. Similarly, the MAX-SAT run for $d=25$ proved to be too computationally demanding. Thus, for the $n=8$ instance we are unable to claim optimality, though the current result provides a high-quality upper bound.

\section{Benchmarking with matchgate synthesis}\label{sec:benchmarking}
Besides enabling the synthesis of free-fermionic evolutions, the compilation of matchgate circuits provides a natural framework for benchmarking early fault-tolerant quantum computers. In this regime, circuits are executed only after they are compiled into an error-correcting code’s discrete logical gate set. Consequently, benchmarking protocols must target the fidelity of the compiled circuits.

Classical simulation offers a powerful tool for benchmarking and verification, but only for circuits that remain tractable under other classical methods when scaling the system size up. A prominent example of circuits that fall into this category are those from the Clifford group, acting on stabilizer states.  However, in most leading QEC architectures, Clifford gates are relatively inexpensive, as they are often transversal or require modest overhead~\cite{fowler2012surface}.  
In contrast, non-Clifford gates such as the $T$ gate are significantly more costly. Hence, as quantum devices scale, validating the correct implementation of non-Clifford resources at the logical level becomes increasingly important.

In this context, matchgate circuits sit at an interesting intermediate point: they are efficiently classically simulable (for standard computational-basis measurements and input states), yet they can still generate significant amounts of  magic~\cite{collura2024quantum}. 
To measure the amount of magic of an $n$-qubit state $\rho$, we can consider the so-called stabilizer entropy~\cite{leone2022stabilizer},
\begin{equation}\label{eq:magic_purity}
     S(\rho) = \frac{1}{2^n} \sum_{P \in \mathcal{P}_n} \Tr{[\rho P]}^4\,,
\end{equation}
where the sum runs over all Pauli operators $P \in \mathcal{P}_n = \{I, X, Y, Z\}^{\otimes n}$. This quantity represents the fourth-order moment of the Pauli distribution associated with $\rho$, up to a $2^{-n}$ factor. For stabilizer states, all nonzero Pauli expectation values have magnitude one, yielding $S(\rho)=1$, while non-stabilizer states have smaller values of $S(\rho)$ because their weights on the Pauli basis are more evenly spread. 
In Fig.~\ref{fig:magic-SO}, we compare the value of $S$ for states generated by  matchgate circuits with $\ell$ gates uniformly drawn at random from $\mathcal{G}$, to that expected for Haar-random states~\cite{deneris2025analyzing}. 
The relative difference between these two values increases with the number of qubits, indicating that matchgate circuits produce  less magic than Haar-random unitaries~\footnote{Importantly, even though we do not sample matchgate unitaries according to the corresponding Haar measure on $\mathbb{SPIN}(2n)$, the found values of $S$ in the large-$d$ limit agree with the behavior of Haar-random fermionic Gaussian states, as reported in Ref.~\cite{collura2024quantum}.}.

\begin{figure}[t]
    \centering
    \includegraphics[width=\linewidth]{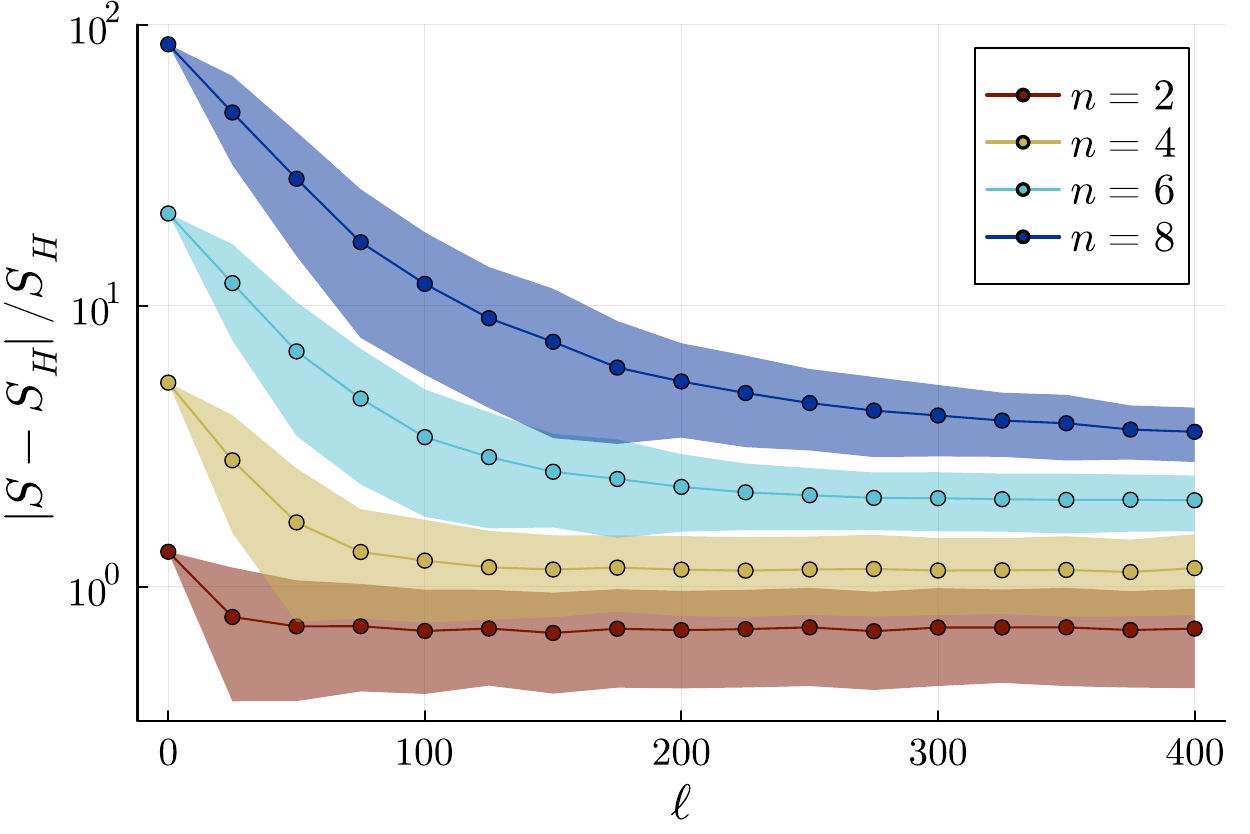}
    \caption{\textbf{Magic of Haar-random states vs.\@ random matchgate circuits.} Relative distance between the average magic $S_H$ of Haar-random states (computed analytically in Ref.~\cite{deneris2025analyzing}), and the magic $S$ of states generated by matchgate circuits acting on $|0\rangle^{\otimes n}$, whose $\ell$ gates are drawn uniformly at random from $\mathcal G$ (computed numerically). The circles show the mean over $1000$ independent samples of matchgate circuits with $\ell$ gates, and the shaded bands indicate one standard deviation about the mean. Different colors correspond to different numbers of qubits.}\label{fig:magic-SO}
\end{figure}

This makes them attractive candidates for testing device performance in the early fault-tolerant regime. Besides being structured enough to admit scalable classical validation, matchgate circuits generate a controlled amount of magic, sufficient to probe truly non-Clifford behavior, yet significantly below the requirements of Haar-random unitaries. This intermediate regime renders  matchgate-based benchmarking protocols remarkably suitable for early fault-tolerant devices, where full Haar-random circuits might be too demanding.

Indeed, we note that  a recent hardware-agnostic volumetric benchmark called the Free-Fermion Volume (FFV) has been proposed~\cite{portik2025clifford}. Here, one samples random matchgate unitaries acting on a simple fiducial Gaussian state, and then measures Majorana operators. Complementary to those ideas, verification techniques for arbitrary quantum computations based on matchgate circuits have also been recently proposed~\cite{carrasco2024gaining}. In the fault-tolerant setting, such protocols will necessarily rely on discrete logical gate sets, and then our proposed framework for matchgate synthesis may become particularly useful. Especially as the set  $\GC$ is well aligned with the native operations of some important quantum-computing platforms, such as some based on trapped ions~\cite{chen2024benchmarking} or neutral atoms~\cite{wurtz2023aquila}.

Finally, it is worth emphasizing that, in architectures where SWAP operations can be implemented with negligible error, either through physical ion movement~\cite{walther2012controlling,ruster2014experimental}, photonic routing~\cite{taballione20198}, or high-fidelity logical swaps~\cite{horsman2012surface}, benchmarking matchgate operations effectively probes a universal gate set~\cite{jozsa2008matchgates,brod2014computational}.  Hence, fidelity estimates obtained from matchgate benchmarks~\cite{helsen2022matchgate} on such platforms will provide meaningful information regarding the broader computational capability of the device.

\section{Outlook}\label{sec:outlook}

In this work, we have uncovered a universal matchgate set, and we have begun to explore the possibility of compiling matchgate circuits using only matchgates. We have shown how this idea allows for an exponential reduction in the size of the compiled matrices, and we have addressed both the approximate and exact variants of the matchgate synthesis problem. Our results may find broad applicability in the simulation of fermionic systems, and in benchmarking early fault-tolerant quantum computers.

Importantly, our work opens several research avenues. For instance, as already explained, one could adapt the SK algorithm in~\cite{kuperberg2023breaking} to $\SO(2n)$ using our discrete matchgate set, and study how it outperforms local strategies based on $\SU(2)$ compilers. One could also explore normal forms~\cite{matsumoto2008representation, giles2013remarks} for matchgates, or the use of ancilla qubits~\cite{tan2025unitary} or mid-circuit measurements~\cite{baumer2024efficient} to try to reduce the depth of the compiled circuits. Moreover, it would be interesting to tackle physically motivated free-fermionic systems at scale, introducing symmetries such as translational invariance into the mix.  Yet another possibility is to use machine-learning techniques in the synthesis process~\cite{zhang2024scalable,ruiz2025quantum,bosco2026quantum}.

In summary, the connection between matchgate circuits and the standard representation of $\SO(2n)$ provides a natural framework and playground to combine number- and group-theoretic ideas,  heuristic optimization and combinatorial search for matchgate synthesis.

\section*{Acknowledgments}

We thank  Martin Larocca, Lukasz Cincio, N. L. Diaz, Ricard Puig, Austin Pechan, Scott Pakin, Max West, Aniruddha Sen and Richard K\"ueng for insightful conversations. B.C. acknowledges funding from the Spanish Ministry for Digital Transformation and the Civil Service of the Spanish Government through the QUANTUM ENIA project call - Quantum Spain, EU, through the Recovery, Transformation and Resilience Plan – NextGenerationEU, within the framework of Digital Spain 2026. B.C. was also initially supported  by the U.S. Department of Energy (DOE) through a quantum computing program sponsored by the Los Alamos National Laboratory (LANL) Information Science \& Technology Institute. P.B., M.C. and D.G.-M. were supported by the  Laboratory Directed Research and Development (LDRD) program of LANL under project number 20260043DR. MC and D.G.-M. were also supported by by LANL's ASC Beyond Moore’s Law project. D.G.-M. also acknowledges financial support from the European Research Council (ERC) via the Starting grant q-shadows (101117138) and from the Austrian Science Fund (FWF) via the SFB BeyondC (10.55776/FG7).

\bibliography{quantum}

\begin{thebibliography}{127}%
\makeatletter
\providecommand \@ifxundefined [1]{%
 \@ifx{#1\undefined}
}%
\providecommand \@ifnum [1]{%
 \ifnum #1\expandafter \@firstoftwo
 \else \expandafter \@secondoftwo
 \fi
}%
\providecommand \@ifx [1]{%
 \ifx #1\expandafter \@firstoftwo
 \else \expandafter \@secondoftwo
 \fi
}%
\providecommand \natexlab [1]{#1}%
\providecommand \enquote  [1]{``#1''}%
\providecommand \bibnamefont  [1]{#1}%
\providecommand \bibfnamefont [1]{#1}%
\providecommand \citenamefont [1]{#1}%
\providecommand \href@noop [0]{\@secondoftwo}%
\providecommand \href [0]{\begingroup \@sanitize@url \@href}%
\providecommand \@href[1]{\@@startlink{#1}\@@href}%
\providecommand \@@href[1]{\endgroup#1\@@endlink}%
\providecommand \@sanitize@url [0]{\catcode `\\12\catcode `\$12\catcode `\&12\catcode `\#12\catcode `\^12\catcode `\_12\catcode `\%12\relax}%
\providecommand \@@startlink[1]{}%
\providecommand \@@endlink[0]{}%
\providecommand \url  [0]{\begingroup\@sanitize@url \@url }%
\providecommand \@url [1]{\endgroup\@href {#1}{\urlprefix }}%
\providecommand \urlprefix  [0]{URL }%
\providecommand \Eprint [0]{\href }%
\providecommand \doibase [0]{https://doi.org/}%
\providecommand \selectlanguage [0]{\@gobble}%
\providecommand \bibinfo  [0]{\@secondoftwo}%
\providecommand \bibfield  [0]{\@secondoftwo}%
\providecommand \translation [1]{[#1]}%
\providecommand \BibitemOpen [0]{}%
\providecommand \bibitemStop [0]{}%
\providecommand \bibitemNoStop [0]{.\EOS\space}%
\providecommand \EOS [0]{\spacefactor3000\relax}%
\providecommand \BibitemShut  [1]{\csname bibitem#1\endcsname}%
\let\auto@bib@innerbib\@empty
\bibitem [{\citenamefont {Wang}\ \emph {et~al.}(2024)\citenamefont {Wang}, \citenamefont {Czarnik}, \citenamefont {Arrasmith}, \citenamefont {Cerezo}, \citenamefont {Cincio},\ and\ \citenamefont {Coles}}]{wang2021can}%
  \BibitemOpen
  \bibfield  {author} {\bibinfo {author} {\bibfnamefont {S.}~\bibnamefont {Wang}}, \bibinfo {author} {\bibfnamefont {P.}~\bibnamefont {Czarnik}}, \bibinfo {author} {\bibfnamefont {A.}~\bibnamefont {Arrasmith}}, \bibinfo {author} {\bibfnamefont {M.}~\bibnamefont {Cerezo}}, \bibinfo {author} {\bibfnamefont {L.}~\bibnamefont {Cincio}},\ and\ \bibinfo {author} {\bibfnamefont {P.~J.}\ \bibnamefont {Coles}},\ }\bibfield  {title} {\bibinfo {title} {Can error mitigation improve trainability of noisy variational quantum algorithms?},\ }\href {https://doi.org/10.22331/q-2024-03-14-1287} {\bibfield  {journal} {\bibinfo  {journal} {Quantum}\ }\textbf {\bibinfo {volume} {8}},\ \bibinfo {pages} {1287} (\bibinfo {year} {2024})}\BibitemShut {NoStop}%
\bibitem [{\citenamefont {Endo}\ \emph {et~al.}(2021)\citenamefont {Endo}, \citenamefont {Cai}, \citenamefont {Benjamin},\ and\ \citenamefont {Yuan}}]{endo2021hybrid}%
  \BibitemOpen
  \bibfield  {author} {\bibinfo {author} {\bibfnamefont {S.}~\bibnamefont {Endo}}, \bibinfo {author} {\bibfnamefont {Z.}~\bibnamefont {Cai}}, \bibinfo {author} {\bibfnamefont {S.~C.}\ \bibnamefont {Benjamin}},\ and\ \bibinfo {author} {\bibfnamefont {X.}~\bibnamefont {Yuan}},\ }\bibfield  {title} {\bibinfo {title} {Hybrid quantum-classical algorithms and quantum error mitigation},\ }\href {https://doi.org/10.7566/JPSJ.90.032001} {\bibfield  {journal} {\bibinfo  {journal} {Journal of the Physical Society of Japan}\ }\textbf {\bibinfo {volume} {90}},\ \bibinfo {pages} {032001} (\bibinfo {year} {2021})}\BibitemShut {NoStop}%
\bibitem [{\citenamefont {Gottesman}(2010)}]{gottesman2009introduction}%
  \BibitemOpen
  \bibfield  {author} {\bibinfo {author} {\bibfnamefont {D.}~\bibnamefont {Gottesman}},\ }\bibfield  {title} {\bibinfo {title} {An introduction to quantum error correction and fault-tolerant quantum computation},\ }\href {https://doi.org/10.1090/psapm/068/2762145} {\bibfield  {journal} {\bibinfo  {journal} {Quantum information science and its contributions to mathematics, Proceedings of Symposia in Applied Mathematics}\ }\textbf {\bibinfo {volume} {63}},\ \bibinfo {pages} {13} (\bibinfo {year} {2010})}\BibitemShut {NoStop}%
\bibitem [{\citenamefont {Eastin}\ and\ \citenamefont {Knill}(2009)}]{eastin2009restrictions}%
  \BibitemOpen
  \bibfield  {author} {\bibinfo {author} {\bibfnamefont {B.}~\bibnamefont {Eastin}}\ and\ \bibinfo {author} {\bibfnamefont {E.}~\bibnamefont {Knill}},\ }\bibfield  {title} {\bibinfo {title} {Restrictions on transversal encoded quantum gate sets},\ }\href {https://doi.org/10.1103/PhysRevLett.102.110502} {\bibfield  {journal} {\bibinfo  {journal} {Physical review letters}\ }\textbf {\bibinfo {volume} {102}},\ \bibinfo {pages} {110502} (\bibinfo {year} {2009})}\BibitemShut {NoStop}%
\bibitem [{\citenamefont {Kitaev}(1997)}]{kitaev1997quantum}%
  \BibitemOpen
  \bibfield  {author} {\bibinfo {author} {\bibfnamefont {A.~Y.}\ \bibnamefont {Kitaev}},\ }\bibfield  {title} {\bibinfo {title} {Quantum computations: algorithms and error correction},\ }\href {https://doi.org/10.1070/RM1997v052n06ABEH002155} {\bibfield  {journal} {\bibinfo  {journal} {Russian Mathematical Surveys}\ }\textbf {\bibinfo {volume} {52}},\ \bibinfo {pages} {1191} (\bibinfo {year} {1997})}\BibitemShut {NoStop}%
\bibitem [{\citenamefont {Kitaev}\ \emph {et~al.}(2002)\citenamefont {Kitaev}, \citenamefont {Shen},\ and\ \citenamefont {Vyalyi}}]{kitaev2002classical}%
  \BibitemOpen
  \bibfield  {author} {\bibinfo {author} {\bibfnamefont {A.~Y.}\ \bibnamefont {Kitaev}}, \bibinfo {author} {\bibfnamefont {A.}~\bibnamefont {Shen}},\ and\ \bibinfo {author} {\bibfnamefont {M.~N.}\ \bibnamefont {Vyalyi}},\ }\href@noop {} {\emph {\bibinfo {title} {Classical and quantum computation}}},\ \bibinfo {number} {47}\ (\bibinfo  {publisher} {American Mathematical Soc.},\ \bibinfo {year} {2002})\BibitemShut {NoStop}%
\bibitem [{\citenamefont {Dawson}\ and\ \citenamefont {Nielsen}(2005)}]{dawson2005solovay}%
  \BibitemOpen
  \bibfield  {author} {\bibinfo {author} {\bibfnamefont {C.~M.}\ \bibnamefont {Dawson}}\ and\ \bibinfo {author} {\bibfnamefont {M.~A.}\ \bibnamefont {Nielsen}},\ }\bibfield  {title} {\bibinfo {title} {The solovay-kitaev algorithm},\ }\bibfield  {journal} {\bibinfo  {journal} {arXiv preprint quant-ph/0505030}\ }\href {https://doi.org/10.48550/arXiv.quant-ph/0505030} {10.48550/arXiv.quant-ph/0505030} (\bibinfo {year} {2005})\BibitemShut {NoStop}%
\bibitem [{\citenamefont {Fowler}\ \emph {et~al.}(2012)\citenamefont {Fowler}, \citenamefont {Mariantoni}, \citenamefont {Martinis},\ and\ \citenamefont {Cleland}}]{fowler2012surface}%
  \BibitemOpen
  \bibfield  {author} {\bibinfo {author} {\bibfnamefont {A.~G.}\ \bibnamefont {Fowler}}, \bibinfo {author} {\bibfnamefont {M.}~\bibnamefont {Mariantoni}}, \bibinfo {author} {\bibfnamefont {J.~M.}\ \bibnamefont {Martinis}},\ and\ \bibinfo {author} {\bibfnamefont {A.~N.}\ \bibnamefont {Cleland}},\ }\bibfield  {title} {\bibinfo {title} {Surface codes: Towards practical large-scale quantum computation},\ }\href {https://doi.org/10.1103/PhysRevA.86.032324} {\bibfield  {journal} {\bibinfo  {journal} {Physical Review A}\ }\textbf {\bibinfo {volume} {86}},\ \bibinfo {pages} {032324} (\bibinfo {year} {2012})}\BibitemShut {NoStop}%
\bibitem [{\citenamefont {Campbell}\ \emph {et~al.}(2017)\citenamefont {Campbell}, \citenamefont {Terhal},\ and\ \citenamefont {Vuillot}}]{campbell2017roads}%
  \BibitemOpen
  \bibfield  {author} {\bibinfo {author} {\bibfnamefont {E.~T.}\ \bibnamefont {Campbell}}, \bibinfo {author} {\bibfnamefont {B.~M.}\ \bibnamefont {Terhal}},\ and\ \bibinfo {author} {\bibfnamefont {C.}~\bibnamefont {Vuillot}},\ }\bibfield  {title} {\bibinfo {title} {Roads towards fault-tolerant universal quantum computation},\ }\href {https://doi.org/https://doi.org/10.1038/nature23460} {\bibfield  {journal} {\bibinfo  {journal} {Nature}\ }\textbf {\bibinfo {volume} {549}},\ \bibinfo {pages} {172} (\bibinfo {year} {2017})}\BibitemShut {NoStop}%
\bibitem [{\citenamefont {Litinski}(2019)}]{litinski2019game}%
  \BibitemOpen
  \bibfield  {author} {\bibinfo {author} {\bibfnamefont {D.}~\bibnamefont {Litinski}},\ }\bibfield  {title} {\bibinfo {title} {A {G}ame of {S}urface {C}odes: {L}arge-{S}cale {Q}uantum {C}omputing with {L}attice {S}urgery},\ }\href {https://doi.org/10.22331/q-2019-03-05-128} {\bibfield  {journal} {\bibinfo  {journal} {{Quantum}}\ }\textbf {\bibinfo {volume} {3}},\ \bibinfo {pages} {128} (\bibinfo {year} {2019})}\BibitemShut {NoStop}%
\bibitem [{\citenamefont {Fowler}\ and\ \citenamefont {Gidney}(2018)}]{fowler2018low}%
  \BibitemOpen
  \bibfield  {author} {\bibinfo {author} {\bibfnamefont {A.~G.}\ \bibnamefont {Fowler}}\ and\ \bibinfo {author} {\bibfnamefont {C.}~\bibnamefont {Gidney}},\ }\bibfield  {title} {\bibinfo {title} {Low overhead quantum computation using lattice surgery},\ }\bibfield  {journal} {\bibinfo  {journal} {arXiv preprint arXiv:1808.06709}\ }\href {https://doi.org/10.48550/arXiv.1808.06709} {10.48550/arXiv.1808.06709} (\bibinfo {year} {2018})\BibitemShut {NoStop}%
\bibitem [{\citenamefont {Bravyi}\ and\ \citenamefont {Kitaev}(2005)}]{bravyi2005universal}%
  \BibitemOpen
  \bibfield  {author} {\bibinfo {author} {\bibfnamefont {S.}~\bibnamefont {Bravyi}}\ and\ \bibinfo {author} {\bibfnamefont {A.}~\bibnamefont {Kitaev}},\ }\bibfield  {title} {\bibinfo {title} {Universal quantum computation with ideal clifford gates and noisy ancillas},\ }\href {https://doi.org/10.1103/PhysRevA.71.022316} {\bibfield  {journal} {\bibinfo  {journal} {Phys. Rev. A}\ }\textbf {\bibinfo {volume} {71}},\ \bibinfo {pages} {022316} (\bibinfo {year} {2005})}\BibitemShut {NoStop}%
\bibitem [{\citenamefont {Gidney}\ \emph {et~al.}(2024)\citenamefont {Gidney}, \citenamefont {Shutty},\ and\ \citenamefont {Jones}}]{gidney2024magic}%
  \BibitemOpen
  \bibfield  {author} {\bibinfo {author} {\bibfnamefont {C.}~\bibnamefont {Gidney}}, \bibinfo {author} {\bibfnamefont {N.}~\bibnamefont {Shutty}},\ and\ \bibinfo {author} {\bibfnamefont {C.}~\bibnamefont {Jones}},\ }\bibfield  {title} {\bibinfo {title} {Magic state cultivation: growing t states as cheap as cnot gates},\ }\bibfield  {journal} {\bibinfo  {journal} {arXiv preprint arXiv:2409.17595}\ }\href {https://doi.org/10.48550/arXiv.2409.17595} {10.48550/arXiv.2409.17595} (\bibinfo {year} {2024})\BibitemShut {NoStop}%
\bibitem [{\citenamefont {Chamberland}\ and\ \citenamefont {Cross}(2019)}]{chamberland2019fault}%
  \BibitemOpen
  \bibfield  {author} {\bibinfo {author} {\bibfnamefont {C.}~\bibnamefont {Chamberland}}\ and\ \bibinfo {author} {\bibfnamefont {A.~W.}\ \bibnamefont {Cross}},\ }\bibfield  {title} {\bibinfo {title} {Fault-tolerant magic state preparation with flag qubits},\ }\href {https://doi.org/10.22331/q-2019-05-20-143} {\bibfield  {journal} {\bibinfo  {journal} {Quantum}\ }\textbf {\bibinfo {volume} {3}},\ \bibinfo {pages} {143} (\bibinfo {year} {2019})}\BibitemShut {NoStop}%
\bibitem [{\citenamefont {Paetznick}\ and\ \citenamefont {Reichardt}(2013)}]{paetznick2013universal}%
  \BibitemOpen
  \bibfield  {author} {\bibinfo {author} {\bibfnamefont {A.}~\bibnamefont {Paetznick}}\ and\ \bibinfo {author} {\bibfnamefont {B.~W.}\ \bibnamefont {Reichardt}},\ }\bibfield  {title} {\bibinfo {title} {Universal fault-tolerant quantum computation with only transversal gates and error correction},\ }\href {https://doi.org/10.1103/PhysRevLett.111.090505} {\bibfield  {journal} {\bibinfo  {journal} {Phys. Rev. Lett.}\ }\textbf {\bibinfo {volume} {111}},\ \bibinfo {pages} {090505} (\bibinfo {year} {2013})}\BibitemShut {NoStop}%
\bibitem [{\citenamefont {Bomb{\'\i}n}(2015)}]{bombin2015gauge}%
  \BibitemOpen
  \bibfield  {author} {\bibinfo {author} {\bibfnamefont {H.}~\bibnamefont {Bomb{\'\i}n}},\ }\bibfield  {title} {\bibinfo {title} {Gauge color codes: optimal transversal gates and gauge fixing in topological stabilizer codes},\ }\href {https://doi.org/10.1088/1367-2630/17/8/083002} {\bibfield  {journal} {\bibinfo  {journal} {New Journal of Physics}\ }\textbf {\bibinfo {volume} {17}},\ \bibinfo {pages} {083002} (\bibinfo {year} {2015})}\BibitemShut {NoStop}%
\bibitem [{\citenamefont {Beverland}\ \emph {et~al.}(2022)\citenamefont {Beverland}, \citenamefont {Murali}, \citenamefont {Troyer}, \citenamefont {Svore}, \citenamefont {Hoefler}, \citenamefont {Kliuchnikov}, \citenamefont {Low}, \citenamefont {Soeken}, \citenamefont {Sundaram},\ and\ \citenamefont {Vaschillo}}]{beverland2022assessing}%
  \BibitemOpen
  \bibfield  {author} {\bibinfo {author} {\bibfnamefont {M.~E.}\ \bibnamefont {Beverland}}, \bibinfo {author} {\bibfnamefont {P.}~\bibnamefont {Murali}}, \bibinfo {author} {\bibfnamefont {M.}~\bibnamefont {Troyer}}, \bibinfo {author} {\bibfnamefont {K.~M.}\ \bibnamefont {Svore}}, \bibinfo {author} {\bibfnamefont {T.}~\bibnamefont {Hoefler}}, \bibinfo {author} {\bibfnamefont {V.}~\bibnamefont {Kliuchnikov}}, \bibinfo {author} {\bibfnamefont {G.~H.}\ \bibnamefont {Low}}, \bibinfo {author} {\bibfnamefont {M.}~\bibnamefont {Soeken}}, \bibinfo {author} {\bibfnamefont {A.}~\bibnamefont {Sundaram}},\ and\ \bibinfo {author} {\bibfnamefont {A.}~\bibnamefont {Vaschillo}},\ }\bibfield  {title} {\bibinfo {title} {Assessing requirements to scale to practical quantum advantage},\ }\bibfield  {journal} {\bibinfo  {journal} {arXiv preprint arXiv:2211.07629}\ }\href {https://doi.org/10.48550/arXiv.2211.07629} {10.48550/arXiv.2211.07629} (\bibinfo {year} {2022})\BibitemShut {NoStop}%
\bibitem [{\citenamefont {Kliuchnikov}\ \emph {et~al.}(2013{\natexlab{a}})\citenamefont {Kliuchnikov}, \citenamefont {Maslov},\ and\ \citenamefont {Mosca}}]{kliuchnikov2013fast}%
  \BibitemOpen
  \bibfield  {author} {\bibinfo {author} {\bibfnamefont {V.}~\bibnamefont {Kliuchnikov}}, \bibinfo {author} {\bibfnamefont {D.}~\bibnamefont {Maslov}},\ and\ \bibinfo {author} {\bibfnamefont {M.}~\bibnamefont {Mosca}},\ }\bibfield  {title} {\bibinfo {title} {Fast and efficient exact synthesis of single qubit unitaries generated by clifford and t gates},\ }\href {https://doi.org/10.5555/2535649.2535653} {\bibfield  {journal} {\bibinfo  {journal} {Quantum Information and Computation}\ }\textbf {\bibinfo {volume} {13}},\ \bibinfo {pages} {607} (\bibinfo {year} {2013}{\natexlab{a}})}\BibitemShut {NoStop}%
\bibitem [{\citenamefont {Ross}\ and\ \citenamefont {Selinger}(2016)}]{ross2014optimal}%
  \BibitemOpen
  \bibfield  {author} {\bibinfo {author} {\bibfnamefont {N.~J.}\ \bibnamefont {Ross}}\ and\ \bibinfo {author} {\bibfnamefont {P.}~\bibnamefont {Selinger}},\ }\bibfield  {title} {\bibinfo {title} {Optimal ancilla-free clifford+{T} approximation of z-rotations},\ }\href {https://doi.org/0.26421/qic16.11-12-1} {\bibfield  {journal} {\bibinfo  {journal} {Quantum Info. Comput.}\ }\textbf {\bibinfo {volume} {16}},\ \bibinfo {pages} {901–953} (\bibinfo {year} {2016})}\BibitemShut {NoStop}%
\bibitem [{\citenamefont {Kim}(2025)}]{kim2025catalytic}%
  \BibitemOpen
  \bibfield  {author} {\bibinfo {author} {\bibfnamefont {I.~H.}\ \bibnamefont {Kim}},\ }\bibfield  {title} {\bibinfo {title} {Catalytic $ z $-rotations in constant $ t $-depth},\ }\bibfield  {journal} {\bibinfo  {journal} {arXiv preprint arXiv:2506.15147}\ }\href {https://doi.org/10.48550/arXiv.2506.15147} {10.48550/arXiv.2506.15147} (\bibinfo {year} {2025})\BibitemShut {NoStop}%
\bibitem [{\citenamefont {Giles}\ and\ \citenamefont {Selinger}(2013{\natexlab{a}})}]{giles2013exact}%
  \BibitemOpen
  \bibfield  {author} {\bibinfo {author} {\bibfnamefont {B.}~\bibnamefont {Giles}}\ and\ \bibinfo {author} {\bibfnamefont {P.}~\bibnamefont {Selinger}},\ }\bibfield  {title} {\bibinfo {title} {Exact synthesis of multiqubit clifford+ t circuits},\ }\href {https://doi.org/10.1103/PhysRevA.87.032332} {\bibfield  {journal} {\bibinfo  {journal} {Physical Review A}\ }\textbf {\bibinfo {volume} {87}},\ \bibinfo {pages} {032332} (\bibinfo {year} {2013}{\natexlab{a}})}\BibitemShut {NoStop}%
\bibitem [{\citenamefont {Bocharov}\ \emph {et~al.}(2015)\citenamefont {Bocharov}, \citenamefont {Roetteler},\ and\ \citenamefont {Svore}}]{bocharov2015efficient}%
  \BibitemOpen
  \bibfield  {author} {\bibinfo {author} {\bibfnamefont {A.}~\bibnamefont {Bocharov}}, \bibinfo {author} {\bibfnamefont {M.}~\bibnamefont {Roetteler}},\ and\ \bibinfo {author} {\bibfnamefont {K.~M.}\ \bibnamefont {Svore}},\ }\bibfield  {title} {\bibinfo {title} {Efficient synthesis of universal repeat-until-success quantum circuits},\ }\href {https://doi.org/10.1103/PhysRevLett.114.080502} {\bibfield  {journal} {\bibinfo  {journal} {Physical review letters}\ }\textbf {\bibinfo {volume} {114}},\ \bibinfo {pages} {080502} (\bibinfo {year} {2015})}\BibitemShut {NoStop}%
\bibitem [{\citenamefont {Kliuchnikov}\ \emph {et~al.}(2023)\citenamefont {Kliuchnikov}, \citenamefont {Lauter}, \citenamefont {Minko}, \citenamefont {Paetznick},\ and\ \citenamefont {Petit}}]{kliuchnikov2023shorter}%
  \BibitemOpen
  \bibfield  {author} {\bibinfo {author} {\bibfnamefont {V.}~\bibnamefont {Kliuchnikov}}, \bibinfo {author} {\bibfnamefont {K.}~\bibnamefont {Lauter}}, \bibinfo {author} {\bibfnamefont {R.}~\bibnamefont {Minko}}, \bibinfo {author} {\bibfnamefont {A.}~\bibnamefont {Paetznick}},\ and\ \bibinfo {author} {\bibfnamefont {C.}~\bibnamefont {Petit}},\ }\bibfield  {title} {\bibinfo {title} {Shorter quantum circuits via single-qubit gate approximation},\ }\href {https://doi.org/10.22331/q-2023-12-18-1208} {\bibfield  {journal} {\bibinfo  {journal} {Quantum}\ }\textbf {\bibinfo {volume} {7}},\ \bibinfo {pages} {1208} (\bibinfo {year} {2023})}\BibitemShut {NoStop}%
\bibitem [{\citenamefont {Morisaki}\ \emph {et~al.}(2025)\citenamefont {Morisaki}, \citenamefont {Sano},\ and\ \citenamefont {Akibue}}]{morisaki2025optimal}%
  \BibitemOpen
  \bibfield  {author} {\bibinfo {author} {\bibfnamefont {H.}~\bibnamefont {Morisaki}}, \bibinfo {author} {\bibfnamefont {K.}~\bibnamefont {Sano}},\ and\ \bibinfo {author} {\bibfnamefont {S.}~\bibnamefont {Akibue}},\ }\bibfield  {title} {\bibinfo {title} {Optimal ancilla-free clifford+ t synthesis for general single-qubit unitaries},\ }\bibfield  {journal} {\bibinfo  {journal} {arXiv preprint arXiv:2510.05816}\ }\href {https://doi.org/10.48550/arXiv.2510.05816} {10.48550/arXiv.2510.05816} (\bibinfo {year} {2025})\BibitemShut {NoStop}%
\bibitem [{\citenamefont {Bravyi}\ \emph {et~al.}(2021)\citenamefont {Bravyi}, \citenamefont {Shaydulin}, \citenamefont {Hu},\ and\ \citenamefont {Maslov}}]{bravyi2021clifford}%
  \BibitemOpen
  \bibfield  {author} {\bibinfo {author} {\bibfnamefont {S.}~\bibnamefont {Bravyi}}, \bibinfo {author} {\bibfnamefont {R.}~\bibnamefont {Shaydulin}}, \bibinfo {author} {\bibfnamefont {S.}~\bibnamefont {Hu}},\ and\ \bibinfo {author} {\bibfnamefont {D.}~\bibnamefont {Maslov}},\ }\bibfield  {title} {\bibinfo {title} {Clifford circuit optimization with templates and symbolic pauli gates},\ }\href {https://doi.org/10.22331/q-2021-11-16-580} {\bibfield  {journal} {\bibinfo  {journal} {Quantum}\ }\textbf {\bibinfo {volume} {5}},\ \bibinfo {pages} {580} (\bibinfo {year} {2021})}\BibitemShut {NoStop}%
\bibitem [{\citenamefont {Yang}\ and\ \citenamefont {Rall}(2024)}]{yang2024harnessing}%
  \BibitemOpen
  \bibfield  {author} {\bibinfo {author} {\bibfnamefont {W.}~\bibnamefont {Yang}}\ and\ \bibinfo {author} {\bibfnamefont {P.}~\bibnamefont {Rall}},\ }\bibfield  {title} {\bibinfo {title} {Harnessing the power of long-range entanglement for clifford circuit synthesis},\ }\href {https://doi.org/10.1109/TQE.2024.3402085} {\bibfield  {journal} {\bibinfo  {journal} {IEEE Transactions on Quantum Engineering}\ }\textbf {\bibinfo {volume} {5}},\ \bibinfo {pages} {1} (\bibinfo {year} {2024})}\BibitemShut {NoStop}%
\bibitem [{\citenamefont {Peham}\ \emph {et~al.}(2023)\citenamefont {Peham}, \citenamefont {Brandl}, \citenamefont {Kueng}, \citenamefont {Wille},\ and\ \citenamefont {Burgholzer}}]{peham2023depth}%
  \BibitemOpen
  \bibfield  {author} {\bibinfo {author} {\bibfnamefont {T.}~\bibnamefont {Peham}}, \bibinfo {author} {\bibfnamefont {N.}~\bibnamefont {Brandl}}, \bibinfo {author} {\bibfnamefont {R.}~\bibnamefont {Kueng}}, \bibinfo {author} {\bibfnamefont {R.}~\bibnamefont {Wille}},\ and\ \bibinfo {author} {\bibfnamefont {L.}~\bibnamefont {Burgholzer}},\ }\bibfield  {title} {\bibinfo {title} {Depth-optimal synthesis of clifford circuits with sat solvers},\ }\href {https://doi.org/10.1109/QCE57702.2023.00095} {\bibfield  {journal} {\bibinfo  {journal} {2023 IEEE International Conference on Quantum Computing and Engineering (QCE)}\ }\textbf {\bibinfo {volume} {1}},\ \bibinfo {pages} {802} (\bibinfo {year} {2023})}\BibitemShut {NoStop}%
\bibitem [{\citenamefont {Webster}\ \emph {et~al.}(2025)\citenamefont {Webster}, \citenamefont {Koutsioumpas},\ and\ \citenamefont {Browne}}]{webster2025heuristic}%
  \BibitemOpen
  \bibfield  {author} {\bibinfo {author} {\bibfnamefont {M.}~\bibnamefont {Webster}}, \bibinfo {author} {\bibfnamefont {S.}~\bibnamefont {Koutsioumpas}},\ and\ \bibinfo {author} {\bibfnamefont {D.~E.}\ \bibnamefont {Browne}},\ }\bibfield  {title} {\bibinfo {title} {Heuristic and optimal synthesis of cnot and clifford circuits},\ }\bibfield  {journal} {\bibinfo  {journal} {arXiv preprint arXiv:2503.14660}\ }\href {https://doi.org/10.48550/arXiv.2503.14660} {10.48550/arXiv.2503.14660} (\bibinfo {year} {2025})\BibitemShut {NoStop}%
\bibitem [{\citenamefont {Huang}\ \emph {et~al.}(2020)\citenamefont {Huang}, \citenamefont {Kueng},\ and\ \citenamefont {Preskill}}]{huang2020predicting}%
  \BibitemOpen
  \bibfield  {author} {\bibinfo {author} {\bibfnamefont {H.-Y.}\ \bibnamefont {Huang}}, \bibinfo {author} {\bibfnamefont {R.}~\bibnamefont {Kueng}},\ and\ \bibinfo {author} {\bibfnamefont {J.}~\bibnamefont {Preskill}},\ }\bibfield  {title} {\bibinfo {title} {Predicting many properties of a quantum system from very few measurements},\ }\href {https://doi.org/10.1038/s41567-020-0932-7} {\bibfield  {journal} {\bibinfo  {journal} {Nature Physics}\ }\textbf {\bibinfo {volume} {16}},\ \bibinfo {pages} {1050} (\bibinfo {year} {2020})}\BibitemShut {NoStop}%
\bibitem [{\citenamefont {West}\ \emph {et~al.}(2024)\citenamefont {West}, \citenamefont {Mele}, \citenamefont {Larocca},\ and\ \citenamefont {Cerezo}}]{west2024real}%
  \BibitemOpen
  \bibfield  {author} {\bibinfo {author} {\bibfnamefont {M.}~\bibnamefont {West}}, \bibinfo {author} {\bibfnamefont {A.~A.}\ \bibnamefont {Mele}}, \bibinfo {author} {\bibfnamefont {M.}~\bibnamefont {Larocca}},\ and\ \bibinfo {author} {\bibfnamefont {M.}~\bibnamefont {Cerezo}},\ }\bibfield  {title} {\bibinfo {title} {Real classical shadows},\ }\bibfield  {journal} {\bibinfo  {journal} {arXiv preprint arXiv:2410.23481}\ }\href {https://doi.org/10.48550/arXiv.2410.23481} {10.48550/arXiv.2410.23481} (\bibinfo {year} {2024})\BibitemShut {NoStop}%
\bibitem [{\citenamefont {Aaronson}\ and\ \citenamefont {Gottesman}(2004)}]{aaronson2004improved}%
  \BibitemOpen
  \bibfield  {author} {\bibinfo {author} {\bibfnamefont {S.}~\bibnamefont {Aaronson}}\ and\ \bibinfo {author} {\bibfnamefont {D.}~\bibnamefont {Gottesman}},\ }\bibfield  {title} {\bibinfo {title} {Improved simulation of stabilizer circuits},\ }\href {https://doi.org/10.1103/PhysRevA.70.052328} {\bibfield  {journal} {\bibinfo  {journal} {Physical Review A}\ }\textbf {\bibinfo {volume} {70}},\ \bibinfo {pages} {052328} (\bibinfo {year} {2004})}\BibitemShut {NoStop}%
\bibitem [{\citenamefont {Gidney}(2021)}]{gidney2021stim}%
  \BibitemOpen
  \bibfield  {author} {\bibinfo {author} {\bibfnamefont {C.}~\bibnamefont {Gidney}},\ }\bibfield  {title} {\bibinfo {title} {Stim: a fast stabilizer circuit simulator},\ }\href {https://doi.org/10.22331/q-2021-07-06-497} {\bibfield  {journal} {\bibinfo  {journal} {Quantum}\ }\textbf {\bibinfo {volume} {5}},\ \bibinfo {pages} {497} (\bibinfo {year} {2021})}\BibitemShut {NoStop}%
\bibitem [{\citenamefont {Gouzien}\ and\ \citenamefont {Sangouard}(2025)}]{gouzien2025provably}%
  \BibitemOpen
  \bibfield  {author} {\bibinfo {author} {\bibfnamefont {{\'E}.}~\bibnamefont {Gouzien}}\ and\ \bibinfo {author} {\bibfnamefont {N.}~\bibnamefont {Sangouard}},\ }\bibfield  {title} {\bibinfo {title} {Provably optimal exact gate synthesis from a discrete gate set},\ }\bibfield  {journal} {\bibinfo  {journal} {arXiv preprint arXiv:2503.15452}\ }\href {https://doi.org/10.48550/arXiv.2503.15452} {10.48550/arXiv.2503.15452} (\bibinfo {year} {2025})\BibitemShut {NoStop}%
\bibitem [{\citenamefont {Knill}\ \emph {et~al.}(2008)\citenamefont {Knill}, \citenamefont {Leibfried}, \citenamefont {Reichle}, \citenamefont {Britton}, \citenamefont {Blakestad}, \citenamefont {Jost}, \citenamefont {Langer}, \citenamefont {Ozeri}, \citenamefont {Seidelin},\ and\ \citenamefont {Wineland}}]{knill2008randomized}%
  \BibitemOpen
  \bibfield  {author} {\bibinfo {author} {\bibfnamefont {E.}~\bibnamefont {Knill}}, \bibinfo {author} {\bibfnamefont {D.}~\bibnamefont {Leibfried}}, \bibinfo {author} {\bibfnamefont {R.}~\bibnamefont {Reichle}}, \bibinfo {author} {\bibfnamefont {J.}~\bibnamefont {Britton}}, \bibinfo {author} {\bibfnamefont {R.~B.}\ \bibnamefont {Blakestad}}, \bibinfo {author} {\bibfnamefont {J.~D.}\ \bibnamefont {Jost}}, \bibinfo {author} {\bibfnamefont {C.}~\bibnamefont {Langer}}, \bibinfo {author} {\bibfnamefont {R.}~\bibnamefont {Ozeri}}, \bibinfo {author} {\bibfnamefont {S.}~\bibnamefont {Seidelin}},\ and\ \bibinfo {author} {\bibfnamefont {D.~J.}\ \bibnamefont {Wineland}},\ }\bibfield  {title} {\bibinfo {title} {Randomized benchmarking of quantum gates},\ }\href {https://doi.org/10.1103/PhysRevA.77.012307} {\bibfield  {journal} {\bibinfo  {journal} {Physical Review A}\ }\textbf {\bibinfo {volume} {77}},\ \bibinfo {pages} {012307} (\bibinfo {year} {2008})}\BibitemShut {NoStop}%
\bibitem [{\citenamefont {Valiant}(2001)}]{valiant2001quantum}%
  \BibitemOpen
  \bibfield  {author} {\bibinfo {author} {\bibfnamefont {L.~G.}\ \bibnamefont {Valiant}},\ }\bibfield  {title} {\bibinfo {title} {Quantum computers that can be simulated classically in polynomial time},\ }in\ \href {https://doi.org/10.1145/380752.380785} {\emph {\bibinfo {booktitle} {Proceedings of the thirty-third annual ACM symposium on Theory of computing}}}\ (\bibinfo {year} {2001})\ pp.\ \bibinfo {pages} {114--123}\BibitemShut {NoStop}%
\bibitem [{\citenamefont {Knill}(2001)}]{knill2001fermionic}%
  \BibitemOpen
  \bibfield  {author} {\bibinfo {author} {\bibfnamefont {E.}~\bibnamefont {Knill}},\ }\bibfield  {title} {\bibinfo {title} {Fermionic linear optics and matchgates},\ }\bibfield  {journal} {\bibinfo  {journal} {arXiv preprint arXiv:quant-ph/0108033}\ }\href {https://doi.org/10.48550/arXiv.quant-ph/0108033} {10.48550/arXiv.quant-ph/0108033} (\bibinfo {year} {2001})\BibitemShut {NoStop}%
\bibitem [{\citenamefont {Terhal}\ and\ \citenamefont {DiVincenzo}(2002)}]{terhal2002classical}%
  \BibitemOpen
  \bibfield  {author} {\bibinfo {author} {\bibfnamefont {B.~M.}\ \bibnamefont {Terhal}}\ and\ \bibinfo {author} {\bibfnamefont {D.~P.}\ \bibnamefont {DiVincenzo}},\ }\bibfield  {title} {\bibinfo {title} {Classical simulation of noninteracting-fermion quantum circuits},\ }\href {https://doi.org/10.1103/PhysRevA.65.032325} {\bibfield  {journal} {\bibinfo  {journal} {Physical Review A}\ }\textbf {\bibinfo {volume} {65}},\ \bibinfo {pages} {032325} (\bibinfo {year} {2002})}\BibitemShut {NoStop}%
\bibitem [{\citenamefont {DiVincenzo}\ and\ \citenamefont {Terhal}(2005)}]{divincenzo2005fermionic}%
  \BibitemOpen
  \bibfield  {author} {\bibinfo {author} {\bibfnamefont {D.~P.}\ \bibnamefont {DiVincenzo}}\ and\ \bibinfo {author} {\bibfnamefont {B.~M.}\ \bibnamefont {Terhal}},\ }\bibfield  {title} {\bibinfo {title} {Fermionic linear optics revisited},\ }\href {https://doi.org/10.1007/s10701-005-8657-0} {\bibfield  {journal} {\bibinfo  {journal} {Foundations of Physics}\ }\textbf {\bibinfo {volume} {35}},\ \bibinfo {pages} {1967} (\bibinfo {year} {2005})}\BibitemShut {NoStop}%
\bibitem [{\citenamefont {Bravyi}(2005)}]{bravyi2004lagrangian}%
  \BibitemOpen
  \bibfield  {author} {\bibinfo {author} {\bibfnamefont {S.}~\bibnamefont {Bravyi}},\ }\bibfield  {title} {\bibinfo {title} {Lagrangian representation for fermionic linear optics},\ }\href {https://doi.org/10.5555/2011637.2011640} {\bibfield  {journal} {\bibinfo  {journal} {Quantum Info. Comput.}\ }\textbf {\bibinfo {volume} {5}},\ \bibinfo {pages} {216–238} (\bibinfo {year} {2005})}\BibitemShut {NoStop}%
\bibitem [{\citenamefont {Jozsa}\ and\ \citenamefont {Miyake}(2008)}]{jozsa2008matchgates}%
  \BibitemOpen
  \bibfield  {author} {\bibinfo {author} {\bibfnamefont {R.}~\bibnamefont {Jozsa}}\ and\ \bibinfo {author} {\bibfnamefont {A.}~\bibnamefont {Miyake}},\ }\bibfield  {title} {\bibinfo {title} {Matchgates and classical simulation of quantum circuits},\ }\href {https://doi.org/10.1098/rspa.2008.0189} {\bibfield  {journal} {\bibinfo  {journal} {Proceedings of the Royal Society A: Mathematical, Physical and Engineering Sciences}\ }\textbf {\bibinfo {volume} {464}},\ \bibinfo {pages} {3089} (\bibinfo {year} {2008})}\BibitemShut {NoStop}%
\bibitem [{\citenamefont {Jozsa}\ \emph {et~al.}(2010)\citenamefont {Jozsa}, \citenamefont {Kraus}, \citenamefont {Miyake},\ and\ \citenamefont {Watrous}}]{jozsa2010matchgate}%
  \BibitemOpen
  \bibfield  {author} {\bibinfo {author} {\bibfnamefont {R.}~\bibnamefont {Jozsa}}, \bibinfo {author} {\bibfnamefont {B.}~\bibnamefont {Kraus}}, \bibinfo {author} {\bibfnamefont {A.}~\bibnamefont {Miyake}},\ and\ \bibinfo {author} {\bibfnamefont {J.}~\bibnamefont {Watrous}},\ }\bibfield  {title} {\bibinfo {title} {Matchgate and space-bounded quantum computations are equivalent},\ }\href {https://doi.org/10.1098/rspa.2009.0433} {\bibfield  {journal} {\bibinfo  {journal} {Proceedings of the Royal Society A: Mathematical, Physical and Engineering Sciences}\ }\textbf {\bibinfo {volume} {466}},\ \bibinfo {pages} {809} (\bibinfo {year} {2010})}\BibitemShut {NoStop}%
\bibitem [{\citenamefont {Ramelow}\ \emph {et~al.}(2010)\citenamefont {Ramelow}, \citenamefont {Fedrizzi}, \citenamefont {Steinberg},\ and\ \citenamefont {White}}]{ramelow2010matchgate}%
  \BibitemOpen
  \bibfield  {author} {\bibinfo {author} {\bibfnamefont {S.}~\bibnamefont {Ramelow}}, \bibinfo {author} {\bibfnamefont {A.}~\bibnamefont {Fedrizzi}}, \bibinfo {author} {\bibfnamefont {A.~M.}\ \bibnamefont {Steinberg}},\ and\ \bibinfo {author} {\bibfnamefont {A.}~\bibnamefont {White}},\ }\bibfield  {title} {\bibinfo {title} {Matchgate quantum computing and non-local process analysis},\ }\href {https://doi.org/10.1088/1367-2630/12/8/083027} {\bibfield  {journal} {\bibinfo  {journal} {New Journal of Physics}\ }\textbf {\bibinfo {volume} {12}},\ \bibinfo {pages} {083027} (\bibinfo {year} {2010})}\BibitemShut {NoStop}%
\bibitem [{\citenamefont {Brod}\ and\ \citenamefont {Galvao}(2011)}]{brod2011extending}%
  \BibitemOpen
  \bibfield  {author} {\bibinfo {author} {\bibfnamefont {D.~J.}\ \bibnamefont {Brod}}\ and\ \bibinfo {author} {\bibfnamefont {E.~F.}\ \bibnamefont {Galvao}},\ }\bibfield  {title} {\bibinfo {title} {Extending matchgates into universal quantum computation},\ }\href {https://doi.org/10.1103/PhysRevA.84.022310} {\bibfield  {journal} {\bibinfo  {journal} {Physical Review A—Atomic, Molecular, and Optical Physics}\ }\textbf {\bibinfo {volume} {84}},\ \bibinfo {pages} {022310} (\bibinfo {year} {2011})}\BibitemShut {NoStop}%
\bibitem [{\citenamefont {Brod}\ and\ \citenamefont {Childs}(2014)}]{brod2014computational}%
  \BibitemOpen
  \bibfield  {author} {\bibinfo {author} {\bibfnamefont {D.~J.}\ \bibnamefont {Brod}}\ and\ \bibinfo {author} {\bibfnamefont {A.~M.}\ \bibnamefont {Childs}},\ }\bibfield  {title} {\bibinfo {title} {The computational power of matchgates and the xy interaction on arbitrary graphs},\ }\href {https://doi.org/10.26421/QIC14.11-12-1} {\bibfield  {journal} {\bibinfo  {journal} {Quantum Information and Computation}\ }\textbf {\bibinfo {volume} {14}},\ \bibinfo {pages} {901} (\bibinfo {year} {2014})}\BibitemShut {NoStop}%
\bibitem [{\citenamefont {Brod}(2016)}]{brod2016efficient}%
  \BibitemOpen
  \bibfield  {author} {\bibinfo {author} {\bibfnamefont {D.~J.}\ \bibnamefont {Brod}},\ }\bibfield  {title} {\bibinfo {title} {Efficient classical simulation of matchgate circuits with generalized inputs and measurements},\ }\href {https://doi.org/10.1103/PhysRevA.93.062332} {\bibfield  {journal} {\bibinfo  {journal} {Physical Review A}\ }\textbf {\bibinfo {volume} {93}},\ \bibinfo {pages} {062332} (\bibinfo {year} {2016})}\BibitemShut {NoStop}%
\bibitem [{\citenamefont {Oszmaniec}\ and\ \citenamefont {Zimbor{\'a}s}(2017)}]{oszmaniec2017universal}%
  \BibitemOpen
  \bibfield  {author} {\bibinfo {author} {\bibfnamefont {M.}~\bibnamefont {Oszmaniec}}\ and\ \bibinfo {author} {\bibfnamefont {Z.}~\bibnamefont {Zimbor{\'a}s}},\ }\bibfield  {title} {\bibinfo {title} {Universal extensions of restricted classes of quantum operations},\ }\href {https://doi.org/10.1103/PhysRevLett.119.220502} {\bibfield  {journal} {\bibinfo  {journal} {Physical review letters}\ }\textbf {\bibinfo {volume} {119}},\ \bibinfo {pages} {220502} (\bibinfo {year} {2017})}\BibitemShut {NoStop}%
\bibitem [{\citenamefont {Zhao}\ \emph {et~al.}(2021)\citenamefont {Zhao}, \citenamefont {Rubin},\ and\ \citenamefont {Miyake}}]{zhao2021fermionic}%
  \BibitemOpen
  \bibfield  {author} {\bibinfo {author} {\bibfnamefont {A.}~\bibnamefont {Zhao}}, \bibinfo {author} {\bibfnamefont {N.~C.}\ \bibnamefont {Rubin}},\ and\ \bibinfo {author} {\bibfnamefont {A.}~\bibnamefont {Miyake}},\ }\bibfield  {title} {\bibinfo {title} {Fermionic partial tomography via classical shadows},\ }\href {https://doi.org/10.1103/PhysRevLett.127.110504} {\bibfield  {journal} {\bibinfo  {journal} {Physical Review Letters}\ }\textbf {\bibinfo {volume} {127}},\ \bibinfo {pages} {110504} (\bibinfo {year} {2021})}\BibitemShut {NoStop}%
\bibitem [{\citenamefont {Wan}\ \emph {et~al.}(2023)\citenamefont {Wan}, \citenamefont {Huggins}, \citenamefont {Lee},\ and\ \citenamefont {Babbush}}]{wan2022matchgate}%
  \BibitemOpen
  \bibfield  {author} {\bibinfo {author} {\bibfnamefont {K.}~\bibnamefont {Wan}}, \bibinfo {author} {\bibfnamefont {W.~J.}\ \bibnamefont {Huggins}}, \bibinfo {author} {\bibfnamefont {J.}~\bibnamefont {Lee}},\ and\ \bibinfo {author} {\bibfnamefont {R.}~\bibnamefont {Babbush}},\ }\bibfield  {title} {\bibinfo {title} {Matchgate shadows for fermionic quantum simulation},\ }\href {https://doi.org/10.1007/s00220-023-04844-0} {\bibfield  {journal} {\bibinfo  {journal} {Communications in Mathematical Physics}\ }\textbf {\bibinfo {volume} {404}},\ \bibinfo {pages} {629} (\bibinfo {year} {2023})}\BibitemShut {NoStop}%
\bibitem [{\citenamefont {Helsen}\ \emph {et~al.}(2022)\citenamefont {Helsen}, \citenamefont {Nezami}, \citenamefont {Reagor},\ and\ \citenamefont {Walter}}]{helsen2022matchgate}%
  \BibitemOpen
  \bibfield  {author} {\bibinfo {author} {\bibfnamefont {J.}~\bibnamefont {Helsen}}, \bibinfo {author} {\bibfnamefont {S.}~\bibnamefont {Nezami}}, \bibinfo {author} {\bibfnamefont {M.}~\bibnamefont {Reagor}},\ and\ \bibinfo {author} {\bibfnamefont {M.}~\bibnamefont {Walter}},\ }\bibfield  {title} {\bibinfo {title} {Matchgate benchmarking: Scalable benchmarking of a continuous family of many-qubit gates},\ }\href {https://doi.org/10.22331/q-2022-02-21-657} {\bibfield  {journal} {\bibinfo  {journal} {Quantum}\ }\textbf {\bibinfo {volume} {6}},\ \bibinfo {pages} {657} (\bibinfo {year} {2022})}\BibitemShut {NoStop}%
\bibitem [{\citenamefont {Diaz}\ \emph {et~al.}(2023)\citenamefont {Diaz}, \citenamefont {Garc{\'\i}a-Mart{\'\i}n}, \citenamefont {Kazi}, \citenamefont {Larocca},\ and\ \citenamefont {Cerezo}}]{diaz2023showcasing}%
  \BibitemOpen
  \bibfield  {author} {\bibinfo {author} {\bibfnamefont {N.~L.}\ \bibnamefont {Diaz}}, \bibinfo {author} {\bibfnamefont {D.}~\bibnamefont {Garc{\'\i}a-Mart{\'\i}n}}, \bibinfo {author} {\bibfnamefont {S.}~\bibnamefont {Kazi}}, \bibinfo {author} {\bibfnamefont {M.}~\bibnamefont {Larocca}},\ and\ \bibinfo {author} {\bibfnamefont {M.}~\bibnamefont {Cerezo}},\ }\bibfield  {title} {\bibinfo {title} {Showcasing a barren plateau theory beyond the dynamical lie algebra},\ }\href {https://arxiv.org/abs/2310.11505} {\bibfield  {journal} {\bibinfo  {journal} {arXiv preprint arXiv:2310.11505}\ } (\bibinfo {year} {2023})}\BibitemShut {NoStop}%
\bibitem [{\citenamefont {Mele}\ and\ \citenamefont {Herasymenko}(2025)}]{mele2024efficient}%
  \BibitemOpen
  \bibfield  {author} {\bibinfo {author} {\bibfnamefont {A.~A.}\ \bibnamefont {Mele}}\ and\ \bibinfo {author} {\bibfnamefont {Y.}~\bibnamefont {Herasymenko}},\ }\bibfield  {title} {\bibinfo {title} {Efficient learning of quantum states prepared with few fermionic non-gaussian gates},\ }\href {https://doi.org/10.1103/PRXQuantum.6.010319} {\bibfield  {journal} {\bibinfo  {journal} {PRX Quantum}\ }\textbf {\bibinfo {volume} {6}},\ \bibinfo {pages} {010319} (\bibinfo {year} {2025})}\BibitemShut {NoStop}%
\bibitem [{\citenamefont {Stroeks}\ \emph {et~al.}(2024)\citenamefont {Stroeks}, \citenamefont {Lenterman}, \citenamefont {Terhal},\ and\ \citenamefont {Herasymenko}}]{stroeks2024solving}%
  \BibitemOpen
  \bibfield  {author} {\bibinfo {author} {\bibfnamefont {M.}~\bibnamefont {Stroeks}}, \bibinfo {author} {\bibfnamefont {D.}~\bibnamefont {Lenterman}}, \bibinfo {author} {\bibfnamefont {B.}~\bibnamefont {Terhal}},\ and\ \bibinfo {author} {\bibfnamefont {Y.}~\bibnamefont {Herasymenko}},\ }\bibfield  {title} {\bibinfo {title} {Solving free fermion problems on a quantum computer},\ }\bibfield  {journal} {\bibinfo  {journal} {arXiv preprint arXiv:2409.04550}\ }\href {https://doi.org/10.48550/arXiv.2409.04550} {10.48550/arXiv.2409.04550} (\bibinfo {year} {2024})\BibitemShut {NoStop}%
\bibitem [{\citenamefont {Veitch}\ \emph {et~al.}(2014)\citenamefont {Veitch}, \citenamefont {Mousavian}, \citenamefont {Gottesman},\ and\ \citenamefont {Emerson}}]{veitch2014resource}%
  \BibitemOpen
  \bibfield  {author} {\bibinfo {author} {\bibfnamefont {V.}~\bibnamefont {Veitch}}, \bibinfo {author} {\bibfnamefont {S.~H.}\ \bibnamefont {Mousavian}}, \bibinfo {author} {\bibfnamefont {D.}~\bibnamefont {Gottesman}},\ and\ \bibinfo {author} {\bibfnamefont {J.}~\bibnamefont {Emerson}},\ }\bibfield  {title} {\bibinfo {title} {The resource theory of stabilizer quantum computation},\ }\href {https://doi.org/10.1088/1367-2630/16/1/013009} {\bibfield  {journal} {\bibinfo  {journal} {New Journal of Physics}\ }\textbf {\bibinfo {volume} {16}},\ \bibinfo {pages} {013009} (\bibinfo {year} {2014})}\BibitemShut {NoStop}%
\bibitem [{\citenamefont {Howard}\ and\ \citenamefont {Campbell}(2017)}]{howard2017application}%
  \BibitemOpen
  \bibfield  {author} {\bibinfo {author} {\bibfnamefont {M.}~\bibnamefont {Howard}}\ and\ \bibinfo {author} {\bibfnamefont {E.}~\bibnamefont {Campbell}},\ }\bibfield  {title} {\bibinfo {title} {Application of a resource theory for magic states to fault-tolerant quantum computing},\ }\href {https://doi.org/10.1103/PhysRevLett.118.090501} {\bibfield  {journal} {\bibinfo  {journal} {Physical review letters}\ }\textbf {\bibinfo {volume} {118}},\ \bibinfo {pages} {090501} (\bibinfo {year} {2017})}\BibitemShut {NoStop}%
\bibitem [{\citenamefont {Chitambar}\ and\ \citenamefont {Gour}(2019)}]{chitambar2019quantum}%
  \BibitemOpen
  \bibfield  {author} {\bibinfo {author} {\bibfnamefont {E.}~\bibnamefont {Chitambar}}\ and\ \bibinfo {author} {\bibfnamefont {G.}~\bibnamefont {Gour}},\ }\bibfield  {title} {\bibinfo {title} {Quantum resource theories},\ }\href {https://doi.org/10.1103/RevModPhys.91.025001} {\bibfield  {journal} {\bibinfo  {journal} {Reviews of modern physics}\ }\textbf {\bibinfo {volume} {91}},\ \bibinfo {pages} {025001} (\bibinfo {year} {2019})}\BibitemShut {NoStop}%
\bibitem [{\citenamefont {Leone}\ \emph {et~al.}(2022)\citenamefont {Leone}, \citenamefont {Oliviero},\ and\ \citenamefont {Hamma}}]{leone2022stabilizer}%
  \BibitemOpen
  \bibfield  {author} {\bibinfo {author} {\bibfnamefont {L.}~\bibnamefont {Leone}}, \bibinfo {author} {\bibfnamefont {S.~F.}\ \bibnamefont {Oliviero}},\ and\ \bibinfo {author} {\bibfnamefont {A.}~\bibnamefont {Hamma}},\ }\bibfield  {title} {\bibinfo {title} {Stabilizer r{\'e}nyi entropy},\ }\href {https://doi.org/10.1103/PhysRevLett.128.050402} {\bibfield  {journal} {\bibinfo  {journal} {Physical Review Letters}\ }\textbf {\bibinfo {volume} {128}},\ \bibinfo {pages} {050402} (\bibinfo {year} {2022})}\BibitemShut {NoStop}%
\bibitem [{\citenamefont {Chen}\ \emph {et~al.}(2024)\citenamefont {Chen}, \citenamefont {Nielsen}, \citenamefont {Ebert}, \citenamefont {Inlek}, \citenamefont {Wright}, \citenamefont {Chaplin}, \citenamefont {Maksymov}, \citenamefont {P{\'a}ez}, \citenamefont {Poudel}, \citenamefont {Maunz} \emph {et~al.}}]{chen2024benchmarking}%
  \BibitemOpen
  \bibfield  {author} {\bibinfo {author} {\bibfnamefont {J.-S.}\ \bibnamefont {Chen}}, \bibinfo {author} {\bibfnamefont {E.}~\bibnamefont {Nielsen}}, \bibinfo {author} {\bibfnamefont {M.}~\bibnamefont {Ebert}}, \bibinfo {author} {\bibfnamefont {V.}~\bibnamefont {Inlek}}, \bibinfo {author} {\bibfnamefont {K.}~\bibnamefont {Wright}}, \bibinfo {author} {\bibfnamefont {V.}~\bibnamefont {Chaplin}}, \bibinfo {author} {\bibfnamefont {A.}~\bibnamefont {Maksymov}}, \bibinfo {author} {\bibfnamefont {E.}~\bibnamefont {P{\'a}ez}}, \bibinfo {author} {\bibfnamefont {A.}~\bibnamefont {Poudel}}, \bibinfo {author} {\bibfnamefont {P.}~\bibnamefont {Maunz}}, \emph {et~al.},\ }\bibfield  {title} {\bibinfo {title} {Benchmarking a trapped-ion quantum computer with 30 qubits},\ }\href {https://doi.org/10.22331/q-2024-11-07-1516} {\bibfield  {journal} {\bibinfo  {journal} {Quantum}\ }\textbf {\bibinfo {volume} {8}},\ \bibinfo {pages} {1516} (\bibinfo {year} {2024})}\BibitemShut {NoStop}%
\bibitem [{\citenamefont {Wurtz}\ \emph {et~al.}(2023)\citenamefont {Wurtz}, \citenamefont {Bylinskii}, \citenamefont {Braverman}, \citenamefont {Amato-Grill}, \citenamefont {Cantu}, \citenamefont {Huber}, \citenamefont {Lukin}, \citenamefont {Liu}, \citenamefont {Weinberg}, \citenamefont {Long} \emph {et~al.}}]{wurtz2023aquila}%
  \BibitemOpen
  \bibfield  {author} {\bibinfo {author} {\bibfnamefont {J.}~\bibnamefont {Wurtz}}, \bibinfo {author} {\bibfnamefont {A.}~\bibnamefont {Bylinskii}}, \bibinfo {author} {\bibfnamefont {B.}~\bibnamefont {Braverman}}, \bibinfo {author} {\bibfnamefont {J.}~\bibnamefont {Amato-Grill}}, \bibinfo {author} {\bibfnamefont {S.~H.}\ \bibnamefont {Cantu}}, \bibinfo {author} {\bibfnamefont {F.}~\bibnamefont {Huber}}, \bibinfo {author} {\bibfnamefont {A.}~\bibnamefont {Lukin}}, \bibinfo {author} {\bibfnamefont {F.}~\bibnamefont {Liu}}, \bibinfo {author} {\bibfnamefont {P.}~\bibnamefont {Weinberg}}, \bibinfo {author} {\bibfnamefont {J.}~\bibnamefont {Long}}, \emph {et~al.},\ }\bibfield  {title} {\bibinfo {title} {Aquila: Quera's 256-qubit neutral-atom quantum computer},\ }\href {https://arxiv.org/pdf/2306.11727.pdf} {\bibfield  {journal} {\bibinfo  {journal} {arXiv preprint arXiv:2306.11727}\ } (\bibinfo {year} {2023})}\BibitemShut {NoStop}%
\bibitem [{\citenamefont {Verstraete}\ \emph {et~al.}(2009)\citenamefont {Verstraete}, \citenamefont {Cirac},\ and\ \citenamefont {Latorre}}]{verstraete2009quantum}%
  \BibitemOpen
  \bibfield  {author} {\bibinfo {author} {\bibfnamefont {F.}~\bibnamefont {Verstraete}}, \bibinfo {author} {\bibfnamefont {J.~I.}\ \bibnamefont {Cirac}},\ and\ \bibinfo {author} {\bibfnamefont {J.~I.}\ \bibnamefont {Latorre}},\ }\bibfield  {title} {\bibinfo {title} {Quantum circuits for strongly correlated quantum systems},\ }\href {https://doi.org/10.1103/PhysRevA.79.032316} {\bibfield  {journal} {\bibinfo  {journal} {Physical Review A}\ }\textbf {\bibinfo {volume} {79}},\ \bibinfo {pages} {032316} (\bibinfo {year} {2009})}\BibitemShut {NoStop}%
\bibitem [{\citenamefont {Crooks}()}]{quantumgates}%
  \BibitemOpen
  \bibfield  {author} {\bibinfo {author} {\bibfnamefont {G.~E.}\ \bibnamefont {Crooks}},\ }\href {https://github.com/gecrooks/on_gates} {\bibinfo {title} {{Quantum Gates}}}\BibitemShut {NoStop}%
\bibitem [{\citenamefont {Boykin}\ \emph {et~al.}(2000)\citenamefont {Boykin}, \citenamefont {Mor}, \citenamefont {Pulver}, \citenamefont {Roychowdhury},\ and\ \citenamefont {Vatan}}]{boykin2000new}%
  \BibitemOpen
  \bibfield  {author} {\bibinfo {author} {\bibfnamefont {P.~O.}\ \bibnamefont {Boykin}}, \bibinfo {author} {\bibfnamefont {T.}~\bibnamefont {Mor}}, \bibinfo {author} {\bibfnamefont {M.}~\bibnamefont {Pulver}}, \bibinfo {author} {\bibfnamefont {V.}~\bibnamefont {Roychowdhury}},\ and\ \bibinfo {author} {\bibfnamefont {F.}~\bibnamefont {Vatan}},\ }\bibfield  {title} {\bibinfo {title} {A new universal and fault-tolerant quantum basis},\ }\href {https://doi.org/10.1016/S0020-0190(00)00084-3} {\bibfield  {journal} {\bibinfo  {journal} {Information Processing Letters}\ }\textbf {\bibinfo {volume} {75}},\ \bibinfo {pages} {101} (\bibinfo {year} {2000})}\BibitemShut {NoStop}%
\bibitem [{\citenamefont {Barenco}\ \emph {et~al.}(1995)\citenamefont {Barenco}, \citenamefont {Bennett}, \citenamefont {Cleve}, \citenamefont {DiVincenzo}, \citenamefont {Margolus}, \citenamefont {Shor}, \citenamefont {Sleator}, \citenamefont {Smolin},\ and\ \citenamefont {Weinfurter}}]{barenco1995elementary}%
  \BibitemOpen
  \bibfield  {author} {\bibinfo {author} {\bibfnamefont {A.}~\bibnamefont {Barenco}}, \bibinfo {author} {\bibfnamefont {C.~H.}\ \bibnamefont {Bennett}}, \bibinfo {author} {\bibfnamefont {R.}~\bibnamefont {Cleve}}, \bibinfo {author} {\bibfnamefont {D.~P.}\ \bibnamefont {DiVincenzo}}, \bibinfo {author} {\bibfnamefont {N.}~\bibnamefont {Margolus}}, \bibinfo {author} {\bibfnamefont {P.}~\bibnamefont {Shor}}, \bibinfo {author} {\bibfnamefont {T.}~\bibnamefont {Sleator}}, \bibinfo {author} {\bibfnamefont {J.~A.}\ \bibnamefont {Smolin}},\ and\ \bibinfo {author} {\bibfnamefont {H.}~\bibnamefont {Weinfurter}},\ }\bibfield  {title} {\bibinfo {title} {Elementary gates for quantum computation},\ }\href {https://doi.org/10.1103/PhysRevA.52.3457} {\bibfield  {journal} {\bibinfo  {journal} {Physical review A}\ }\textbf {\bibinfo {volume} {52}},\ \bibinfo {pages} {3457} (\bibinfo {year} {1995})}\BibitemShut {NoStop}%
\bibitem [{\citenamefont {Gottesman}\ and\ \citenamefont {Chuang}(1999)}]{gottesman1999demonstrating}%
  \BibitemOpen
  \bibfield  {author} {\bibinfo {author} {\bibfnamefont {D.}~\bibnamefont {Gottesman}}\ and\ \bibinfo {author} {\bibfnamefont {I.~L.}\ \bibnamefont {Chuang}},\ }\bibfield  {title} {\bibinfo {title} {Demonstrating the viability of universal quantum computation using teleportation and single-qubit operations},\ }\href {https://doi.org/10.1038/46503} {\bibfield  {journal} {\bibinfo  {journal} {Nature}\ }\textbf {\bibinfo {volume} {402}},\ \bibinfo {pages} {390} (\bibinfo {year} {1999})}\BibitemShut {NoStop}%
\bibitem [{\citenamefont {Zhou}\ \emph {et~al.}(2000)\citenamefont {Zhou}, \citenamefont {Leung},\ and\ \citenamefont {Chuang}}]{zhou2000methodology}%
  \BibitemOpen
  \bibfield  {author} {\bibinfo {author} {\bibfnamefont {X.}~\bibnamefont {Zhou}}, \bibinfo {author} {\bibfnamefont {D.~W.}\ \bibnamefont {Leung}},\ and\ \bibinfo {author} {\bibfnamefont {I.~L.}\ \bibnamefont {Chuang}},\ }\bibfield  {title} {\bibinfo {title} {Methodology for quantum logic gate construction},\ }\href {https://doi.org/10.1103/PhysRevA.62.052316} {\bibfield  {journal} {\bibinfo  {journal} {Physical Review A}\ }\textbf {\bibinfo {volume} {62}},\ \bibinfo {pages} {052316} (\bibinfo {year} {2000})}\BibitemShut {NoStop}%
\bibitem [{\citenamefont {Knill}(2004)}]{knill2004fault}%
  \BibitemOpen
  \bibfield  {author} {\bibinfo {author} {\bibfnamefont {E.}~\bibnamefont {Knill}},\ }\bibfield  {title} {\bibinfo {title} {Fault-tolerant postselected quantum computation: Schemes},\ }\bibfield  {journal} {\bibinfo  {journal} {arXiv preprint quant-ph/0402171}\ }\href {https://doi.org/10.48550/arXiv.quant-ph/0402171} {10.48550/arXiv.quant-ph/0402171} (\bibinfo {year} {2004})\BibitemShut {NoStop}%
\bibitem [{\citenamefont {Kuperberg}(2023)}]{kuperberg2023breaking}%
  \BibitemOpen
  \bibfield  {author} {\bibinfo {author} {\bibfnamefont {G.}~\bibnamefont {Kuperberg}},\ }\bibfield  {title} {\bibinfo {title} {Breaking the cubic barrier in the solovay-kitaev algorithm},\ }\bibfield  {journal} {\bibinfo  {journal} {arXiv preprint arXiv:2306.13158}\ }\href {https://doi.org/10.48550/arXiv.2306.13158} {10.48550/arXiv.2306.13158} (\bibinfo {year} {2023})\BibitemShut {NoStop}%
\bibitem [{\citenamefont {Nielsen}\ and\ \citenamefont {Chuang}(2000)}]{nielsen2000quantum}%
  \BibitemOpen
  \bibfield  {author} {\bibinfo {author} {\bibfnamefont {M.~A.}\ \bibnamefont {Nielsen}}\ and\ \bibinfo {author} {\bibfnamefont {I.~L.}\ \bibnamefont {Chuang}},\ }\href@noop {} {\emph {\bibinfo {title} {Quantum Computation and Quantum Information}}}\ (\bibinfo  {publisher} {Cambridge University Press},\ \bibinfo {address} {Cambridge},\ \bibinfo {year} {2000})\BibitemShut {NoStop}%
\bibitem [{\citenamefont {Harrow}\ \emph {et~al.}(2002)\citenamefont {Harrow}, \citenamefont {Recht},\ and\ \citenamefont {Chuang}}]{harrow2002efficient}%
  \BibitemOpen
  \bibfield  {author} {\bibinfo {author} {\bibfnamefont {A.~W.}\ \bibnamefont {Harrow}}, \bibinfo {author} {\bibfnamefont {B.}~\bibnamefont {Recht}},\ and\ \bibinfo {author} {\bibfnamefont {I.~L.}\ \bibnamefont {Chuang}},\ }\bibfield  {title} {\bibinfo {title} {Efficient discrete approximations of quantum gates},\ }\href {https://doi.org/10.1063/1.1495899} {\bibfield  {journal} {\bibinfo  {journal} {Journal of Mathematical Physics}\ }\textbf {\bibinfo {volume} {43}},\ \bibinfo {pages} {4445} (\bibinfo {year} {2002})}\BibitemShut {NoStop}%
\bibitem [{\citenamefont {Bourgain}\ and\ \citenamefont {Gamburd}(2008)}]{bourgain2008spectral}%
  \BibitemOpen
  \bibfield  {author} {\bibinfo {author} {\bibfnamefont {J.}~\bibnamefont {Bourgain}}\ and\ \bibinfo {author} {\bibfnamefont {A.}~\bibnamefont {Gamburd}},\ }\bibfield  {title} {\bibinfo {title} {On the spectral gap for finitely-generated subgroups of su (2)},\ }\href {https://doi.org/10.1007/s00222-007-0072-z} {\bibfield  {journal} {\bibinfo  {journal} {Inventiones mathematicae}\ }\textbf {\bibinfo {volume} {171}},\ \bibinfo {pages} {83} (\bibinfo {year} {2008})}\BibitemShut {NoStop}%
\bibitem [{\citenamefont {Bourgain}\ and\ \citenamefont {Gamburd}(2010)}]{bourgain2010spectral}%
  \BibitemOpen
  \bibfield  {author} {\bibinfo {author} {\bibfnamefont {J.}~\bibnamefont {Bourgain}}\ and\ \bibinfo {author} {\bibfnamefont {A.}~\bibnamefont {Gamburd}},\ }\bibfield  {title} {\bibinfo {title} {Spectral gaps in su(d)},\ }\href {https://doi.org/10.1016/j.crma.2010.04.024} {\bibfield  {journal} {\bibinfo  {journal} {Comptes Rendus. Math{\'e}matique}\ }\textbf {\bibinfo {volume} {348}},\ \bibinfo {pages} {609} (\bibinfo {year} {2010})}\BibitemShut {NoStop}%
\bibitem [{\citenamefont {Bourgain}\ and\ \citenamefont {Gamburd}(2012)}]{bourgain2012spectral}%
  \BibitemOpen
  \bibfield  {author} {\bibinfo {author} {\bibfnamefont {J.}~\bibnamefont {Bourgain}}\ and\ \bibinfo {author} {\bibfnamefont {A.}~\bibnamefont {Gamburd}},\ }\bibfield  {title} {\bibinfo {title} {A spectral gap theorem in su(d)},\ }\bibfield  {journal} {\bibinfo  {journal} {Journal of the European Mathematical Society}\ }\textbf {\bibinfo {volume} {14}},\ \href {https://doi.org/10.4171/JEMS/337} {10.4171/JEMS/337} (\bibinfo {year} {2012})\BibitemShut {NoStop}%
\bibitem [{\citenamefont {Benoist}\ and\ \citenamefont {de~Saxc{\'e}}(2016)}]{benoist2016spectral}%
  \BibitemOpen
  \bibfield  {author} {\bibinfo {author} {\bibfnamefont {Y.}~\bibnamefont {Benoist}}\ and\ \bibinfo {author} {\bibfnamefont {N.}~\bibnamefont {de~Saxc{\'e}}},\ }\bibfield  {title} {\bibinfo {title} {A spectral gap theorem in simple lie groups},\ }\href {https://doi.org/10.1007/s00222-015-0636-2} {\bibfield  {journal} {\bibinfo  {journal} {Inventiones mathematicae}\ }\textbf {\bibinfo {volume} {205}},\ \bibinfo {pages} {337} (\bibinfo {year} {2016})}\BibitemShut {NoStop}%
\bibitem [{\citenamefont {Bouland}\ and\ \citenamefont {Giurgica-Tiron}(2021)}]{bouland2021efficient}%
  \BibitemOpen
  \bibfield  {author} {\bibinfo {author} {\bibfnamefont {A.}~\bibnamefont {Bouland}}\ and\ \bibinfo {author} {\bibfnamefont {T.}~\bibnamefont {Giurgica-Tiron}},\ }\bibfield  {title} {\bibinfo {title} {Efficient universal quantum compilation: An inverse-free solovay-kitaev algorithm},\ }\bibfield  {journal} {\bibinfo  {journal} {arXiv preprint arXiv:2112.02040}\ }\href {https://doi.org/10.48550/arXiv.2112.02040} {10.48550/arXiv.2112.02040} (\bibinfo {year} {2021})\BibitemShut {NoStop}%
\bibitem [{\citenamefont {Kliuchnikov}\ \emph {et~al.}(2013{\natexlab{b}})\citenamefont {Kliuchnikov}, \citenamefont {Maslov},\ and\ \citenamefont {Mosca}}]{kliuchnikov2013asymptotically}%
  \BibitemOpen
  \bibfield  {author} {\bibinfo {author} {\bibfnamefont {V.}~\bibnamefont {Kliuchnikov}}, \bibinfo {author} {\bibfnamefont {D.}~\bibnamefont {Maslov}},\ and\ \bibinfo {author} {\bibfnamefont {M.}~\bibnamefont {Mosca}},\ }\bibfield  {title} {\bibinfo {title} {Asymptotically optimal approximation of single qubit unitaries by clifford and t circuits using a constant number of ancillary qubits},\ }\href {https://doi.org/10.1103/PhysRevLett.110.190502} {\bibfield  {journal} {\bibinfo  {journal} {Physical review letters}\ }\textbf {\bibinfo {volume} {110}},\ \bibinfo {pages} {190502} (\bibinfo {year} {2013}{\natexlab{b}})}\BibitemShut {NoStop}%
\bibitem [{\citenamefont {Campbell}(2017)}]{campbell2017shorter}%
  \BibitemOpen
  \bibfield  {author} {\bibinfo {author} {\bibfnamefont {E.}~\bibnamefont {Campbell}},\ }\bibfield  {title} {\bibinfo {title} {Shorter gate sequences for quantum computing by mixing unitaries},\ }\href {https://doi.org/10.1103/PhysRevA.95.042306} {\bibfield  {journal} {\bibinfo  {journal} {Physical Review A}\ }\textbf {\bibinfo {volume} {95}},\ \bibinfo {pages} {042306} (\bibinfo {year} {2017})}\BibitemShut {NoStop}%
\bibitem [{\citenamefont {Di~Matteo}\ and\ \citenamefont {Mosca}(2016)}]{dimatteo2016parallelizing}%
  \BibitemOpen
  \bibfield  {author} {\bibinfo {author} {\bibfnamefont {O.}~\bibnamefont {Di~Matteo}}\ and\ \bibinfo {author} {\bibfnamefont {M.}~\bibnamefont {Mosca}},\ }\bibfield  {title} {\bibinfo {title} {Parallelizing quantum circuit synthesis},\ }\href {https://doi.org/10.1088/2058-9565/1/1/015003} {\bibfield  {journal} {\bibinfo  {journal} {Quantum Science and Technology}\ }\textbf {\bibinfo {volume} {1}},\ \bibinfo {pages} {015003} (\bibinfo {year} {2016})}\BibitemShut {NoStop}%
\bibitem [{\citenamefont {Low}\ \emph {et~al.}(2024)\citenamefont {Low}, \citenamefont {Kliuchnikov},\ and\ \citenamefont {Schaeffer}}]{low2024trading}%
  \BibitemOpen
  \bibfield  {author} {\bibinfo {author} {\bibfnamefont {G.~H.}\ \bibnamefont {Low}}, \bibinfo {author} {\bibfnamefont {V.}~\bibnamefont {Kliuchnikov}},\ and\ \bibinfo {author} {\bibfnamefont {L.}~\bibnamefont {Schaeffer}},\ }\bibfield  {title} {\bibinfo {title} {Trading t gates for dirty qubits in state preparation and unitary synthesis},\ }\href {https://doi.org/10.22331/q-2024-06-17-1375} {\bibfield  {journal} {\bibinfo  {journal} {Quantum}\ }\textbf {\bibinfo {volume} {8}},\ \bibinfo {pages} {1375} (\bibinfo {year} {2024})}\BibitemShut {NoStop}%
\bibitem [{\citenamefont {Tan}(2025)}]{tan2025unitary}%
  \BibitemOpen
  \bibfield  {author} {\bibinfo {author} {\bibfnamefont {X.}~\bibnamefont {Tan}},\ }\bibfield  {title} {\bibinfo {title} {Unitary synthesis with fewer t gates},\ }\bibfield  {journal} {\bibinfo  {journal} {arXiv preprint arXiv:2509.25702}\ }\href {https://doi.org/10.48550/arXiv.2509.25702} {10.48550/arXiv.2509.25702} (\bibinfo {year} {2025})\BibitemShut {NoStop}%
\bibitem [{\citenamefont {Gosset}\ \emph {et~al.}(2024)\citenamefont {Gosset}, \citenamefont {Kothari},\ and\ \citenamefont {Wu}}]{gosset2024quantum}%
  \BibitemOpen
  \bibfield  {author} {\bibinfo {author} {\bibfnamefont {D.}~\bibnamefont {Gosset}}, \bibinfo {author} {\bibfnamefont {R.}~\bibnamefont {Kothari}},\ and\ \bibinfo {author} {\bibfnamefont {K.}~\bibnamefont {Wu}},\ }\bibfield  {title} {\bibinfo {title} {Quantum state preparation with optimal t-count},\ }\bibfield  {journal} {\bibinfo  {journal} {arXiv preprint arXiv:2411.04790}\ }\href {https://doi.org/10.48550/arXiv.2411.04790} {10.48550/arXiv.2411.04790} (\bibinfo {year} {2024})\BibitemShut {NoStop}%
\bibitem [{\citenamefont {Amy}\ \emph {et~al.}(2020)\citenamefont {Amy}, \citenamefont {Glaudell},\ and\ \citenamefont {Ross}}]{amy2020number}%
  \BibitemOpen
  \bibfield  {author} {\bibinfo {author} {\bibfnamefont {M.}~\bibnamefont {Amy}}, \bibinfo {author} {\bibfnamefont {A.~N.}\ \bibnamefont {Glaudell}},\ and\ \bibinfo {author} {\bibfnamefont {N.~J.}\ \bibnamefont {Ross}},\ }\bibfield  {title} {\bibinfo {title} {Number-theoretic characterizations of some restricted clifford+ t circuits},\ }\href {https://doi.org/10.22331/q-2020-04-06-252} {\bibfield  {journal} {\bibinfo  {journal} {Quantum}\ }\textbf {\bibinfo {volume} {4}},\ \bibinfo {pages} {252} (\bibinfo {year} {2020})}\BibitemShut {NoStop}%
\bibitem [{\citenamefont {Kliuchnikov}(2013)}]{kliuchnikov2013synthesis}%
  \BibitemOpen
  \bibfield  {author} {\bibinfo {author} {\bibfnamefont {V.}~\bibnamefont {Kliuchnikov}},\ }\bibfield  {title} {\bibinfo {title} {Synthesis of unitaries with clifford+ t circuits},\ }\bibfield  {journal} {\bibinfo  {journal} {arXiv preprint arXiv:1306.3200}\ }\href {https://doi.org/10.48550/arXiv.1306.3200} {10.48550/arXiv.1306.3200} (\bibinfo {year} {2013})\BibitemShut {NoStop}%
\bibitem [{\citenamefont {Matsumoto}\ and\ \citenamefont {Amano}(2008)}]{matsumoto2008representation}%
  \BibitemOpen
  \bibfield  {author} {\bibinfo {author} {\bibfnamefont {K.}~\bibnamefont {Matsumoto}}\ and\ \bibinfo {author} {\bibfnamefont {K.}~\bibnamefont {Amano}},\ }\bibfield  {title} {\bibinfo {title} {Representation of quantum circuits with clifford and $\pi/8$ gates},\ }\bibfield  {journal} {\bibinfo  {journal} {arXiv preprint arXiv:0806.3834}\ }\href {https://doi.org/10.48550/arXiv.0806.3834} {10.48550/arXiv.0806.3834} (\bibinfo {year} {2008})\BibitemShut {NoStop}%
\bibitem [{\citenamefont {Giles}\ and\ \citenamefont {Selinger}(2013{\natexlab{b}})}]{giles2013remarks}%
  \BibitemOpen
  \bibfield  {author} {\bibinfo {author} {\bibfnamefont {B.}~\bibnamefont {Giles}}\ and\ \bibinfo {author} {\bibfnamefont {P.}~\bibnamefont {Selinger}},\ }\bibfield  {title} {\bibinfo {title} {Remarks on matsumoto and amano's normal form for single-qubit clifford+ t operators},\ }\bibfield  {journal} {\bibinfo  {journal} {arXiv preprint arXiv:1312.6584}\ }\href {https://doi.org/10.48550/arXiv.1312.6584} {10.48550/arXiv.1312.6584} (\bibinfo {year} {2013}{\natexlab{b}})\BibitemShut {NoStop}%
\bibitem [{\citenamefont {Gosset}\ \emph {et~al.}(2013)\citenamefont {Gosset}, \citenamefont {Kliuchnikov}, \citenamefont {Mosca},\ and\ \citenamefont {Russo}}]{gosset2013algorithm}%
  \BibitemOpen
  \bibfield  {author} {\bibinfo {author} {\bibfnamefont {D.}~\bibnamefont {Gosset}}, \bibinfo {author} {\bibfnamefont {V.}~\bibnamefont {Kliuchnikov}}, \bibinfo {author} {\bibfnamefont {M.}~\bibnamefont {Mosca}},\ and\ \bibinfo {author} {\bibfnamefont {V.}~\bibnamefont {Russo}},\ }\bibfield  {title} {\bibinfo {title} {An algorithm for the t-count},\ }\bibfield  {journal} {\bibinfo  {journal} {arXiv preprint arXiv:1308.4134}\ }\href {https://doi.org/10.48550/arXiv.1308.4134} {10.48550/arXiv.1308.4134} (\bibinfo {year} {2013})\BibitemShut {NoStop}%
\bibitem [{\citenamefont {Gheorghiu}\ \emph {et~al.}(2022)\citenamefont {Gheorghiu}, \citenamefont {Mosca},\ and\ \citenamefont {Mukhopadhyay}}]{gheorghiu2022t}%
  \BibitemOpen
  \bibfield  {author} {\bibinfo {author} {\bibfnamefont {V.}~\bibnamefont {Gheorghiu}}, \bibinfo {author} {\bibfnamefont {M.}~\bibnamefont {Mosca}},\ and\ \bibinfo {author} {\bibfnamefont {P.}~\bibnamefont {Mukhopadhyay}},\ }\bibfield  {title} {\bibinfo {title} {T-count and t-depth of any multi-qubit unitary},\ }\href {https://doi.org/10.1038/s41534-022-00651-y} {\bibfield  {journal} {\bibinfo  {journal} {npj Quantum Information}\ }\textbf {\bibinfo {volume} {8}},\ \bibinfo {pages} {141} (\bibinfo {year} {2022})}\BibitemShut {NoStop}%
\bibitem [{\citenamefont {Kivlichan}\ \emph {et~al.}(2018)\citenamefont {Kivlichan}, \citenamefont {McClean}, \citenamefont {Wiebe}, \citenamefont {Gidney}, \citenamefont {Aspuru-Guzik}, \citenamefont {Chan},\ and\ \citenamefont {Babbush}}]{kivlichan2018quantum}%
  \BibitemOpen
  \bibfield  {author} {\bibinfo {author} {\bibfnamefont {I.~D.}\ \bibnamefont {Kivlichan}}, \bibinfo {author} {\bibfnamefont {J.}~\bibnamefont {McClean}}, \bibinfo {author} {\bibfnamefont {N.}~\bibnamefont {Wiebe}}, \bibinfo {author} {\bibfnamefont {C.}~\bibnamefont {Gidney}}, \bibinfo {author} {\bibfnamefont {A.}~\bibnamefont {Aspuru-Guzik}}, \bibinfo {author} {\bibfnamefont {G.~K.-L.}\ \bibnamefont {Chan}},\ and\ \bibinfo {author} {\bibfnamefont {R.}~\bibnamefont {Babbush}},\ }\bibfield  {title} {\bibinfo {title} {Quantum simulation of electronic structure with linear depth and connectivity},\ }\href {https://doi.org/10.1103/PhysRevLett.120.110501} {\bibfield  {journal} {\bibinfo  {journal} {Physical review letters}\ }\textbf {\bibinfo {volume} {120}},\ \bibinfo {pages} {110501} (\bibinfo {year} {2018})}\BibitemShut {NoStop}%
\bibitem [{\citenamefont {Arute}\ \emph {et~al.}(2020)\citenamefont {Arute}, \citenamefont {Arya}, \citenamefont {Babbush}, \citenamefont {Bacon}, \citenamefont {Bardin}, \citenamefont {Barends}, \citenamefont {Boixo}, \citenamefont {Broughton}, \citenamefont {Buckley}, \citenamefont {Buell} \emph {et~al.}}]{arute2020hartree}%
  \BibitemOpen
  \bibfield  {author} {\bibinfo {author} {\bibfnamefont {F.}~\bibnamefont {Arute}}, \bibinfo {author} {\bibfnamefont {K.}~\bibnamefont {Arya}}, \bibinfo {author} {\bibfnamefont {R.}~\bibnamefont {Babbush}}, \bibinfo {author} {\bibfnamefont {D.}~\bibnamefont {Bacon}}, \bibinfo {author} {\bibfnamefont {J.~C.}\ \bibnamefont {Bardin}}, \bibinfo {author} {\bibfnamefont {R.}~\bibnamefont {Barends}}, \bibinfo {author} {\bibfnamefont {S.}~\bibnamefont {Boixo}}, \bibinfo {author} {\bibfnamefont {M.}~\bibnamefont {Broughton}}, \bibinfo {author} {\bibfnamefont {B.~B.}\ \bibnamefont {Buckley}}, \bibinfo {author} {\bibfnamefont {D.~A.}\ \bibnamefont {Buell}}, \emph {et~al.},\ }\bibfield  {title} {\bibinfo {title} {Hartree-fock on a superconducting qubit quantum computer},\ }\href {https://doi.org/10.1126/science.abb9811} {\bibfield  {journal} {\bibinfo  {journal} {Science}\ }\textbf {\bibinfo {volume} {369}},\ \bibinfo {pages} {1084} (\bibinfo {year} {2020})}\BibitemShut {NoStop}%
\bibitem [{\citenamefont {Arrazola}\ \emph {et~al.}(2022)\citenamefont {Arrazola}, \citenamefont {Di~Matteo}, \citenamefont {Quesada}, \citenamefont {Jahangiri}, \citenamefont {Delgado},\ and\ \citenamefont {Killoran}}]{arrazola2022universal}%
  \BibitemOpen
  \bibfield  {author} {\bibinfo {author} {\bibfnamefont {J.~M.}\ \bibnamefont {Arrazola}}, \bibinfo {author} {\bibfnamefont {O.}~\bibnamefont {Di~Matteo}}, \bibinfo {author} {\bibfnamefont {N.}~\bibnamefont {Quesada}}, \bibinfo {author} {\bibfnamefont {S.}~\bibnamefont {Jahangiri}}, \bibinfo {author} {\bibfnamefont {A.}~\bibnamefont {Delgado}},\ and\ \bibinfo {author} {\bibfnamefont {N.}~\bibnamefont {Killoran}},\ }\bibfield  {title} {\bibinfo {title} {Universal quantum circuits for quantum chemistry},\ }\href {https://doi.org/10.22331/q-2022-06-20-742} {\bibfield  {journal} {\bibinfo  {journal} {Quantum}\ }\textbf {\bibinfo {volume} {6}},\ \bibinfo {pages} {742} (\bibinfo {year} {2022})}\BibitemShut {NoStop}%
\bibitem [{\citenamefont {Kraus}(2011)}]{kraus2011compressed}%
  \BibitemOpen
  \bibfield  {author} {\bibinfo {author} {\bibfnamefont {B.}~\bibnamefont {Kraus}},\ }\bibfield  {title} {\bibinfo {title} {Compressed quantum simulation of the ising model},\ }\href {https://doi.org/10.1103/PhysRevLett.107.250503} {\bibfield  {journal} {\bibinfo  {journal} {Physical review letters}\ }\textbf {\bibinfo {volume} {107}},\ \bibinfo {pages} {250503} (\bibinfo {year} {2011})}\BibitemShut {NoStop}%
\bibitem [{\citenamefont {Cervera-Lierta}(2018)}]{cervera2018exact}%
  \BibitemOpen
  \bibfield  {author} {\bibinfo {author} {\bibfnamefont {A.}~\bibnamefont {Cervera-Lierta}},\ }\bibfield  {title} {\bibinfo {title} {Exact ising model simulation on a quantum computer},\ }\href {https://doi.org/10.22331/q-2018-12-21-114} {\bibfield  {journal} {\bibinfo  {journal} {Quantum}\ }\textbf {\bibinfo {volume} {2}},\ \bibinfo {pages} {114} (\bibinfo {year} {2018})}\BibitemShut {NoStop}%
\bibitem [{\citenamefont {Jiang}\ \emph {et~al.}(2018)\citenamefont {Jiang}, \citenamefont {Sung}, \citenamefont {Kechedzhi}, \citenamefont {Smelyanskiy},\ and\ \citenamefont {Boixo}}]{jiang2018quantum}%
  \BibitemOpen
  \bibfield  {author} {\bibinfo {author} {\bibfnamefont {Z.}~\bibnamefont {Jiang}}, \bibinfo {author} {\bibfnamefont {K.~J.}\ \bibnamefont {Sung}}, \bibinfo {author} {\bibfnamefont {K.}~\bibnamefont {Kechedzhi}}, \bibinfo {author} {\bibfnamefont {V.~N.}\ \bibnamefont {Smelyanskiy}},\ and\ \bibinfo {author} {\bibfnamefont {S.}~\bibnamefont {Boixo}},\ }\bibfield  {title} {\bibinfo {title} {Quantum algorithms to simulate many-body physics of correlated fermions},\ }\href {https://doi.org/10.1103/PhysRevApplied.9.044036} {\bibfield  {journal} {\bibinfo  {journal} {Physical Review Applied}\ }\textbf {\bibinfo {volume} {9}},\ \bibinfo {pages} {044036} (\bibinfo {year} {2018})}\BibitemShut {NoStop}%
\bibitem [{\citenamefont {Dallaire-Demers}\ \emph {et~al.}(2019)\citenamefont {Dallaire-Demers}, \citenamefont {Romero}, \citenamefont {Veis}, \citenamefont {Sim},\ and\ \citenamefont {Aspuru-Guzik}}]{dallaire2019low}%
  \BibitemOpen
  \bibfield  {author} {\bibinfo {author} {\bibfnamefont {P.-L.}\ \bibnamefont {Dallaire-Demers}}, \bibinfo {author} {\bibfnamefont {J.}~\bibnamefont {Romero}}, \bibinfo {author} {\bibfnamefont {L.}~\bibnamefont {Veis}}, \bibinfo {author} {\bibfnamefont {S.}~\bibnamefont {Sim}},\ and\ \bibinfo {author} {\bibfnamefont {A.}~\bibnamefont {Aspuru-Guzik}},\ }\bibfield  {title} {\bibinfo {title} {Low-depth circuit ansatz for preparing correlated fermionic states on a quantum computer},\ }\href {https://doi.org/10.1088/2058-9565/ab3951/} {\bibfield  {journal} {\bibinfo  {journal} {Quantum Science and Technology}\ }\textbf {\bibinfo {volume} {4}},\ \bibinfo {pages} {045005} (\bibinfo {year} {2019})}\BibitemShut {NoStop}%
\bibitem [{\citenamefont {Sopena}\ \emph {et~al.}(2022)\citenamefont {Sopena}, \citenamefont {Gordon}, \citenamefont {Garc{\'{i}}a-Mart{\'{i}}n}, \citenamefont {Sierra},\ and\ \citenamefont {L{\'{o}}pez}}]{sopena2022algebraic}%
  \BibitemOpen
  \bibfield  {author} {\bibinfo {author} {\bibfnamefont {A.}~\bibnamefont {Sopena}}, \bibinfo {author} {\bibfnamefont {M.~H.}\ \bibnamefont {Gordon}}, \bibinfo {author} {\bibfnamefont {D.}~\bibnamefont {Garc{\'{i}}a-Mart{\'{i}}n}}, \bibinfo {author} {\bibfnamefont {G.}~\bibnamefont {Sierra}},\ and\ \bibinfo {author} {\bibfnamefont {E.}~\bibnamefont {L{\'{o}}pez}},\ }\bibfield  {title} {\bibinfo {title} {Algebraic {B}ethe {C}ircuits},\ }\href {https://doi.org/10.22331/q-2022-09-08-796} {\bibfield  {journal} {\bibinfo  {journal} {{Quantum}}\ }\textbf {\bibinfo {volume} {6}},\ \bibinfo {pages} {796} (\bibinfo {year} {2022})}\BibitemShut {NoStop}%
\bibitem [{\citenamefont {Ruiz}\ \emph {et~al.}(2024{\natexlab{a}})\citenamefont {Ruiz}, \citenamefont {Sopena}, \citenamefont {Gordon}, \citenamefont {Sierra},\ and\ \citenamefont {L{\'o}pez}}]{ruiz2024bethe}%
  \BibitemOpen
  \bibfield  {author} {\bibinfo {author} {\bibfnamefont {R.}~\bibnamefont {Ruiz}}, \bibinfo {author} {\bibfnamefont {A.}~\bibnamefont {Sopena}}, \bibinfo {author} {\bibfnamefont {M.~H.}\ \bibnamefont {Gordon}}, \bibinfo {author} {\bibfnamefont {G.}~\bibnamefont {Sierra}},\ and\ \bibinfo {author} {\bibfnamefont {E.}~\bibnamefont {L{\'o}pez}},\ }\bibfield  {title} {\bibinfo {title} {The bethe ansatz as a quantum circuit},\ }\href {https://doi.org/10.22331/q-2024-05-23-1356} {\bibfield  {journal} {\bibinfo  {journal} {Quantum}\ }\textbf {\bibinfo {volume} {8}},\ \bibinfo {pages} {1356} (\bibinfo {year} {2024}{\natexlab{a}})}\BibitemShut {NoStop}%
\bibitem [{\citenamefont {Ruiz}\ \emph {et~al.}(2024{\natexlab{b}})\citenamefont {Ruiz}, \citenamefont {Sopena}, \citenamefont {Pozsgay},\ and\ \citenamefont {L{\'o}pez}}]{ruiz2024efficient}%
  \BibitemOpen
  \bibfield  {author} {\bibinfo {author} {\bibfnamefont {R.}~\bibnamefont {Ruiz}}, \bibinfo {author} {\bibfnamefont {A.}~\bibnamefont {Sopena}}, \bibinfo {author} {\bibfnamefont {B.}~\bibnamefont {Pozsgay}},\ and\ \bibinfo {author} {\bibfnamefont {E.}~\bibnamefont {L{\'o}pez}},\ }\bibfield  {title} {\bibinfo {title} {Efficient eigenstate preparation in an integrable model with hilbert space fragmentation},\ }\href {https://arxiv.org/abs/2411.15132} {\bibfield  {journal} {\bibinfo  {journal} {arXiv preprint arXiv:2411.15132}\ } (\bibinfo {year} {2024}{\natexlab{b}})}\BibitemShut {NoStop}%
\bibitem [{\citenamefont {K{\"o}kc{\"u}}\ \emph {et~al.}(2022{\natexlab{a}})\citenamefont {K{\"o}kc{\"u}}, \citenamefont {Steckmann}, \citenamefont {Wang}, \citenamefont {Freericks}, \citenamefont {Dumitrescu},\ and\ \citenamefont {Kemper}}]{kokcu2022fixed}%
  \BibitemOpen
  \bibfield  {author} {\bibinfo {author} {\bibfnamefont {E.}~\bibnamefont {K{\"o}kc{\"u}}}, \bibinfo {author} {\bibfnamefont {T.}~\bibnamefont {Steckmann}}, \bibinfo {author} {\bibfnamefont {Y.}~\bibnamefont {Wang}}, \bibinfo {author} {\bibfnamefont {J.}~\bibnamefont {Freericks}}, \bibinfo {author} {\bibfnamefont {E.~F.}\ \bibnamefont {Dumitrescu}},\ and\ \bibinfo {author} {\bibfnamefont {A.~F.}\ \bibnamefont {Kemper}},\ }\bibfield  {title} {\bibinfo {title} {Fixed depth hamiltonian simulation via cartan decomposition},\ }\href {https://doi.org/10.1103/PhysRevLett.129.070501} {\bibfield  {journal} {\bibinfo  {journal} {Physical Review Letters}\ }\textbf {\bibinfo {volume} {129}},\ \bibinfo {pages} {070501} (\bibinfo {year} {2022}{\natexlab{a}})}\BibitemShut {NoStop}%
\bibitem [{\citenamefont {K{\"o}kc{\"u}}\ \emph {et~al.}(2022{\natexlab{b}})\citenamefont {K{\"o}kc{\"u}}, \citenamefont {Camps}, \citenamefont {Oftelie}, \citenamefont {Freericks}, \citenamefont {de~Jong}, \citenamefont {Van~Beeumen},\ and\ \citenamefont {Kemper}}]{kokcu2022algebraic}%
  \BibitemOpen
  \bibfield  {author} {\bibinfo {author} {\bibfnamefont {E.}~\bibnamefont {K{\"o}kc{\"u}}}, \bibinfo {author} {\bibfnamefont {D.}~\bibnamefont {Camps}}, \bibinfo {author} {\bibfnamefont {L.~B.}\ \bibnamefont {Oftelie}}, \bibinfo {author} {\bibfnamefont {J.~K.}\ \bibnamefont {Freericks}}, \bibinfo {author} {\bibfnamefont {W.~A.}\ \bibnamefont {de~Jong}}, \bibinfo {author} {\bibfnamefont {R.}~\bibnamefont {Van~Beeumen}},\ and\ \bibinfo {author} {\bibfnamefont {A.~F.}\ \bibnamefont {Kemper}},\ }\bibfield  {title} {\bibinfo {title} {Algebraic compression of quantum circuits for hamiltonian evolution},\ }\href {https://doi.org/10.1103/PhysRevA.105.032420} {\bibfield  {journal} {\bibinfo  {journal} {Physical Review A}\ }\textbf {\bibinfo {volume} {105}},\ \bibinfo {pages} {032420} (\bibinfo {year} {2022}{\natexlab{b}})}\BibitemShut {NoStop}%
\bibitem [{\citenamefont {Guaita}\ \emph {et~al.}(2024)\citenamefont {Guaita}, \citenamefont {Hackl},\ and\ \citenamefont {Quella}}]{guaita2024representation}%
  \BibitemOpen
  \bibfield  {author} {\bibinfo {author} {\bibfnamefont {T.}~\bibnamefont {Guaita}}, \bibinfo {author} {\bibfnamefont {L.}~\bibnamefont {Hackl}},\ and\ \bibinfo {author} {\bibfnamefont {T.}~\bibnamefont {Quella}},\ }\bibfield  {title} {\bibinfo {title} {Representation theory of gaussian unitary transformations for bosonic and fermionic systems},\ }\href {https://arxiv.org/abs/2409.11628} {\bibfield  {journal} {\bibinfo  {journal} {arXiv preprint arXiv:2409.11628}\ } (\bibinfo {year} {2024})}\BibitemShut {NoStop}%
\bibitem [{\citenamefont {Braccia}\ \emph {et~al.}(2025)\citenamefont {Braccia}, \citenamefont {Diaz}, \citenamefont {Larocca}, \citenamefont {Cerezo},\ and\ \citenamefont {Garc{\'\i}a-Mart{\'\i}n}}]{braccia2025optimal}%
  \BibitemOpen
  \bibfield  {author} {\bibinfo {author} {\bibfnamefont {P.}~\bibnamefont {Braccia}}, \bibinfo {author} {\bibfnamefont {N.}~\bibnamefont {Diaz}}, \bibinfo {author} {\bibfnamefont {M.}~\bibnamefont {Larocca}}, \bibinfo {author} {\bibfnamefont {M.}~\bibnamefont {Cerezo}},\ and\ \bibinfo {author} {\bibfnamefont {D.}~\bibnamefont {Garc{\'\i}a-Mart{\'\i}n}},\ }\bibfield  {title} {\bibinfo {title} {Optimal haar random fermionic linear optics circuits},\ }\bibfield  {journal} {\bibinfo  {journal} {arXiv preprint arXiv:2505.24212}\ }\href {https://doi.org/10.48550/arXiv.2505.24212} {10.48550/arXiv.2505.24212} (\bibinfo {year} {2025})\BibitemShut {NoStop}%
\bibitem [{\citenamefont {Oszmaniec}\ \emph {et~al.}(2022)\citenamefont {Oszmaniec}, \citenamefont {Dangniam}, \citenamefont {Morales},\ and\ \citenamefont {Zimbor{\'a}s}}]{oszmaniec2022fermion}%
  \BibitemOpen
  \bibfield  {author} {\bibinfo {author} {\bibfnamefont {M.}~\bibnamefont {Oszmaniec}}, \bibinfo {author} {\bibfnamefont {N.}~\bibnamefont {Dangniam}}, \bibinfo {author} {\bibfnamefont {M.~E.}\ \bibnamefont {Morales}},\ and\ \bibinfo {author} {\bibfnamefont {Z.}~\bibnamefont {Zimbor{\'a}s}},\ }\bibfield  {title} {\bibinfo {title} {Fermion sampling: a robust quantum computational advantage scheme using fermionic linear optics and magic input states},\ }\href {https://doi.org/10.1103/PRXQuantum.3.020328} {\bibfield  {journal} {\bibinfo  {journal} {PRX Quantum}\ }\textbf {\bibinfo {volume} {3}},\ \bibinfo {pages} {020328} (\bibinfo {year} {2022})}\BibitemShut {NoStop}%
\bibitem [{\citenamefont {Zanardi}\ \emph {et~al.}(2000)\citenamefont {Zanardi}, \citenamefont {Zalka},\ and\ \citenamefont {Faoro}}]{zanardi2000entangling}%
  \BibitemOpen
  \bibfield  {author} {\bibinfo {author} {\bibfnamefont {P.}~\bibnamefont {Zanardi}}, \bibinfo {author} {\bibfnamefont {C.}~\bibnamefont {Zalka}},\ and\ \bibinfo {author} {\bibfnamefont {L.}~\bibnamefont {Faoro}},\ }\bibfield  {title} {\bibinfo {title} {Entangling power of quantum evolutions},\ }\href {https://doi.org/10.1103/PhysRevA.62.030301} {\bibfield  {journal} {\bibinfo  {journal} {Physical Review A}\ }\textbf {\bibinfo {volume} {62}},\ \bibinfo {pages} {030301} (\bibinfo {year} {2000})}\BibitemShut {NoStop}%
\bibitem [{\citenamefont {Zanardi}(2001)}]{zanardi2001entanglement}%
  \BibitemOpen
  \bibfield  {author} {\bibinfo {author} {\bibfnamefont {P.}~\bibnamefont {Zanardi}},\ }\bibfield  {title} {\bibinfo {title} {Entanglement of quantum evolutions},\ }\href {https://doi.org/10.1103/PhysRevA.63.040304} {\bibfield  {journal} {\bibinfo  {journal} {Physical Review A}\ }\textbf {\bibinfo {volume} {63}},\ \bibinfo {pages} {040304} (\bibinfo {year} {2001})}\BibitemShut {NoStop}%
\bibitem [{\citenamefont {Hall}(2013)}]{hall2013lie}%
  \BibitemOpen
  \bibfield  {author} {\bibinfo {author} {\bibfnamefont {B.~C.}\ \bibnamefont {Hall}},\ }\href@noop {} {\emph {\bibinfo {title} {Lie groups, Lie algebras, and representations}}}\ (\bibinfo  {publisher} {Springer},\ \bibinfo {year} {2013})\BibitemShut {NoStop}%
\bibitem [{\citenamefont {Diaconis}\ and\ \citenamefont {Forrester}(2017)}]{diaconis2017hurwitz}%
  \BibitemOpen
  \bibfield  {author} {\bibinfo {author} {\bibfnamefont {P.}~\bibnamefont {Diaconis}}\ and\ \bibinfo {author} {\bibfnamefont {P.~J.}\ \bibnamefont {Forrester}},\ }\bibfield  {title} {\bibinfo {title} {Hurwitz and the origins of random matrix theory in mathematics},\ }\href {https://doi.org/10.1142/S2010326317300017} {\bibfield  {journal} {\bibinfo  {journal} {Random Matrices: Theory and Applications}\ }\textbf {\bibinfo {volume} {6}},\ \bibinfo {pages} {1730001} (\bibinfo {year} {2017})}\BibitemShut {NoStop}%
\bibitem [{\citenamefont {Selinger}()}]{gridsynth}%
  \BibitemOpen
  \bibfield  {author} {\bibinfo {author} {\bibfnamefont {P.}~\bibnamefont {Selinger}},\ }\href {https://www.mathstat.dal.ca/~selinger/newsynth/} {\bibinfo {title} {Exact and approximate synthesis of quantum circuits}}\BibitemShut {NoStop}%
\bibitem [{\citenamefont {Yamamoto}\ \emph {et~al.}()\citenamefont {Yamamoto}, \citenamefont {Hamaguchi.},\ and\ \citenamefont {Lapeyre}}]{pygridsynth}%
  \BibitemOpen
  \bibfield  {author} {\bibinfo {author} {\bibfnamefont {S.}~\bibnamefont {Yamamoto}}, \bibinfo {author} {\bibfnamefont {H.}~\bibnamefont {Hamaguchi.}},\ and\ \bibinfo {author} {\bibfnamefont {J.}~\bibnamefont {Lapeyre}},\ }\href {https://github.com/quantum-programming/pygridsynth} {\bibinfo {title} {pygridsynth}}\BibitemShut {NoStop}%
\bibitem [{\citenamefont {Ruiz}\ \emph {et~al.}(2025)\citenamefont {Ruiz}, \citenamefont {Laakkonen}, \citenamefont {Bausch}, \citenamefont {Balog}, \citenamefont {Barekatain}, \citenamefont {Heras}, \citenamefont {Novikov}, \citenamefont {Fitzpatrick}, \citenamefont {Romera-Paredes}, \citenamefont {van~de Wetering} \emph {et~al.}}]{ruiz2025quantum}%
  \BibitemOpen
  \bibfield  {author} {\bibinfo {author} {\bibfnamefont {F.~J.}\ \bibnamefont {Ruiz}}, \bibinfo {author} {\bibfnamefont {T.}~\bibnamefont {Laakkonen}}, \bibinfo {author} {\bibfnamefont {J.}~\bibnamefont {Bausch}}, \bibinfo {author} {\bibfnamefont {M.}~\bibnamefont {Balog}}, \bibinfo {author} {\bibfnamefont {M.}~\bibnamefont {Barekatain}}, \bibinfo {author} {\bibfnamefont {F.~J.}\ \bibnamefont {Heras}}, \bibinfo {author} {\bibfnamefont {A.}~\bibnamefont {Novikov}}, \bibinfo {author} {\bibfnamefont {N.}~\bibnamefont {Fitzpatrick}}, \bibinfo {author} {\bibfnamefont {B.}~\bibnamefont {Romera-Paredes}}, \bibinfo {author} {\bibfnamefont {J.}~\bibnamefont {van~de Wetering}}, \emph {et~al.},\ }\bibfield  {title} {\bibinfo {title} {Quantum circuit optimization with alphatensor},\ }\href {https://doi.org/10.1038/s42256-025-01001-1} {\bibfield  {journal} {\bibinfo  {journal} {Nature Machine Intelligence}\ }\textbf {\bibinfo {volume} {7}},\ \bibinfo {pages} {374} (\bibinfo {year} {2025})}\BibitemShut {NoStop}%
\bibitem [{\citenamefont {Valcarce}\ \emph {et~al.}(2025)\citenamefont {Valcarce}, \citenamefont {Grivet},\ and\ \citenamefont {Sangouard}}]{valcarce2025unitary}%
  \BibitemOpen
  \bibfield  {author} {\bibinfo {author} {\bibfnamefont {X.}~\bibnamefont {Valcarce}}, \bibinfo {author} {\bibfnamefont {B.}~\bibnamefont {Grivet}},\ and\ \bibinfo {author} {\bibfnamefont {N.}~\bibnamefont {Sangouard}},\ }\bibfield  {title} {\bibinfo {title} {Unitary synthesis with alphazero via dynamic circuits},\ }\bibfield  {journal} {\bibinfo  {journal} {arXiv preprint arXiv:2508.21217}\ }\href {https://doi.org/10.48550/arXiv.2508.21217} {10.48550/arXiv.2508.21217} (\bibinfo {year} {2025})\BibitemShut {NoStop}%
\bibitem [{\citenamefont {Bosco}\ \emph {et~al.}(2026)\citenamefont {Bosco}, \citenamefont {Cincio}, \citenamefont {Serra},\ and\ \citenamefont {Cerezo}}]{bosco2026quantum}%
  \BibitemOpen
  \bibfield  {author} {\bibinfo {author} {\bibfnamefont {D.~L.}\ \bibnamefont {Bosco}}, \bibinfo {author} {\bibfnamefont {L.}~\bibnamefont {Cincio}}, \bibinfo {author} {\bibfnamefont {G.}~\bibnamefont {Serra}},\ and\ \bibinfo {author} {\bibfnamefont {M.}~\bibnamefont {Cerezo}},\ }\bibfield  {title} {\bibinfo {title} {Quantum circuit pre-synthesis: Learning local edits to reduce $ t $-count},\ }\bibfield  {journal} {\bibinfo  {journal} {arXiv preprint arXiv:2601.19738}\ }\href {https://doi.org/10.48550/arXiv.2601.19738} {10.48550/arXiv.2601.19738} (\bibinfo {year} {2026})\BibitemShut {NoStop}%
\bibitem [{\citenamefont {Bernstein}\ and\ \citenamefont {Vazirani}(1997)}]{bernstein1997quantum}%
  \BibitemOpen
  \bibfield  {author} {\bibinfo {author} {\bibfnamefont {E.}~\bibnamefont {Bernstein}}\ and\ \bibinfo {author} {\bibfnamefont {U.}~\bibnamefont {Vazirani}},\ }\bibfield  {title} {\bibinfo {title} {Quantum complexity theory},\ }\href {https://doi.org/https://doi.org/10.1137/S0097539796300921} {\bibfield  {journal} {\bibinfo  {journal} {SIAM Journal on computing}\ }\textbf {\bibinfo {volume} {26}},\ \bibinfo {pages} {1411} (\bibinfo {year} {1997})}\BibitemShut {NoStop}%
\bibitem [{\citenamefont {McKay}\ \emph {et~al.}(2017)\citenamefont {McKay}, \citenamefont {Wood}, \citenamefont {Sheldon}, \citenamefont {Chow},\ and\ \citenamefont {Gambetta}}]{mckay2017efficient}%
  \BibitemOpen
  \bibfield  {author} {\bibinfo {author} {\bibfnamefont {D.~C.}\ \bibnamefont {McKay}}, \bibinfo {author} {\bibfnamefont {C.~J.}\ \bibnamefont {Wood}}, \bibinfo {author} {\bibfnamefont {S.}~\bibnamefont {Sheldon}}, \bibinfo {author} {\bibfnamefont {J.~M.}\ \bibnamefont {Chow}},\ and\ \bibinfo {author} {\bibfnamefont {J.~M.}\ \bibnamefont {Gambetta}},\ }\bibfield  {title} {\bibinfo {title} {Efficient $z$ gates for quantum computing},\ }\href {https://doi.org/10.1103/PhysRevA.96.022330} {\bibfield  {journal} {\bibinfo  {journal} {Phys. Rev. A}\ }\textbf {\bibinfo {volume} {96}},\ \bibinfo {pages} {022330} (\bibinfo {year} {2017})}\BibitemShut {NoStop}%
\bibitem [{\citenamefont {Gidney}\ and\ \citenamefont {Fowler}(2019)}]{gidney2019efficient}%
  \BibitemOpen
  \bibfield  {author} {\bibinfo {author} {\bibfnamefont {C.}~\bibnamefont {Gidney}}\ and\ \bibinfo {author} {\bibfnamefont {A.~G.}\ \bibnamefont {Fowler}},\ }\bibfield  {title} {\bibinfo {title} {Efficient magic state factories with a catalyzed {$\vert$CCZ$\rangle$} to {2$\vert$T$\rangle$} transformation},\ }\href {https://doi.org/10.22331/q-2019-04-30-135} {\bibfield  {journal} {\bibinfo  {journal} {Quantum}\ }\textbf {\bibinfo {volume} {3}},\ \bibinfo {pages} {135} (\bibinfo {year} {2019})}\BibitemShut {NoStop}%
\bibitem [{\citenamefont {Goh}\ \emph {et~al.}(2025)\citenamefont {Goh}, \citenamefont {Larocca}, \citenamefont {Cincio}, \citenamefont {Cerezo},\ and\ \citenamefont {Sauvage}}]{goh2023lie}%
  \BibitemOpen
  \bibfield  {author} {\bibinfo {author} {\bibfnamefont {M.~L.}\ \bibnamefont {Goh}}, \bibinfo {author} {\bibfnamefont {M.}~\bibnamefont {Larocca}}, \bibinfo {author} {\bibfnamefont {L.}~\bibnamefont {Cincio}}, \bibinfo {author} {\bibfnamefont {M.}~\bibnamefont {Cerezo}},\ and\ \bibinfo {author} {\bibfnamefont {F.}~\bibnamefont {Sauvage}},\ }\bibfield  {title} {\bibinfo {title} {Lie-algebraic classical simulations for quantum computing},\ }\href {https://doi.org/10.1103/3y65-f5w6} {\bibfield  {journal} {\bibinfo  {journal} {Physical Review Research}\ }\textbf {\bibinfo {volume} {7}},\ \bibinfo {pages} {033266} (\bibinfo {year} {2025})}\BibitemShut {NoStop}%
\bibitem [{\citenamefont {Biere}\ \emph {et~al.}(2024)\citenamefont {Biere}, \citenamefont {Faller}, \citenamefont {Fazekas}, \citenamefont {Fleury}, \citenamefont {Froleyks},\ and\ \citenamefont {Pollitt}}]{biere2024cadical}%
  \BibitemOpen
  \bibfield  {author} {\bibinfo {author} {\bibfnamefont {A.}~\bibnamefont {Biere}}, \bibinfo {author} {\bibfnamefont {T.}~\bibnamefont {Faller}}, \bibinfo {author} {\bibfnamefont {K.}~\bibnamefont {Fazekas}}, \bibinfo {author} {\bibfnamefont {M.}~\bibnamefont {Fleury}}, \bibinfo {author} {\bibfnamefont {N.}~\bibnamefont {Froleyks}},\ and\ \bibinfo {author} {\bibfnamefont {F.}~\bibnamefont {Pollitt}},\ }\bibfield  {title} {\bibinfo {title} {Cadical, gimsatul, isasat and kissat entering the sat competition 2024},\ }\href {https://doi.org/10138/584822} {\bibfield  {journal} {\bibinfo  {journal} {Proc. of SAT Competition}\ ,\ \bibinfo {pages} {8}} (\bibinfo {year} {2024})}\BibitemShut {NoStop}%
\bibitem [{\citenamefont {Gouzien}()}]{SATsynthesis}%
  \BibitemOpen
  \bibfield  {author} {\bibinfo {author} {\bibfnamefont {E.}~\bibnamefont {Gouzien}},\ }\href {https://github.com/ElieGouzien/quatum_gate_sat_synthesis} {\bibinfo {title} {Quatum gate sat synthesis}}\BibitemShut {NoStop}%
\bibitem [{\citenamefont {Martins}\ \emph {et~al.}(2014)\citenamefont {Martins}, \citenamefont {Manquinho},\ and\ \citenamefont {Lynce}}]{martins2014open}%
  \BibitemOpen
  \bibfield  {author} {\bibinfo {author} {\bibfnamefont {R.}~\bibnamefont {Martins}}, \bibinfo {author} {\bibfnamefont {V.}~\bibnamefont {Manquinho}},\ and\ \bibinfo {author} {\bibfnamefont {I.}~\bibnamefont {Lynce}},\ }\bibfield  {title} {\bibinfo {title} {Open-wbo: A modular maxsat solver},\ }in\ \href {https://doi.org/10.1007/978-3-319-09284-3_33} {\emph {\bibinfo {booktitle} {International Conference on Theory and Applications of Satisfiability Testing}}}\ (\bibinfo {organization} {Springer},\ \bibinfo {year} {2014})\ pp.\ \bibinfo {pages} {438--445}\BibitemShut {NoStop}%
\bibitem [{\citenamefont {Collura}\ \emph {et~al.}(2024)\citenamefont {Collura}, \citenamefont {De~Nardis}, \citenamefont {Alba},\ and\ \citenamefont {Lami}}]{collura2024quantum}%
  \BibitemOpen
  \bibfield  {author} {\bibinfo {author} {\bibfnamefont {M.}~\bibnamefont {Collura}}, \bibinfo {author} {\bibfnamefont {J.}~\bibnamefont {De~Nardis}}, \bibinfo {author} {\bibfnamefont {V.}~\bibnamefont {Alba}},\ and\ \bibinfo {author} {\bibfnamefont {G.}~\bibnamefont {Lami}},\ }\bibfield  {title} {\bibinfo {title} {The quantum magic of fermionic gaussian states},\ }\bibfield  {journal} {\bibinfo  {journal} {arXiv preprint arXiv:2412.05367}\ }\href {https://doi.org/10.48550/arXiv.2412.05367} {10.48550/arXiv.2412.05367} (\bibinfo {year} {2024})\BibitemShut {NoStop}%
\bibitem [{\citenamefont {Deneris}\ \emph {et~al.}(2025)\citenamefont {Deneris}, \citenamefont {Braccia}, \citenamefont {Bermejo}, \citenamefont {Mele},\ and\ \citenamefont {Cerezo}}]{deneris2025analyzing}%
  \BibitemOpen
  \bibfield  {author} {\bibinfo {author} {\bibfnamefont {A.~E.}\ \bibnamefont {Deneris}}, \bibinfo {author} {\bibfnamefont {P.}~\bibnamefont {Braccia}}, \bibinfo {author} {\bibfnamefont {P.}~\bibnamefont {Bermejo}}, \bibinfo {author} {\bibfnamefont {A.~A.}\ \bibnamefont {Mele}},\ and\ \bibinfo {author} {\bibfnamefont {M.}~\bibnamefont {Cerezo}},\ }\bibfield  {title} {\bibinfo {title} {Analyzing the free states of one quantum resource theory as resource states of another},\ }\bibfield  {journal} {\bibinfo  {journal} {arXiv preprint arXiv:2507.11793}\ }\href {https://doi.org/10.48550/arXiv.2507.11793} {10.48550/arXiv.2507.11793} (\bibinfo {year} {2025})\BibitemShut {NoStop}%
\bibitem [{\citenamefont {Portik}\ \emph {et~al.}(2025)\citenamefont {Portik}, \citenamefont {K{\'a}lm{\'a}n}, \citenamefont {Monz},\ and\ \citenamefont {Zimbor{\'a}s}}]{portik2025clifford}%
  \BibitemOpen
  \bibfield  {author} {\bibinfo {author} {\bibfnamefont {A.}~\bibnamefont {Portik}}, \bibinfo {author} {\bibfnamefont {O.}~\bibnamefont {K{\'a}lm{\'a}n}}, \bibinfo {author} {\bibfnamefont {T.}~\bibnamefont {Monz}},\ and\ \bibinfo {author} {\bibfnamefont {Z.}~\bibnamefont {Zimbor{\'a}s}},\ }\bibfield  {title} {\bibinfo {title} {Clifford volume and free fermion volume: Complementary scalable benchmarks for quantum computers},\ }\bibfield  {journal} {\bibinfo  {journal} {arXiv preprint arXiv:2512.19413}\ }\href {https://doi.org/10.48550/arXiv.2512.19413} {10.48550/arXiv.2512.19413} (\bibinfo {year} {2025})\BibitemShut {NoStop}%
\bibitem [{\citenamefont {Carrasco}\ \emph {et~al.}(2024)\citenamefont {Carrasco}, \citenamefont {Langer}, \citenamefont {Neven},\ and\ \citenamefont {Kraus}}]{carrasco2024gaining}%
  \BibitemOpen
  \bibfield  {author} {\bibinfo {author} {\bibfnamefont {J.}~\bibnamefont {Carrasco}}, \bibinfo {author} {\bibfnamefont {M.}~\bibnamefont {Langer}}, \bibinfo {author} {\bibfnamefont {A.}~\bibnamefont {Neven}},\ and\ \bibinfo {author} {\bibfnamefont {B.}~\bibnamefont {Kraus}},\ }\bibfield  {title} {\bibinfo {title} {Gaining confidence on the correct realization of arbitrary quantum computations},\ }\href {https://doi.org/10.1103/PhysRevResearch.6.L032074} {\bibfield  {journal} {\bibinfo  {journal} {Physical Review Research}\ }\textbf {\bibinfo {volume} {6}},\ \bibinfo {pages} {L032074} (\bibinfo {year} {2024})}\BibitemShut {NoStop}%
\bibitem [{\citenamefont {Walther}\ \emph {et~al.}(2012)\citenamefont {Walther}, \citenamefont {Ziesel}, \citenamefont {Ruster}, \citenamefont {Dawkins}, \citenamefont {Ott}, \citenamefont {Hettrich}, \citenamefont {Singer}, \citenamefont {Schmidt-Kaler},\ and\ \citenamefont {Poschinger}}]{walther2012controlling}%
  \BibitemOpen
  \bibfield  {author} {\bibinfo {author} {\bibfnamefont {A.}~\bibnamefont {Walther}}, \bibinfo {author} {\bibfnamefont {F.}~\bibnamefont {Ziesel}}, \bibinfo {author} {\bibfnamefont {T.}~\bibnamefont {Ruster}}, \bibinfo {author} {\bibfnamefont {S.~T.}\ \bibnamefont {Dawkins}}, \bibinfo {author} {\bibfnamefont {K.}~\bibnamefont {Ott}}, \bibinfo {author} {\bibfnamefont {M.}~\bibnamefont {Hettrich}}, \bibinfo {author} {\bibfnamefont {K.}~\bibnamefont {Singer}}, \bibinfo {author} {\bibfnamefont {F.}~\bibnamefont {Schmidt-Kaler}},\ and\ \bibinfo {author} {\bibfnamefont {U.}~\bibnamefont {Poschinger}},\ }\bibfield  {title} {\bibinfo {title} {Controlling fast transport of cold trapped ions},\ }\href {https://doi.org/10.1103/PhysRevLett.109.080501} {\bibfield  {journal} {\bibinfo  {journal} {Physical review letters}\ }\textbf {\bibinfo {volume} {109}},\ \bibinfo {pages} {080501} (\bibinfo {year} {2012})}\BibitemShut {NoStop}%
\bibitem [{\citenamefont {Ruster}\ \emph {et~al.}(2014)\citenamefont {Ruster}, \citenamefont {Warschburger}, \citenamefont {Kaufmann}, \citenamefont {Schmiegelow}, \citenamefont {Walther}, \citenamefont {Hettrich}, \citenamefont {Pfister}, \citenamefont {Kaushal}, \citenamefont {Schmidt-Kaler},\ and\ \citenamefont {Poschinger}}]{ruster2014experimental}%
  \BibitemOpen
  \bibfield  {author} {\bibinfo {author} {\bibfnamefont {T.}~\bibnamefont {Ruster}}, \bibinfo {author} {\bibfnamefont {C.}~\bibnamefont {Warschburger}}, \bibinfo {author} {\bibfnamefont {H.}~\bibnamefont {Kaufmann}}, \bibinfo {author} {\bibfnamefont {C.~T.}\ \bibnamefont {Schmiegelow}}, \bibinfo {author} {\bibfnamefont {A.}~\bibnamefont {Walther}}, \bibinfo {author} {\bibfnamefont {M.}~\bibnamefont {Hettrich}}, \bibinfo {author} {\bibfnamefont {A.}~\bibnamefont {Pfister}}, \bibinfo {author} {\bibfnamefont {V.}~\bibnamefont {Kaushal}}, \bibinfo {author} {\bibfnamefont {F.}~\bibnamefont {Schmidt-Kaler}},\ and\ \bibinfo {author} {\bibfnamefont {U.~G.}\ \bibnamefont {Poschinger}},\ }\bibfield  {title} {\bibinfo {title} {Experimental realization of fast ion separation in segmented paul traps},\ }\href {https://doi.org/10.1103/PhysRevA.90.033410} {\bibfield  {journal} {\bibinfo  {journal} {Phys. Rev. A}\ }\textbf {\bibinfo {volume} {90}},\ \bibinfo {pages} {033410} (\bibinfo {year} {2014})}\BibitemShut {NoStop}%
\bibitem [{\citenamefont {Taballione}\ \emph {et~al.}(2019)\citenamefont {Taballione}, \citenamefont {Wolterink}, \citenamefont {Lugani}, \citenamefont {Eckstein}, \citenamefont {Bell}, \citenamefont {Grootjans}, \citenamefont {Visscher}, \citenamefont {Geskus}, \citenamefont {Roeloffzen}, \citenamefont {Renema} \emph {et~al.}}]{taballione20198}%
  \BibitemOpen
  \bibfield  {author} {\bibinfo {author} {\bibfnamefont {C.}~\bibnamefont {Taballione}}, \bibinfo {author} {\bibfnamefont {T.~A.}\ \bibnamefont {Wolterink}}, \bibinfo {author} {\bibfnamefont {J.}~\bibnamefont {Lugani}}, \bibinfo {author} {\bibfnamefont {A.}~\bibnamefont {Eckstein}}, \bibinfo {author} {\bibfnamefont {B.~A.}\ \bibnamefont {Bell}}, \bibinfo {author} {\bibfnamefont {R.}~\bibnamefont {Grootjans}}, \bibinfo {author} {\bibfnamefont {I.}~\bibnamefont {Visscher}}, \bibinfo {author} {\bibfnamefont {D.}~\bibnamefont {Geskus}}, \bibinfo {author} {\bibfnamefont {C.~G.}\ \bibnamefont {Roeloffzen}}, \bibinfo {author} {\bibfnamefont {J.~J.}\ \bibnamefont {Renema}}, \emph {et~al.},\ }\bibfield  {title} {\bibinfo {title} {8$\times$ 8 reconfigurable quantum photonic processor based on silicon nitride waveguides},\ }\href {https://doi.org/10.1364/OE.27.026842} {\bibfield  {journal} {\bibinfo  {journal} {Optics express}\ }\textbf {\bibinfo {volume} {27}},\ \bibinfo {pages} {26842} (\bibinfo {year}
  {2019})}\BibitemShut {NoStop}%
\bibitem [{\citenamefont {Horsman}\ \emph {et~al.}(2012)\citenamefont {Horsman}, \citenamefont {Fowler}, \citenamefont {Devitt},\ and\ \citenamefont {Van~Meter}}]{horsman2012surface}%
  \BibitemOpen
  \bibfield  {author} {\bibinfo {author} {\bibfnamefont {C.}~\bibnamefont {Horsman}}, \bibinfo {author} {\bibfnamefont {A.~G.}\ \bibnamefont {Fowler}}, \bibinfo {author} {\bibfnamefont {S.}~\bibnamefont {Devitt}},\ and\ \bibinfo {author} {\bibfnamefont {R.}~\bibnamefont {Van~Meter}},\ }\bibfield  {title} {\bibinfo {title} {Surface code quantum computing by lattice surgery},\ }\href {https://doi.org/10.1088/1367-2630/14/12/123011} {\bibfield  {journal} {\bibinfo  {journal} {New Journal of Physics}\ }\textbf {\bibinfo {volume} {14}},\ \bibinfo {pages} {123011} (\bibinfo {year} {2012})}\BibitemShut {NoStop}%
\bibitem [{\citenamefont {B{\"a}umer}\ \emph {et~al.}(2024)\citenamefont {B{\"a}umer}, \citenamefont {Tripathi}, \citenamefont {Wang}, \citenamefont {Rall}, \citenamefont {Chen}, \citenamefont {Majumder}, \citenamefont {Seif},\ and\ \citenamefont {Minev}}]{baumer2024efficient}%
  \BibitemOpen
  \bibfield  {author} {\bibinfo {author} {\bibfnamefont {E.}~\bibnamefont {B{\"a}umer}}, \bibinfo {author} {\bibfnamefont {V.}~\bibnamefont {Tripathi}}, \bibinfo {author} {\bibfnamefont {D.~S.}\ \bibnamefont {Wang}}, \bibinfo {author} {\bibfnamefont {P.}~\bibnamefont {Rall}}, \bibinfo {author} {\bibfnamefont {E.~H.}\ \bibnamefont {Chen}}, \bibinfo {author} {\bibfnamefont {S.}~\bibnamefont {Majumder}}, \bibinfo {author} {\bibfnamefont {A.}~\bibnamefont {Seif}},\ and\ \bibinfo {author} {\bibfnamefont {Z.~K.}\ \bibnamefont {Minev}},\ }\bibfield  {title} {\bibinfo {title} {Efficient long-range entanglement using dynamic circuits},\ }\href {https://doi.org/10.1103/PRXQuantum.5.030339} {\bibfield  {journal} {\bibinfo  {journal} {PRX Quantum}\ }\textbf {\bibinfo {volume} {5}},\ \bibinfo {pages} {030339} (\bibinfo {year} {2024})}\BibitemShut {NoStop}%
\bibitem [{\citenamefont {Zhang}\ \emph {et~al.}(2024)\citenamefont {Zhang}, \citenamefont {Wiersema}, \citenamefont {Carrasquilla}, \citenamefont {Cincio},\ and\ \citenamefont {Kim}}]{zhang2024scalable}%
  \BibitemOpen
  \bibfield  {author} {\bibinfo {author} {\bibfnamefont {Y.}~\bibnamefont {Zhang}}, \bibinfo {author} {\bibfnamefont {R.}~\bibnamefont {Wiersema}}, \bibinfo {author} {\bibfnamefont {J.}~\bibnamefont {Carrasquilla}}, \bibinfo {author} {\bibfnamefont {L.}~\bibnamefont {Cincio}},\ and\ \bibinfo {author} {\bibfnamefont {Y.~B.}\ \bibnamefont {Kim}},\ }\bibfield  {title} {\bibinfo {title} {Scalable quantum dynamics compilation via quantum machine learning},\ }\bibfield  {journal} {\bibinfo  {journal} {arXiv preprint arXiv:2409.16346}\ }\href {https://doi.org/10.48550/arXiv.2409.16346} {10.48550/arXiv.2409.16346} (\bibinfo {year} {2024})\BibitemShut {NoStop}%
\bibitem [{\citenamefont {Nomizu}(1969)}]{nomizu1969foundations}%
  \BibitemOpen
  \bibfield  {author} {\bibinfo {author} {\bibfnamefont {K.}~\bibnamefont {Nomizu}},\ }\href@noop {} {\emph {\bibinfo {title} {Foundations of differential geometry}}}\ (\bibinfo  {publisher} {Interscience},\ \bibinfo {year} {1969})\BibitemShut {NoStop}%
\end{thebibliography}%

\clearpage
\onecolumngrid

\makeatletter
\let\set@footnotewidth\set@footnotewidth@one
\let\compose@footnotes\compose@footnotes@one
\makeatother

\appendix
\setcounter{lemma}{0}
\setcounter{theorem}{0}
\setcounter{proposition}{0}

\section{Approximate matchgate synthesis}
\subsection{An $\SU(2)$ universal gate set}\label{ap:su2_dense}

We begin by restating a standard result in quantum computation~\cite{boykin2000new}, namely, Lemma~\ref{lem:W,T}. This lemma will be instrumental to find a universal matchgate discrete set.

\begin{lemma}\label{lem-ap:W,T}
	The set $\{iW,\,e^{-i\frac{\pi}{8}}T\}$ is dense in $\mathbb{SU}(2)$.
\end{lemma}

\begin{proof}
	
	Let us consider the following unitary gate, made up of $e^{-i\frac{\pi}{8}}T$ and $iW$ gates,
	\begin{equation}
		R_1:= e^{-i\frac{7\pi}{8}}T^{7} iW e^{-i\frac{\pi}{8}}T iW = T^{-1} W T W= Z^{-1/4} Y^{1/4} \in \SU(2)\,,
	\end{equation}
	where we used that  $T^7=T^{-1}$, together with the identity $ Y^\alpha = W Z^\alpha W$ (recall that $T=Z^{1/4}$). 
    Since any $U\in\mathbb{SU}(2)$ can be written as
\begin{equation}
	U= e^{i \theta \hat{n}\cdot\hat{\sigma}}=\cos \theta\, \id+i\sin\theta\, \hat{n}\cdot\hat{\sigma}\,,
\end{equation}
where $\hat{n}$ is a unit vector in $\mathbb{R}^3$, and $\hat{\sigma}=(X,Y,Z)$~\cite{nielsen2000quantum}, it follows that
	\begin{align}\label{eq-ap:R1}
		R_1 = &e^{-i \pi/8} \left(\cos(\pi/8) \id +i \sin(\pi/8) Z\right) \times   e^{i \pi/8}\left(\cos(\pi/8) \id -i \sin(\pi/8) Y\right) \nonumber \\ = & \cos^2(\pi/8) \id - i \sin(\pi/8)\cos(\pi/8)  Y  + i \sin(\pi/8)\cos(\pi/8) Z - i\sin^2(\pi/8) X \nonumber \\  = & \cos (\lambda\pi) \id + i \sin(\lambda\pi)\, \hat{n}_1\cdot\hat{\sigma}\,,
	\end{align}
	with $\lambda$ such that
    \begin{equation}
        \cos(\lambda\pi) = \cos^2(\pi/8) \,,
    \end{equation}
    and
    \begin{equation}
    \hat{n}_1=\frac{1}{\sqrt{1+\cos^2(\pi/8)}}\left(-\sin(\pi/8),\,-\cos(\pi/8),\,\cos(\pi/8)\right).
    \end{equation}
    Now, $\lambda$ can be shown to be an irrational number, as proven in Ref.~\cite{boykin2000new}. 
    
    Moreover, let us consider the gate
	\begin{equation}
	    R_2 :=   W^{-1/2} Z^{-1/4} Y^{1/4} W^{1/2} \in \SU(2)\,,
	\end{equation}
    where it can be verified that $W^{1/2}= Z^{1/2} Y^{1/4}  Z^{1/2} Y^{-1/4} Z^{1/2}$, and so $R_2$ can be exactly synthesized with $e^{-i\frac{\pi}{8}}T$, $W$ and $R_1=Z^{-1/4}Y^{1/4}$ gates.
We find that
    \begin{align}
        R_2 &=  W^{-1/2} \left(\cos (\lambda\pi) \id + i \sin(\lambda\pi)\, \hat{n}_1\cdot\hat{\sigma}\right) W^{1/2} \nonumber\\ &= \cos (\lambda\pi) \id + i \sin(\lambda\pi) W^{-1/2} \,\hat{n}_1\cdot\hat{\sigma} \,W^{1/2} \,,
    \end{align}
    where we used Eq.~\eqref{eq-ap:R1}. By noticing that $W^{1/2} = e^{i\pi/4}\left( \cos\left( \pi/4  \right)\id - i \sin(  \pi/ 4 ) Y - i \sin(  \pi/ 4 )Z\right)$, we arrive at 
    \begin{align}
        W^{-1/2} X W^{1/2} & = \frac{1}{\sqrt 2} (Z-Y)\,, \\  
        W^{-1/2} Y W^{1/2} & = \frac{1}{2}( Z+ Y + \sqrt 2 X)\,, \\
        W^{-1/2} Z W^{1/2} & = \frac{1}{2}( Z+ Y - \sqrt 2 X) \,.
    \end{align}
    Thus, we find 
    \begin{equation}
        R_2 =  \cos (\lambda\pi) \id + i \sin(\lambda\pi) \,\hat{n}_2\cdot\hat{\sigma} \, \,,
    \end{equation}
    with
    \begin{equation}
    \hat{n}_2= \frac{1}{\sqrt{2(1+\cos^2(\pi/8))}}\left(-2\cos(\pi/8),\,\sin(\pi/8),\,-\sin(\pi/8)\right).
    \end{equation}
    Furthermore, it is easy to verify that $\hat{n}_1$ and $\hat{n}_2$ are orthogonal. Two rotations by irrational multiples of $\pi$ about orthogonal axes generate a dense subgroup of $\SU(2)$~\cite{nielsen2000quantum}, hence the set $\{iW, e^{-i\pi/8}T\}$ is dense in $\SU(2)$.
\end{proof}

\subsection{A Bloch-sphere for the matchgate $\SU(2)$ action}\label{ap:bloch_sphere}
In this appendix, we show that the two-qubit even-parity subspace $\HC_{\rm even}$, under the unitary action of the matchgate subgroup generated by $R^z_1$ and $R_{1,2}^{xx}$, admits a Bloch-sphere description.

Let $\mathcal H = \mathbb C^2 \otimes \mathbb C^2$ denote the two-qubit Hilbert space with computational basis
$\{\ket{00}, \ket{01}, \ket{10}, \ket{11}\}$. The even-parity subspace is defined as
\begin{equation}
    \mathcal H_{\mathrm{even}} := \mathrm{span}_{\mathbb{C}}\{\ket{00}, \ket{11}\}\,.
\end{equation}
This subspace is invariant under the action of $R^{z}$ and $R^{xx}$, as it is easy to check that these operators map even-parity states to even-parity states. The representation of $R^{z}$ and $R^{xx}$ in $\HC_{\rm even}$ in the basis $\{\ket{00},\ket{11}\}$ is given by
\begin{equation}\label{eq-ap:H_even_gates}
    R^{z}_{(\rm even)}(\theta) =\begin{pmatrix}
        e^{i\frac{\theta}{2}} & 0 \\
        0   & e^{-i\frac{\theta}{2}}
    \end{pmatrix} \,, \quad R^{xx}_{(\rm even)}(\theta) = \begin{pmatrix}
        \cos \frac{\theta}{2} & i\sin \frac{\theta}{2} \\  i\sin \frac{\theta}{2} & \cos \frac{\theta}{2}
    \end{pmatrix}\,.
\end{equation}

Denoting $\ket{0_L} :=  \ket{00}$ and $\ket{1_L} :=  \ket{11}$, it is straightforward to identify an isomorphism between $\HC_{\rm even}$ and the Hilbert space of one qubit, $\mathbb C^2$, where the standard representation of $\SU(2)$ acts. This isomorphism, $F:\mathcal H_{\mathrm{even}}\to\mathbb C^2$, is simply the linear map
\begin{equation}
F(\ket{0_L})=\ket{0}\,,\quad
F(\ket{1_L})=\ket{1}\,.
\end{equation}

Under this identification, the matrices in Eq.~\eqref{eq-ap:H_even_gates} are exactly equal to those of the single-qubit rotations around the $Z$ and $X$ axes. Therefore, the matchgate subgroup generated by $R^z_1$ and $R^{xx}_{1,2}$ admits a Bloch-sphere representation, where $\ket{00}$ corresponds to the north pole, $\ket{11}$ to the south pole, the $X$ axis to $X_1X_2$, the $Y$ axis to $Y_1X_2$ and the $Z$ axis to $Z_1$, as in Fig.~\ref{fig:su2_is_su2}.

\subsection{Proof of Theorem 1}\label{ap:th_1}

Here we present a proof for Theorem~\ref{th:universal_matchgates}, which we recall for convinience. 

\begin{theorem}\label{th-ap:universal_matchgates}
    The gate set  $\GC={\left\{\overline  T_q,\,\overline S_q\right\}}_{q=1}^n \cup {\left\{R_{q,q+1}^{xx}\left(\frac{\pi}{2}\right)\right\}}_{q=1}^{n-1}$   is universal for matchgate circuits. 
\end{theorem}
\begin{proof}
    Consider the matchgate Lie algebra $\g$ associated with two qubits, $q$ and $q+1$,  given by 
\begin{equation}
    \mathfrak{g} = \text{span}_\mathbb{R}\langle  iZ_{q}, iZ_{q+1}, iX_{q}X_{q+1}\rangle_{\text{Lie}} =\text{span}_\mathbb{R}i\{Z_q, Z_{q+1}, X_qX_{q+1}, Y_qY_{q+1}, X_qY_{q+1}, Y_qX_{q+1} \}\,,
\end{equation}
where $\langle \cdot\rangle_{\text{Lie}}$ denotes the Lie closure under matrix commutation. Moreover, consider the following Lie subalgebras of $\mathfrak{g}$, 
    \begin{equation}
    \begin{aligned}
        \mathfrak{s}_1 &= \text{span}_\mathbb{R}\langle  iZ_{q}, iX_{q}X_{q+1}\rangle_{\text{Lie}} = {\rm span}_\mathbb{R}i\{Z_q, X_qX_{q+1} \,, Y_qX_{q+1}\}\,, \\\mathfrak{s}_2 &=  \text{span}_\mathbb{R}\langle  iZ_{q+1}, iX_{q}X_{q+1}\rangle_{\text{Lie}}= {\rm span}_\mathbb{R}i\{Z_{q+1}, X_qX_{q+1}, X_qY_{q+1}\}\,.
    \end{aligned}
    \end{equation}
    These are precisely the Lie algebras associated to the rotations depicted in Fig.~\ref{fig:su2_is_su2}.     We can identify the isomorphisms $\phi_1:  \mathfrak{su}(2) \rightarrow \mathfrak{s}_1$, and $\phi_2: \mathfrak{su}(2) \rightarrow   \mathfrak{s}_2$, given by
\begin{equation}\label{eq:isomorphisms}
    \begin{aligned}
        \phi_1(iX )        &= iX_q X_{q+1} \,,  &\quad \phi_2(iX)       &= iX_q X_{q+1} \,, \\
        \phi_1(iY)     &= iY_q X_{q+1} \,,  &\quad \phi_2(iY)    &= iX_q Y_{q+1} \,, \\
        \phi_1(iZ)     &= iZ_q \,,  &\quad \phi_2(iZ)    &= iZ_{q+1} \,,
    \end{aligned}
\end{equation}
which clearly preserve the commutation relations. Since $\mathfrak{su}(2)$ is the Lie algebra of a simply connected Lie group, $\SU(2)$, the Lie algebra homomorphisms $\phi_1$ and $\phi_2$ lift uniquely to Lie-group homomorphisms $\Phi_1$, $\Phi_2$, satisfying 
\begin{equation}\label{eq-ap:group-algebra}
    \Phi_q(\exp A) = \exp (\phi_q(A))\,.
\end{equation}
This allows us to extend the action of $\phi_q$ to the group level, for $q\in \{1,2\}$. 
Let $G_1$ and $G_2$ denote the matchgate Lie subgroups with Lie algebras
$\mathfrak{s}_1$ and $\mathfrak{s}_2$, respectively.
We will use the single-qubit gates
\begin{equation}
    \overline S_q = R^z_{q}\!\left(\frac{\pi}{2}\right)\,,\quad \overline S_{q+1} = R^z_{q+1}\!\left(\frac{\pi}{2}\right) \,,
\end{equation}
which belong to $G_1$ and $G_2$, respectively, and the entangling gate
\begin{equation}
    R^{xx}_{q,q+1}\!\left(\frac{\pi}{2}\right)
= e^{i\frac{\pi}{4}X_qX_{q+1}}
= \frac{1}{\sqrt{2}}(\openone + iX_qX_{q+1})\,,
\end{equation}
which belongs to both $G_1$ and $G_2$, since its generator $iX_qX_{q+1}$ lies in $\mathfrak{s}_1 \cap \mathfrak{s}_2$. Notice as well that $\overline S_q^2 = R^z_{q}(\pi) = iZ_q$ and $\overline S_{q+1}^2 = R^z_{q+1}(\pi) = iZ_{q+1}$. We can now find the image in the matchgate group of the $iW_q$ gate using Eq.~\eqref{eq-ap:group-algebra}, as

\begin{align}
   \overline W_q &= \Phi_q(iW_q) \\&= e^{\phi_q\left(i\frac{\pi}{2}\frac{Y + Z}{\sqrt{2}}\right)}  = e^{i\frac{\pi}{2}\frac{Y_q X_{q+1} + Z_q}{\sqrt{2}}}  \nonumber  \\ &=i\frac{\left(Y_q X_{q+1} + Z_q\right)}{\sqrt{2}} = i\frac{\left(\id + iX_q X_{q+1}\right)}{\sqrt{2}}Z_q  \nonumber  \\   
   &=R^{xx}_{q,q+1}\left(\frac{\pi}{2}\right)R^z_{q}\left(
   \pi\right) \nonumber \\ &=R^{xx}_{q,q+1}\left(\frac{\pi}{2}\right)\overline {S}^2_q\,.
\end{align}
And analogously for $\Phi_{q+1}(iW_{q+1})=R^{xx}_{q,q+1}\left(\frac{\pi}{2}\right)\overline {S}^2_{q+1}$. That is, we can implement $\overline{W}_q$ and $\overline{W}_{q+1}$ exactly using only $R^{xx}_{q,q+1}\left(\frac{\pi}{2}\right)$, $\overline {S}_q$ and $\overline {S}_{q+1}$. Similarly, we can identify 
\begin{equation}
    \overline T_q = \Phi_q\left(e^{-i\frac{\pi}{8}}T_q\right) = e^{i\frac{\pi}{8}Z_q} = R^z_{q}\!\left(\frac{\pi}{4}\right)\,,\quad \overline T_{q+1} = \Phi_{q+1}\left(e^{-i\frac{\pi}{8}}T_{q+1}\right) = e^{i\frac{\pi}{8}Z_{q+1}} = R^z_{q+1}\!\left(\frac{\pi}{4}\right)\,.
\end{equation}
As shown in Lemma~\ref{lem:W,T}, the gates
$e^{-i\frac{\pi}{8}}T$ and $iW$ generate a dense subgroup of $\SU(2)$. Since $\Phi$ is surjective, this implies that the sets $\left\{\overline T_q,\,\overline W_q,\,\overline S_q\right\} \subset G_1$ and $\left\{\overline T_{q+1},\,\overline W_{q+1},\,\overline S_{q+1}\right\} \subset G_2$  are dense in $G_1$ and $G_2$, respectively. Therefore, they generate a dense subgroup of the entire two-qubit matchgate group (which is itself generated by $R^z_{q}(\theta)$, $R^{xx}_{q,q+1}(\theta)$ and $R^z_{q+1}(\theta)$).  Repeating this construction for each pair of nearest neighbors $(q,q+1)$ on an $n$-qubit open chain, we conclude that $\GC = {\left\{\overline  T_q,\,\overline S_q\right\}}_{q=1}^n \cup {\left\{R_{q,q+1}^{xx}\left(\frac{\pi}{2}\right)\right\}}_{q=1}^{n-1}$ is universal for matchgate circuits.
\end{proof}

\subsection{Proof of Proposition 1}\label{ap:entanglement}

Let us now prove Proposition~\ref{prop:entanglement}.

\begin{proposition}
    	Let $R^z(\theta)$  be a single-qubit rotation around the $Z$-axis with some $\theta\in [0, 2\pi)$, and let $U_\varepsilon$ be a two-qubit   matchgate circuit approximating  $R^z(\theta)$ such that
	\begin{equation}
		\|U_\varepsilon- R^z(\theta)\|\leq \varepsilon\,,
	\end{equation}
	where $\|\cdot \|$ denotes the operator norm. Then, the entanglement of $U_\varepsilon$, denoted by $ E(U_\varepsilon)$, is bounded by 
	\begin{equation}
		E(U_\varepsilon) \leq 2\varepsilon^2 + \OC(\varepsilon^4)\,.
	\end{equation}
\end{proposition}

\begin{proof}
    Given a unitary operator $U$ acting on the Hilbert $\HC$ space of $n$ qubits, we can perform a Schmidt decomposition of its (normalized) vectorization, as 
    \begin{equation}
        \frac{1}{\sqrt{2^n}}|U\rangle\!\rangle =\sum_{i = 1}^{r} s_i |a_i \rangle \otimes |b_i\rangle,
    \end{equation}
    where the Schmidt coefficients $s_i$ are real, non-negative,  and satisfy $\sum_i s_i^2 = 1$. The states $\{\ket{a_i}\}_i\in\HC_A$ and $ \{\ket{b_i}\}_i\in\HC_B$ are orthonormal (here we defined $\HC_A\simeq\HC_B\simeq\HC$), and form a complete basis only when the Schmidt rank $r = \dim(\mathcal{H})$. Otherwise, they can readily be completed to a full basis.

    Now consider the vectorization of $ U_\varepsilon$ and let the target unitary be $U= R^z(\theta)\otimes \id$. The Euclidean distance  satisfies
    \begin{equation}
        \frac{1}{\sqrt{2^n}}\| |U_\varepsilon\rangle\!\rangle -|U\rangle\!\rangle \|_2 = \frac{1}{\sqrt{2^n}}\|U_\varepsilon- U\|_{HS} \leq \|U_\varepsilon- U\|\leq \varepsilon\,,
        \label{eq:bound_unitary_dist}
    \end{equation}
    where we used that $\||A\rangle\!\rangle\|_2 = \|A\|_{HS}$ (with $\|A\|_{HS} := \sqrt{\Tr(A^\dagger A)}$ the Hilbert-Schmidt norm), and 
$\|A\|_{HS} \le \sqrt{2^n}\,\|A\|$ for any $2^n\times 2^n$ operator $A$. 
Next, we use that the overlap of any two unit vectors is lower bounded by 
    \begin{equation}
        |\langle \psi |\phi \rangle| \geq 1- \frac 12 \|\ket{\psi}-\ket{\phi}\|_2^2 \,,
        \label{eq:bound_unit_vectors}
    \end{equation} 
which follows from the identity $\mathrm{Re}(\langle \psi |\phi \rangle) = 1- \frac{1}{2}\|\ket{\psi}- \ket{\phi} \|_2^2$ and the inequality $\mathrm{Re}(z) \le |z|$ for any complex number $z$. Applying Eq.~\eqref{eq:bound_unit_vectors} to the vectorized $U_\varepsilon$ and $U$, we find
    \begin{equation}
        \frac{1}{2^n}|\langle\!\langle U |U_\varepsilon\rangle\!\rangle| \geq 1- \frac{1}{2^{n+1}} \||U_\varepsilon\rangle\! \rangle -|U \rangle\! \rangle \|_2 ^2 \geq 1- \frac 12 \varepsilon^2 \,,
    \end{equation}
    where we have used Eq.~\eqref{eq:bound_unitary_dist} in the last inequality.

   The largest Schmidt coefficient $s_1$ of  $\frac{1}{\sqrt{2^n}}|U\rangle\!\rangle$ satisfies
\begin{equation}
s_1(U) := \frac{1}{\sqrt{2^n}}\max_{\substack{X = X_A \otimes X_B \\ \|X\|_{HS}=1}}
    |\langle\!\langle X | U \rangle\!\rangle|\,,
\label{eq:largest_schmidt}
\end{equation}
    where the maximization is over bipartite operators
    $X_A \in \mathcal{B}(\mathcal{H}_A)$ and $X_B \in \mathcal{B}(\mathcal{H}_B)$
    such that $\|X_A \otimes X_B\|_{HS}=1$. Here, $\mathcal{B}(\mathcal{H})$ denotes the space of bounded linear operators acting on $\HC$.     To see this, expand $|X\rangle\!\rangle$ in the orthonormal Schmidt basis $\{|a_i\rangle \otimes |b_i\rangle\}$. Since we are maximizing the overlap with $|U\rangle\! \rangle$, the maximum is achieved when $| X\rangle\! \rangle $  aligns with the Schmidt vector corresponding to the largest coefficient $s_1$. That is, when $X$ is a product operator $X_1\otimes X_2$ such that $|X\rangle\! \rangle = |a_1\rangle \otimes |b_1\rangle$. It follows that 
    \begin{equation}
        s_1 (U_\varepsilon) =  \frac{1}{\sqrt{2^n}} \max_{\substack{X = X_A \otimes X_B \\ \|X\|_{HS}=1}} |\langle\!\langle X| U_\varepsilon\rangle\!\rangle | \geq  \frac{1}{2^n}|\langle \!\langle U| U_\varepsilon\rangle\!\rangle|\geq 1- \frac 12 \varepsilon^2\,.
    \end{equation}
    Finally, we can bound the entanglement of $U_\varepsilon$ as 
    \begin{equation}
        E(U_\varepsilon) = 1- \sum_{i =1}^r s_i^4 \leq 1- s_1 ^4\leq 1-{\left(1-\frac 12 \varepsilon^2 \right)}^4 = 2\varepsilon^2 + \OC\left(\varepsilon^4\right)\,,
        \label{eq:bound_ent_proof}
    \end{equation}
    which completes the proof. 
\end{proof}

\setcounter{lemma}{2}

\subsection{Proof of Theorem 2}\label{ap:th_2}

Next, we prove Theorem~\ref{th:SO-error}, that we here recall for convenience.

\begin{theorem}
    Let $Q_{\varepsilon}\in\SO(2n)$ be an approximation to $Q\in\SO(2n)$ such that $||Q-Q_\varepsilon|| = \varepsilon_{\SO(2n)}\ll1$. Then $||\Phi^{-1}(Q)-\Phi^{-1}(Q_\varepsilon)||\leq \frac{\pi}{2} n\, \varepsilon_{\SO(2n)}$.
\end{theorem}

\begin{proof}
    Let us denote $R=Q^TQ_\varepsilon$. Using the invariance of the operator norm under multiplication by a unitary matrix, we can rewrite the approximation condition as
    \begin{equation}
        ||Q-Q_\varepsilon|| = ||R-\id|| = \varepsilon_{\SO(2n)}\,.
    \end{equation}
    Now, since $R$ is by construction a special orthogonal matrix, its spectrum is ${\rm spec}(R)=\left\{e^{\pm i\lambda_j}\right\}_{j=1}^n$, with $\lambda_j \in [0,\pi]$. Moreover, the operator norm of a normal matrix~\footnote{A normal matrix is a matrix that commutes with its dagger.}, such as $R$, is the maximum of the absolute values of its spectrum. Hence, using the identity $e^{i\phi}-1=2ie^{i\frac{\phi}{2}}\sin(\frac{\phi}{2})$, we get
    \begin{equation}
         \varepsilon_{\SO(2n)} = ||R-\id|| = 2\max_j\left( \sin \frac{\lambda_j}{2}\right) = 2\sin\frac{\max_j\lambda_j}{2}\,,
    \end{equation}
    where in the last equality we used the monotonicity of the sine function in the interval $\left[0,\frac{\pi}{2}\right]$.
    We thus find that the approximation error incurred in the synthesis of $Q\in\SO(2n)$ is directly related to the maximal rotation angle of $R$. 
Using the elementary inequality $\sin x \ge \frac{2}{\pi}x$ for $x\in\left[0,\frac{\pi}{2}\right]$, this implies 
\begin{equation}\label{ap_eq:error_R_to_max_angle}
         \varepsilon_{\SO(2n)}  \ge \frac{4}{\pi}\frac{\max_j\lambda_j} {2}  \quad\Longrightarrow\quad \max_j\lambda_j\le \frac{\pi}{2}\,\varepsilon_{\SO(2n)}\,.
    \end{equation}
    Next, consider the $\mathbb{SPIN}(2n)$ superoperators $U \otimes U^*= \Phi^{-1}(Q)$ and $U_\varepsilon\otimes U_\varepsilon^*=\Phi^{-1}(Q_\varepsilon)$. To assess the error $||U\otimes U^*-U_\varepsilon\otimes U_\varepsilon^*||$, we can again study $||V\otimes V^*-\id\otimes\id||$, where now $V=U^\dagger U_\varepsilon$. Since the lift from $\SO(2n)$ to $\mathbb{SPIN}(2n)$ is a group homomorphism, we have that $V\otimes V^*=\Phi^{-1}(R)$. Specifically, if $R=e^h$, for some $h\in\so(2n)$, then $V\otimes V^*=e^{\varphi^{-1}(h)\otimes\id+\id\otimes(\varphi^{-1}(h))^*}$, with $\varphi^{-1}(h)=\frac{1}{4}\sum_{\mu,\nu}h_{\mu,\nu}c_\mu c_\nu$.
    Without loss of generality, we can choose a basis of $\so(2n)$ where $h$ is in the canonical form 
    \begin{equation}
        h=\bigoplus_{j=1}^n \begin{pmatrix}0&\lambda_j\\-\lambda_j&0\end{pmatrix}\,.
    \end{equation}
    Indeed, this amounts to conjugating $R$ by a suitable orthogonal matrix, which in turn leaves the error $||R-\id||$ unchanged.
    Upon this choice of basis, the generator of $V$ reads $\varphi^{-1}(h)=i\sum_{j=1}^n \frac{\lambda_j}{2}Z_j$, whose eigenvalue with maximum modulus is $\lambda_{\rm max}=i\sum_{j=1}^n \frac{\lambda_j}{2}$. Furthermore, we have $(\varphi^{-1}(h))^* = -\varphi^{-1}(h)$.
    Now, let us use the standard integral identity\footnote{A proof: $e^{X}-\mathbb{I}=\int_{0}^{1}\frac{d}{dt}\,e^{tX}\,dt=\int_{0}^{1}Xe^{tX}\,dt=\int_{0}^{1}e^{tX}X\,dt$.}
    \begin{equation}
        e^X - \id = \int_0^1e^{tX} X dt\,,
    \end{equation}
    which, upon taking the operator norm on both sides, and assuming $X$ anti-Hermitian (i.e., $e^X$ unitary), yields
    \begin{equation}
        \big|\big|e^X - \id\big|\big| = \Big|\Big|\int_0^1 dt \, e^{tX} X \Big|\Big| \leq \int_0^1 dt\,\big|\big|e^{tX}\big|\big|\,||X||=||X|| \,,
    \end{equation}
    where we used the triangle inequality and the sub-multiplicativity of the operator norm, and the fact that unitaries have operator norm equal to one.
    We can use this to find
    \begin{equation}
        ||V\otimes V^*-\id\otimes\id|| = ||e^{\varphi^{-1}(h)\otimes\id - \id\otimes\varphi^{-1}(h)}-\id\otimes\id||\leq ||\varphi^{-1}(h)\otimes\id - \id\otimes\varphi^{-1}(h)||=2|\lambda_{\rm max}|\,,
    \end{equation}
    and since
    \begin{equation}
        2|\lambda_{\rm max}|=\sum_{j=1}^n \lambda_j\leq n\max_j\lambda_j\,,
    \end{equation}
    by using Eq.~\eqref{ap_eq:error_R_to_max_angle} we arrive at
    \begin{equation}
        ||V\otimes V^*-\id\otimes\id||\leq n \max_j\lambda_j \leq \frac{\pi}{2} n\, \varepsilon_{\SO(2n)}\,,
    \end{equation}
    which concludes the proof.
\end{proof}

\subsection{Isomorphism between $\so(4)$ and $\su(2)\oplus \su(2)$}\label{ap:so(4)_isomorphism}
The Lie algebra associated to two-qubit matchgates circuits is given by 
\begin{equation}
     \mathfrak{g} = \text{span}_\mathbb{R}i\{Z_1, Z_2, X_1X_2, Y_1Y_2, X_1Y_2, Y_1X_2 \}.
\end{equation}
Is it known that $\mathfrak{g}$ is isomorphic to $\so(4)$. Here, we also show that $\so(4)\cong \su(2)\oplus\su(2)$. This follows by defining the following operators 
\begin{align}
    J_1^+ = \frac 12 \left( Y_1 Y_2 - X_1 X_2\right),  &\quad J_2^+ = \frac {-1}{2} \left( X_1 Y_2 + Y_1 X_2\right), \quad 
     J_3^+ = \frac {1}{2} \left( Z_1 + Z_2\right) \\
    \nonumber J_1^- = \frac 12 \left( Y_1 Y_2 + X_1 X_2\right) &\quad J_2^- = \frac {1}{2} \left( X_1 Y_2 - Y_1 X_2\right), \quad  J_3^- = \frac {-1}{2} \left( Z_1 - Z_2\right).
\end{align}
By using the identity $[\sigma_j, \sigma_k]=2i\varepsilon_{jkl}\sigma_l$, where $\sigma_1 = X$, $\sigma_2 = Y$ and $\sigma_3 =Z$. By direct calculation we can obtain the following commutation relations:
\begin{align}
    &[J_j^+, J_k^+] = 2i \varepsilon_{jkl}J_l^+ \\
    &\nonumber[J_j^- ,J_k^-] = 2i \varepsilon_{jkl}J_l^-\\ &\nonumber[J_j^+, J_k^-] = 0.
\end{align}
This implies that both $\{J_j^+\}$ and $\{J_j^-\}$ form subalgebras isomorphic to $\mathfrak{su}(2)$, since they satisfy the same commutation relations. Also, since each $J_j^+$ commutes with all $J_k^-$, the two subalgebras commute and can be simultaneously block-diagonalized. We can define the explicit isomorphism $\varphi: \mathfrak{su}(2) \oplus \mathfrak{su}(2) \longrightarrow \mathfrak{so}(4) $, so that
\begin{equation}
    \begin{split}
       \varphi(i X, 0) = i J_1^+, \quad \varphi(0, i X) = i J_1^-, \quad  \\
       \varphi(i Y, 0) = i J_2^+, \quad \varphi(0, i Y) = i J_2^-, \quad  \\
       \varphi(i Z, 0) = i J_3^+, \quad \varphi(0, i Z) = i J_3^-. \quad  
    \end{split}
\end{equation}

This Lie algebra isomorphism is induced by the group homomorphism $\Phi: \SU(2) \times \SU(2) \to \SO(4)$ via differentiation at the identity, denoted as $\varphi = \mathrm{d}\Phi$. The relation between the Lie groups and Lie algebras is summarized in the following commutative diagram:
\begin{equation}\label{eq:diagram_SU2_SU4}
    \begin{tikzcd}
\SU(2) \times \SU(2) \quad  \arrow[r, "\Phi"] \arrow[d, "\mathrm{Lie}"'] &\quad  \SO(4) \arrow[d, "\mathrm{Lie}"] \\
\mathfrak{su}(2) \oplus \mathfrak{su}(2) \quad \arrow[r, "\varphi"] &\quad  \mathfrak{so}(4)
\end{tikzcd}\,.
\end{equation}

With the previous isomorphism we can express all elements in $\SO(4)$ in terms of $\SU(2)\times \SU(2)$. Since the group homomorphism $\Phi$ is not injective (it is a 2-to-1 covering map), it does not admit a globally defined inverse. Consequently, we need to work at the Lie algebra level first, and apply $\varphi^{-1}$ (i.e., go around the diagram in Eq~\eqref{eq:diagram_SU2_SU4} counterclockwise). 

As a concrete example, consider the task of  compiling the matchgate $R^z(\theta)\otimes R^z(\theta)\in \SO(4)$. Here, one can express this unitary as two elements of $\SU(2)$ 
\begin{equation}
    \Phi ^{-1} (R^z(\theta) \otimes R^z(\theta)) = \exp(\varphi^{-1} (\theta i(Z_1 + Z_2)/2) )=\\ \exp( \theta \varphi^{-1} (i J_3^+) ) = \exp( \theta(iZ, 0)) = (R^z(2\theta), \id) , 
\end{equation}
Now, we can compile $R^z(2\theta) \in \SU(2)$ with standard techniques, such as the \texttt{gridsynth} algorithm~\cite{ross2014optimal}. The compilation will return a sequence of gates from the set $\{iH, iX, e^{-i \pi/8}T, e^{-i \pi/4}S\}$, where the phases ensure that the gates are in $\SU(2)$. The sequence of gates achieves a given accuracy $\varepsilon$ with respect to $R^z(2\theta)$. Since our goal is to compile $R^z(\theta)\otimes R^z(\theta)\in \SO(4)$ using only matchgates, we need to map the elements of the set generating $\SU(2)$ back to $\SO(4)$. To do so, we apply $\Phi$, i.e., move clockwise in the diagram. For a $T$ gate one has 
\begin{equation}\label{eq:sqrt(T)}
        \Phi [(e^{-i \pi/8} T,  \id)] = \exp \left( \frac \pi 8 \varphi(iZ, 0)\right)  =  \exp \left( i\frac \pi 8 J_3^+ \right) = \exp \left( i\frac {\pi} {16} (Z_1+Z_2) \right) = \sqrt{T}\otimes \sqrt T \,,
\end{equation}
and for a the Hadamard gate
\begin{equation}\label{eq:H_SO4}
\begin{split}
        &\Phi[ (i H , \id)] = \exp \left( \frac \pi 2 \varphi \left(\frac{iX+iZ}{\sqrt 2}, 0 \right)\right)  =   \exp \left( i\frac{\pi}{4  }( J_3^+ +J_1^+)  \right) = \\ &\exp \left( i\frac{\pi}{4 \sqrt 2 }(Z_1+Z_2 + Y_1Y_2 -X_1X_2) \right) = \exp \left( i\frac{\pi}{4 \sqrt 2 }(Z_1 + Y_1Y_2) \right) \exp \left( i\frac{\pi}{4 \sqrt 2 }(Z_2 - X_1X_2) \right).
\end{split}
\end{equation}
We can proceed similarly for $iX$ and $e^{-i \pi/4}S$.

Although the resulting expressions in Eqs.~\eqref{eq:sqrt(T)} and \eqref{eq:H_SO4} are indeed matchgates, they will not be native to a given device. It remains an open, and device-specific, problem to identify an alternative universal gate set in $\SU(2)$ whose associated matchgates are native.

\section{Exact matchgate synthesis}\label{ap:exact_synthesis}

\subsection{Proof of Theorem 3}\label{ap:th_3}

We here present the proof of Theorem~\ref{th:exact_synthesis}.

\begin{theorem}
    All matchgate unitaries $U\in\mathbb{SPIN}(2n)$ such that $U\otimes U^*$ has entries in the ring $\ZBB\left[\frac{1}{\sqrt{2}},i\right]$ can be exactly synthesized by a sequence of gates from the set $\GC={\left\{\overline  T_q,\,\overline S_q\right\}}_{q=1}^n \cup {\left\{R_{q,q+1}^{xx}\left(\frac{\pi}{2}\right)\right\}}_{q=1}^{n-1}$.
\end{theorem}

\begin{proof}
The first step in our proof strategy is to show that every unitary $U\in\mathbb{SPIN}(2n)$ with entries in the ring $\ZBB\left[\frac{1}{\sqrt{2}},i\right]$ corresponds to a matrix in $\SO(2n)$ with entries in the ring $\mathbb D\left[\sqrt{2}\right]$, defined by
\begin{equation}
    \mathbb D\left[\sqrt 2\right] := \{a + b\sqrt{2} \;\, |\;\, a,b\in \mathbb D \}\,.
    \label{eq-ap:D_sqrt2_ring}
\end{equation}
Above, $\mathbb D$ is the dyadic ring, $\mathbb D := \left\{\frac{m}{2^k} \;|\; m\in \mathbb Z, k\in \mathbb N \right\}$. 
This is proven in the following lemma.
\begin{lemma}
    If $U\in\mathbb{SPIN}(2n)$ is such that $U\otimes U^*$ has matrix entries in the ring $\mathbb Z\left[\frac{1}{\sqrt{2}},i\right]$, then $\Phi(U\otimes U^*)\in\SO(2n)$ has matrix entries in the ring $\mathbb D\left[\sqrt{2}\right]$.
\end{lemma}

\begin{proof}
    Since $U\otimes U^*$ has entries in  $\mathbb Z\left[\frac{1}{\sqrt{2}},i\right]$, and so do the Majorana operators $c_\mu$ under the Jordan--Wigner transformation, it follows that
    \begin{equation}
        2^{-n}\,\Tr\left[c_\nu Uc_\mu U^\dagger\right]=2^{-n}\,\Tr\left[c_\nu \sum_{\nu'} Q_{\mu\nu'}c_{\nu'}\right] =  Q_{\mu \nu}
    \end{equation} 
    also has entries in $\mathbb Z\left[\frac{1}{\sqrt{2}},i\right]$. Moreover, because ${\Phi(U\otimes U^*)}_{\mu\nu}=Q_{\mu\nu}$ is a matrix in the standard representation of $\SO(2n)$, we know that its entries are real. The set of numbers in $\mathbb Z\left[\frac{1}{\sqrt{2}},i\right]$ is
    \begin{equation}
        \left\{\frac{a + bi + \frac{c}{\sqrt{2}} +  \frac{di}{\sqrt{2}}}{\sqrt{2}^k}\quad \Big| \quad a,b,c,d\in\mathbb{Z},\;k\in\mathbb{N} \right\}\,,
    \end{equation}
    so for a number in $\mathbb Z\left[\frac{1}{\sqrt{2}},i\right]$  to be real, we must have $b=d=0$. We thus conclude that $Q_{\mu\nu}$ has entries of the form 
    \begin{equation}
        \left\{\frac{a + c/\sqrt{2}}{\sqrt{2}^k}\quad \Big| \quad a,c\in\mathbb{Z},\;k\in\mathbb{N} \right\} = \left\{\frac{c+a\sqrt{2}}{\sqrt{2}^{k+1}}\quad \Big| \quad a,c\in\mathbb{Z},\;k\in\mathbb{N} \right\}\in\mathbb D\left[\sqrt{2}\right]\,.
    \end{equation}
    
\end{proof}

Notice that $\overline T := e^{i\pi/8} T\in \SU(2)$ does not have entries in the $\mathbb Z\left[\frac{1}{\sqrt{2}},i\right]$ ring due to the global phase. However, when represented as a matchgate, the global phase is irrelevant due to the isomorphism in Eq.~\eqref{eq:group_isomorphism}. 

Next, we show that every matrix in $\SO(2n)$  with entries in the ring $\mathbb D\left[\sqrt 2\right]$ can be exactly synthesized by a finite sequence of gates from the set $\tilde \GC$, which is the image in $\SO(2n)$ of the gates in $\GC={\{\overline  T_q,\,\overline S_q\}}_{q=1}^n \cup {\left\{R_{q,q+1}^{xx}\left(\frac{\pi}{2}\right)\right\}}_{q=1}^{n-1}$. This step concludes the proof, and it is inspired by the techniques presented in Ref.~\cite{giles2013exact}, although a direct application of their results is not possible in our setting. We therefore provide a complete and self-contained proof below.

\begin{lemma}\label{th:exact_synthesis_D2}
    All matrices in $\SO(2n)$ with entries in the ring $\mathbb D\left[\sqrt 2\right]$, as defined in Eq.~\eqref{eq-ap:D_sqrt2_ring}, can be exactly synthesized by a finite sequence of gates from the set $\tilde \GC$.  
\end{lemma}

\begin{proof}

    We first notice that each gate in the set $\tilde \GC$ has entries in $\{ \pm 1/\sqrt 2, \pm 1 \}$. Hence, all finite product of such gates result in matrices with entries in the ring $\mathbb D\left[\sqrt 2\right]$. This is straightforward to see, as matrix multiplication only involves addition and multiplication of real numbers. However, the converse statement is not obvious. 
   
   The general idea of the proof is to show that for every matrix $Q \in \SO(2n)$ with entries in $ \mathbb D[\sqrt 2]$, we can find a sequence of gates  $S_1, \dots , S_l$ in $\tilde \GC$ such that $S_l^T\cdots S_1^T Q = \id$, which implies $Q=S_1\cdots S_l$. To do so, we proceed column by column, showing that if $\vec{z}={(z_1,z_2,\dots,z_{2n})}^T$ is the first column of some $Q\in \SO(2n)$ with entries in $ \mathbb D[\sqrt 2]$, we can always find a sequence of gates in $\tilde \GC$ that transform $\vec{z}$ into $\vec{e_1} = {(1, 0, \dots , 0)}^T$, the first vector of the canonical basis in $\mathbb R^{2n}$. 
   Once we have transformed  $\vec{z}$ into $\vec{e_1}$, we are left with
   \begin{equation}\label{eq:column_reduction}
   Q\longrightarrow
        \begin{pmatrix}
            \begin{array}{c|c}
                1 & 0 \\
                \hline
                0 & Q' \\
            \end{array}
        \end{pmatrix}\,.
   \end{equation}
    Notice here that the first row in~\eqref{eq:column_reduction} is $\vec{e_1}^T$, because the gates in $\tilde \GC$ are orthogonal matrices that preserve the normalization of both rows and columns. Then, we apply the same procedure to the first column of $Q'$, and so forth until we transform $Q$ into the identity matrix. If we show that we can achieve this transformation for any $Q\in \SO(2n)$ with entries in $\mathbb D[\sqrt 2]$ using only gates from $\tilde \GC$, we will have proven Theorem~\ref{th:exact_synthesis_D2}.

   Before approaching the proof, though, we need some definitions. First, the following two rings will be useful,
   \begin{equation}
       \mathbb Z\left[\sqrt 2\right] := \left\{a + b\sqrt{2} \; |\; a,b\in \mathbb Z \right\}\,, \quad  \mathbb Z_2\left[\sqrt 2\right] := \left\{a + b\sqrt{2} \; |\; a,b\in \mathbb Z_2 \right\}\,,
   \end{equation}
   where $\mathbb Z_2 = \mathbb Z/2\mathbb Z$ is the ring of integers modulo $2$. The denominator exponent plays an important role in the proof, so we include its definition. 
\begin{definition}\label{def:denominator_exponent_app}
       Given $r\in \mathbb D\left[\sqrt 2\right]$, we say that $k$ is a denominator exponent of $r$ if $\sqrt 2^k r\in \mathbb Z\left[\sqrt 2\right]$. If $k$ is a denominator exponent of $r$ and $k-1$ is not, i.e. $\sqrt2^{k-1}r \notin \mathbb Z\left[\sqrt 2\right]$, we say that $k$ is the least denominator exponent of $r$. 
   \end{definition} 
   For example, say $r=\frac{1+\sqrt{2}}{2}$, then the least denominator exponent of $r$ is $k=2$ because $\sqrt{2}^2\,r = 1+\sqrt{2}\in \mathbb Z\left[\sqrt{2}\right]$, but $\sqrt{2}\,r=\frac{1+\sqrt{2}}{\sqrt{2}}=1+\frac{1}{2}\sqrt{2}\notin \mathbb Z\left[\sqrt{2}\right]$. 
    We also use the concept of denominator exponent of a matrix or a vector, which we define to be a denominator exponent for all entries of the matrix or vector. 
    The last ingredient that we need is a ring homomorphism $\mathscr P:\mathbb Z\left[\sqrt 2\right]\rightarrow \mathbb Z_2\left[\sqrt 2\right]$, defined by 
    \begin{equation}
        \Par: a + b\sqrt 2 \;\longrightarrow\;  \tilde a + \tilde b\sqrt 2, \quad \text{where}\quad  \tilde a :=  a \;\,{\rm mod}\; 2\,, \; \tilde b :=  b \;\, {\rm mod}\; 2\,.
    \end{equation}
   The homomorphism basically maps each element in $\mathbb Z$ to its parity in $\mathbb Z_2$, respecting the ring structure given by the addition and multiplication operations. This follows from the fact that the parity of the sum (product) is the sum (product) of the parities. To ease notation, we denote $\tilde a + \tilde b \sqrt2 :=  \tilde a \tilde b$. For instance, $\Par\left(1+2\sqrt2\right) = 10$. Furthermore, following the nomenclature of Ref.~\cite{giles2013exact}, we call $\tilde z = \Par(z)$ the residue of $z$.
   \medskip

   Now, we enunciate and prove the lemma that allows us to turn the first column of $Q$ into $\vec{e_1}$, as  in Eq.~\eqref{eq:column_reduction}. 
   
   \begin{lemma}\label{lem:column_lemma}
       Consider a normalized vector $\vec{z}\in {\mathbb D\left[\sqrt 2\right]}^{2n}$ corresponding to a column of a matrix in $\SO(2n)$. Then, there exists a sequence $S_1,  \dots, S_l$ of elements in $\tilde \GC$ such that $S_l^T \cdots  S_1^T \vec{z} = \vec{e_1}$.  
   \end{lemma}
   \begin{proof}
         Since the entries of $\vec{z}$ are in $\mathbb D\left[\sqrt 2\right]$, we can readily find its least denominator exponent $k$, the smallest integer such that $\sqrt2^k z_i\in \mathbb Z\left[\sqrt 2\right]$ for all $i\in \{1,\dots, 2n\}$. The proof then proceeds by induction on the denominator exponent $k$. That is, we first show that when  $k=0$, there always exists a sequence of elements in $\tilde \GC$ that transforms $\vec{z}$ into $\vec{e_1}$. Next, we show that for every $k>0$ there always exists a sequence of gates in $\tilde \GC$ that lowers the least denominator exponent from $k$ to $k-1$, which guarantees that we can always reach the base case $k=0$. The two inductive steps are: 
        \begin{itemize}
            \item Base case $k = 0$: Because of the normalization condition on $\vec{z}$ ($\|\vec{z}\|_2^2 = \sum_{i = 1}^{2n}z_i^2 =1$), 
            we know that just a single $z_i$ will be equal to $\pm1$, and therefore $z_i \in \mathbb Z\left[\sqrt 2\right]$ for all $i\in \{1,\dots, 2n\}$. We can use the Clifford gates $\tilde S$ and $\tilde R$ to permute that $z_i = \pm1$ to the first position, yielding $\pm \vec{e_1}$. We can also use $\tilde S^2=-\id$ (which is also Clifford) to invert the sign of $\vec{e_1}$ if needed.
            
            \item Inductive step $k>0$: Denote $w_j :=  \sqrt 2^k z_j \in \mathbb Z\left[\sqrt2\right]$.  Since $\sum_{i = 1}^{2n} z_i^2 = 1$, then $\sum_{i = 1}^{2n} w_i^2 = 2^k$ (which is just the normalization condition again). If we take residues on both sides, i.e. we apply the parity ring homomorphism $\Par$, we have 
            \begin{equation}
                 \sum_{i = 1}^{2n} \Par\left(w_i^2\right) = \Par\left(2^k\right) = 00\,.
             \label{eq:sum_residues_squared}
            \end{equation}
            This imposes a constraint on the sum of the residues $\Par(w_i^2)$. For each $w_i = a_i + b_i \sqrt 2$, with $a_i, b_i \in \mathbb Z$, we have
            \begin{equation}
                w_i^2 = \underbrace{a_i ^2 +2b_i^2}_{a_i'} \, +\, \underbrace{2a_i b_i}_{b_i'} \sqrt 2 \,.
            \end{equation}
            Since $\Par(w_i)$
            can only take four values, we can enumerate the possible values of $\Par(w_i^2)$ accordingly, as shown in Table~\ref{tab:rho(w_i)^2}.
            \begin{table}[h]
                \centering
                \begin{tabular}{c|c}
                    $\Par(w_i)\,$ & $\,\Par(w_i^2) $\\
                    \hline
                    00 & 00 \\
                    01 & 00 \\
                    10 & 10 \\
                    11 & 10 \\
                \end{tabular}
                \caption{Residue map for squared entries.}\label{tab:rho(w_i)^2}
            \end{table}
            
            The values of $\Par(w_i^2)$ are restricted to just two possibilities, $00$ or $10$, as a consequence of the fact that the parity of $b'_i=2a_ib_i$ is always even. To satisfy Eq.~\eqref{eq:sum_residues_squared}, the number of terms with $\Par(w_i^2) = 10$ in the sum must be even, so that they cancel pairwise and the total sum yields $00$. And because the length of $\vec{z}$ is $2n$, there will also be an even of number of terms with $\Par(w_i^2)=00$.
            
            Next, we analyze the possible values for $\Par(w_i^2)$ and determine that there always exist gates in $\tilde \GC$ that can reduce the denominator exponent from $k$ to $k-1$. To do this, we first introduce the concept of reducibility.\medskip
            
            \begin{minipage}[ht]{\linewidth}
                \begin{definition}
                    A residue $\Par(w)\in \mathbb Z_2[\sqrt 2]$ is \textit{reducible} if it can be expressed as $\Par(w) = \sqrt 2 \Par(w')$ for some \\$\Par(w') \in  \mathbb Z_2[\sqrt 2]$. We call $\Par(w)$ \textit{twice reducible} if $\Par(w) =  2 \Par(w')$.
                \end{definition}
            \end{minipage}\medskip
            
            Intuitively, a reducible residue still carries an overall factor of $\sqrt 2$, so we can lower its denominator exponent $k$ (by one in the reducible case, or by two in the twice reducible case) without the need for any non-Clifford operation. In contrast, if a residue is irreducible, then $k$ is already the least denominator exponent. As we will see, further reducing this exponent necessarily requires the application of a $T$ gate.
            
            Then, the following statements hold: 
            \begin{itemize}
                \item  $\Par(w)$ is reducible $\iff$  $w/\sqrt2 \in \mathbb Z[\sqrt 2]$. \medskip
                
                The  left implication is easy to see, as if $w':=  w/\sqrt2 \in \mathbb Z[\sqrt 2]$, then $\Par(w) = \Par( \sqrt 2 w') =\sqrt 2 \Par(w')$, which is the definition of reducibility for $\Par(w)$. 
                The right implication can be proven by considering that, by definition, reducibility implies  $\Par(w) = \sqrt 2\Par( w')\,\Rightarrow\, \sqrt{2}\Par(w)=2\Par(w')\,\Rightarrow\,\sqrt{2}(\tilde{a}+\tilde{b}\sqrt{2})=\tilde{2b}+\tilde{a}\sqrt{2} = 00$ for some $\Par(w') \in  \mathbb Z_2[\sqrt 2]$, and therefore $\tilde{a}=0$. That is, $w = 2c + b\sqrt{2}$ with $c,b\in\ZBB$, and it follows that               
                $\frac{w}{\sqrt{2}} = b+c\sqrt{2} \in \mathbb Z[\sqrt2]$.
                
                \item $\Par(w)$ is reducible $\iff$ $\Par(w^2) = 00$.\medskip 
                
                The left implication follows because $\Par(w^2)=\Par(a^2+2b^2+(2ab)\sqrt{2})=\tilde a 0$, and the previous statement implies $\tilde a =0$ if $\Par(w)$ is reducible. 
                The right implication similarly follows from the fact that if $\Par(w^2)  = 00$, then $\tilde{a}=0$ as per Table.~\ref{tab:rho(w_i)^2}. Hence, $w/\sqrt{2}\in\ZBB[\sqrt{2}]$ and $\Par(w)$ is reducible.            
               
                \item  $\Par(w)$ is twice reducible $\iff$  $\Par(w) = 00$.\medskip 
                
                For the left implication, notice that both components of $w$ are even when $\Par(w) = 00$, i.e., $\Par(w)=\Par(2w')=2\Par(w')$ for some $\Par(w') \in  \mathbb Z_2[\sqrt 2]$, so $\Par(w)$ is twice reducible. Conversely, if $\Par(w)$ is twice reducible, then $\Par(w)=2\Par(w')=00$ for some $\Par(w') \in  \mathbb Z_2[\sqrt 2]$. 
            \end{itemize}

           We now proceed to analyze the possible values that $\Par(w_i^2)$ can take.
            \begin{enumerate}
                \item  $\Par(w_i^2) = 00$:  \medskip
                
                In this case, $\Par(w)$ is already reducible, and we can just lower $k$ to $k-1$ by dividing $z_i $ by $\sqrt 2$. \medskip
                
                \item $\Par(w_i^2) = 10$:  \medskip
                
                By looking at Table~\ref{tab:rho(w_i)^2}, this case can happen when  $\Par(w_i)\in \{10, 11\}$.  As there needs to be an even number of $w_i$ such that  $\Par(w_i^2) = 10$ to satisfy Eq.~\eqref{eq:sum_residues_squared}, we can always group those $w_i$ into pairs. Therefore, we need to explore all possible combinations of pairs $\Par(w_i)$ and $\Par(w_j)$ with $i\neq j$, which we do next. 

                The first case we consider is $\Par(w_i) = \Par(w_j) = 10$ or  $\Par(w_i) = \Par(w_j) = 11$. Here, notice that we can always apply a signed permutation $P$ using $\tilde S$ and $\tilde R$ gates to place $i,j$ next to each other (say, in positions $i',i'+1$); and that parity does not change under multiplication by $-1$, i.e. $\Par(-w_i)=\Par(w_i)$. Then, we can apply a $\tilde T_{i,i'+1}$ gate, obtaining $\vec{v} =\tilde T_{i,i'+1}  P  \vec{z}$.  Now consider 
                \small
                \begin{align}
                \qquad\qquad\qquad\Par\left(\sqrt 2 \left(\sqrt2^{k} \vec{v}\right)\right) = \Par\left(\sqrt  2^k \begin{pmatrix}
                        1 & 1 \\ -1 & 1 
                    \end{pmatrix}_{i',i'+1} P \vec{z}\right) =\Par\left(\sqrt 2^k \begin{pmatrix}
                        z_i+z_j \\-z_i+ z_j
                    \end{pmatrix}\right)=\begin{pmatrix}
                         \Par(w_i) +\Par(w_j) \\ -\Par(w_i) +\Par(w_j)
                    \end{pmatrix} = \begin{pmatrix}
                        00 \\00
                    \end{pmatrix}\,,
                    \label{eq:T_gate_red}
                \end{align}
                \normalsize
                where we have used that $\Par(w_i) +\Par(w_j) = -\Par(w_i) +\Par(w_j) =00$ when $\Par(w_i) = \Par(w_j ) \in\{10,11\}$. This proves that  $\Par\left(\sqrt 2 \left(\sqrt2^{k} \vec{v}\right)\right)$ is twice reducible, and therefore $\Par\left(\sqrt2^{k} \vec{v}\right)$ is reducible. Hence, by applying the appropriate gate sequence $\tilde T_{i,j}  P$ from $\tilde \GC$, we are always able to lower by one the denominator exponent of the subspace spanned by  $z_i$ and $z_j$ when $\Par(w_i) = \Par(w_j ) \in\{10,11\}$.

                Finally, we show that there cannot be an odd number of $w_i$ such that $\Par(w_i) = 10$, or such that $\Par(w_j) = 11$. To see this, let us look again at the normalization condition of the vector $\vec{z}$,
                \begin{equation}
                  \sum_{i = 1}^{2n} w_i^2  =  \sum_{i = 1}^{2n} (a_i^2 + 2b_i^2 + 2a_ib_i \sqrt{2}) = 2^k \,\implies\, \sum_{i = 1}^{2n}\tilde a_i\tilde b_i = 0 \mod 2\,,
                  \label{eq:parity_condition}
                \end{equation}
                 where the implication follows from the requirement that the irrational part (proportional to $\sqrt{2}$) must cancel out. 
                 Consider $\Par(w_i) = 10$, which implies  $\tilde a_i=1,\, \tilde b_i =0$; and $\Par(w_j) = 11$, therefore $\tilde a_j=1,\, \tilde b_j=1$. Then, the two products give $\tilde a_i \tilde b_i = 0 $ and $\tilde a_j \tilde b_j = 1 $, respectively. For the parity condition to hold, the entries with $\tilde a_j \tilde b_j = 1$ must occur in pairs, which forces the number of entries with $\tilde a_j \tilde b_j = 0$ to also be even. 

                 Therefore, we can always pair the entries of $\vec{z}$ in such a way that we  only have the cases already discussed above, and we can apply the strategy in Eq.~\eqref{eq:T_gate_red} to make the residues reducible and lower the denominator exponent by one. 

            \end{enumerate}
             We conclude that in all cases it is possible to reduce the least denominator exponent from $k$ to $k - 1$,  by using local gate operations from $\tilde \GC$. The process can be iterated until $k = 0$, and then we can permute $\vec{z}$ to $\vec{e_1}$ as explained for the base case.
        \end{itemize}
   \end{proof}
After proving Lemma~\ref{lem:column_lemma}, we can apply it recursively to each column of $Q \in \SO(2n)$, as described in Eq.~\eqref{eq:column_reduction}. At each step, we reduce one orthonormal column vector to  a canonical basis vector  $\vec{e_j}$, while preserving orthogonality with the previously reduced columns, thanks to the group structure of  $\SO(2n)$.

\end{proof}

\end{proof}

\subsection{Proof of Theorem 4}\label{ap:th_4}

In this appendix, we provide a proof for Theorem~\ref{th:gate_counts}. We follow the approach explained in Appendix~\ref{ap:exact_synthesis}, where we bring the $j$-th column to $\vec{e_j}$ by lowering the denominator exponent of the column using $\overline{T}$ gates. Although this synthesis strategy is not optimal, it provides  bounds on the number of $\overline{T}$ gates and Clifford gates needed for exact synthesis.

\begin{theorem}
    Given an $n$-qubit matchgate unitary $U\in\mathbb{SPIN}(2n)$ such that $\Phi(U\otimes U^*)\in\SO(2n)$ has entries in the ring $\mathbb D\left[\sqrt 2\right]$ with maximum least denominator exponent $k_{\max}$, the number of $\overline T$ gates, $N_{T}$, and the number of Clifford gates from $\left\{\overline  S_q\right\}_{q=1}^n \cup {\left\{R_{q,q+1}^{xx}\left(\frac{\pi}{2}\right)\right\}}_{q=1}^{n-1}$, $N_C$, needed to exactly synthesize $U$ are bounded by
    \begin{equation}
        N_T\,\leq\; \frac 23 n^3 k_{\max} + \mathcal O(n^2 k_{\max})\, ,\qquad
       N_C\,\leq\;  \frac 43  n^4 k_{\max}+  \mathcal O(n^3 k_{\max} )\,.
    \end{equation}
\end{theorem}

\begin{proof}
    Let the matrix $Q=\Phi(U\otimes U^*)\in\SO(2n) $ have entries 
    \begin{equation}
        Q_{i,j} =\frac{a_{i,j}+ b_{i,j}\sqrt{2}}{\sqrt{2}^{k_{i,j}}}\,, 
    \end{equation}
    with $a_{i,j},b_{i,j}\in \mathbb Z$. Here $k_{i,j}$ is the least denominator exponent of $Q_{i,j}$ (see Def.~\ref{def:denominator_exponent_app}). We define the least denominator exponent of the matrix $Q$ as the maximum of the entries:
    \begin{equation}
        k_{\max} = \max_{i,j} k_{i,j}\,.
    \end{equation}
    Thus, we can rewrite the entries of our matrix as $\sqrt 2^{k_{\max}} Q_{i,j} =  \alpha_{i,j}+ \beta_{i,j}\sqrt{2} \in \mathbb Z[\sqrt 2]$. In the worst case, $k_{i,j} = k_{\max}$ for all $i,j \in \{1, \ldots, 2n\}$. In this case, we would need $k_{\max}$ $\overline{T}$ gates for each pair of entries in the first column to bring it to $\vec{e_1}$, and we have $n$ pairs. However, applying a $\overline{T}$ gate affects all the entries in the two rows which the targeted pair belongs to. Since we do not control the ordering of the entries besides the target column, in this case the first, in the worst case $\overline{T}$ gate raises the least denominator exponent of the elements of the other columns by one. Hence, applying $k_{\max}$ $\overline{T}$ gates can increase up to $k_{\max}$ the exponent of the rest of the matrix. Thus, for the second column we would need at most $2k_{\max} \left\lceil \frac{2n-1}{2}\right\rceil$ $\overline{T}$ gates. Therefore, the number of $\overline{T}$ gates we would need in the worst case is upper bounded by 
    \begin{equation}
        N_T\leq k_{\max}\frac{2n}{2} + 2k_{\max}\left \lceil \frac{2n-1}{2}\right\rceil + \cdots + (2n-1) \left \lceil \frac{2}{2}\right\rceil = k_{\max}\sum _{j = 1}^{2n-1} j \left \lceil \frac{2n+1-j}{2}\right \rceil
        \label{eq:N_T_bound}\,.
    \end{equation}
    We factor out $k_{\max}$ and focus on the sum
    \begin{equation}
        S :=  \sum_{j=1}^{2n-1} j \left\lceil \frac{2n+1-j}{2} \right\rceil\,.
    \end{equation}
    Since $2n+1-j$ is an integer, we can simplify the ceiling by using that $\left\lceil \frac{m}{2} \right\rceil = \left\lfloor \frac{m+1}{2} \right\rfloor$ for integer $m$: 
    \begin{equation}
        \left\lceil \frac{2n+1-j}{2} \right\rceil  =  \left\lfloor \frac{2n+2-j}{2} \right\rfloor
        = \left\lfloor n+1 - \frac{j}{2} \right\rfloor = n+1 - \left\lceil \frac{j}{2} \right\rceil\,,
    \end{equation}
    where in the last equality we used the identity  $\lfloor m - x \rfloor = m - \lceil x \rceil$ for integer $m$. Now, we substitute it into the sum,
    \begin{equation}
        S = \sum_{j=1}^{2n-1} j \left( n+1 - \left\lceil \frac{j}{2} \right\rceil \right)
    = (n+1) \sum_{j=1}^{2n-1} j - \sum_{j=1}^{2n-1} j \left\lceil \frac{j}{2} \right\rceil\,.
    \end{equation}
    For the first term, we have 
    \begin{equation}
        (n+1)\sum_{j=1}^{2n-1} j = (n+1)\frac{(2n-1) \cdot 2n}{2} = n(2n-1)(n+1)\,.
    \end{equation}
     For the second term, we write $j$ as $2r-1$ (odd), with $r=1,\dots,n$  or $2r$ (even) with $r=1,\dots,n-1$. We have
    \begin{equation}
        \left\lceil \frac{2r-1}{2} \right\rceil = r, \qquad \left\lceil \frac{2r}{2} \right\rceil = r\,.
    \end{equation}
    Thus for each pair of values of $j$, $(2r-1, 2r)$ with $r=1,\ldots,n-1$, we obtain the contribution 
    \begin{equation}
        (2r-1)r + (2r)r = r(4r-1)\,.
    \end{equation}
    Including the contribution $(2n-1)\cdot n$ from the value $j=2n-1$, we have
    \begin{equation}
        \sum_{j=1}^{2n-1} j \left\lceil \frac{j}{2} \right\rceil
    = \sum_{r=1}^{n-1} r(4r-1) + n(2n-1) = 4\sum_{r=1}^{n-1} r^2 - \sum_{r=1}^{n-1} r +n(2n-1)\,.
    \end{equation}
    Now, by using
    \begin{equation}
        \sum_{r=1}^{m} r = \frac{m(m+1)}{2}, \quad \sum_{r=1}^{m} r^2 = \frac{m(m+1)(2m+1)}{6}\,,
    \end{equation}
    with $m=n-1$, we get
    \begin{equation}
        \sum_{j=1}^{2n-1} j \left\lceil \frac{j}{2} \right\rceil
    = \frac{n(n-1)(8n-7)}{6} + n(2n-1) = \frac{n(8n^2 - 3n + 1)}{6}\,.
    \end{equation}
    All together, the final expression for $S$ is given by
    \begin{equation}
        S = n(2n-1)(n+1) - \frac{n(8n^2 - 3n + 1)}{6} = \frac{4n^3 + 9n^2 - 7n}{6}\,.
    \end{equation}
    Including the $k_{\max}$ factor, we arrive at 
    \begin{equation}
        N_T\leq k_{\max} \frac{4n^3 + 9n^2 - 7n}{6} =  \frac{2}{3}k_{\max} n^3 + \mathcal{O}\left( k_{\max} n^2\right)\,.
    \end{equation}

    In a similar fashion, we can derive an upper bound on the number of Clifford operations $N_C$ from $\left\{\overline  S_q\right\}_{q=1}^n \cup {\left\{R_{q,q+1}^{xx}\left(\frac{\pi}{2}\right)\right\}}_{q=1}^{n-1}$ required to exactly synthesize a target matchgate unitary. For the first column, of length $2n$, in the worst case we may need $2n-2$ adjacent signed transpositions to bring the target entry whose denominator exponent we aim to lower to the first position, and another $2n-2$ more to bring the other target entry to the second position. In total, we need $4n-4$ Clifford operations. Taking into account the possible sign correction, this amounts to at most $4n-3$ Clifford operations per application of a $\overline{T}$ gate. Since this strategy is applied column by column, for column $j$ of size $2n+1-j$ the Clifford cost per $\overline{T}$ gate is bounded by $(2n-j-1) + (2n-j-1)+1 = 4n-2j -1$ Clifford operations. Using Eq.~\eqref{eq:N_T_bound}, we obtain
    \begin{equation}
        k_{\max}\sum _{j = 1}^{2n-1} j \left \lceil \frac{2n+1-j}{2}\right \rceil(4n-2j-1) = k_{\max}\frac{2}{3}\,n(n-1)(n+2)(2n-1) =\frac{4}{3} k_{\max}n^4 + \mathcal O\left(k_{\max} n^3\right)\,.
    \end{equation}
    This expression follows from an analogous decomposition of the sum used to obtain the bound for $N_T$.
    
    Finally, when $k_{\max} = 0$ (that is, when $U$ is a matchgate Clifford), we may still need to permute the $\pm 1$ entries of the orthogonal matrix to bring the matrix $Q$ to the identity. By an analogous reasoning, for the first column we may need up to $2n-1$ adjacent transpositions, plus up to two extra operations for the sign correction. This yields at most $2n+1$ operations for the first column. The second column has size $2n-1$ and requires at most $2n$ Clifford operations to bring the nonzero entry to the desired position, and so on. Altogether, we obtain
    \begin{equation}
        \sum_{j = 1}^{2n}(2n+2-j) = 2n(2n+2)-\sum_{j = 1}^{2n}  j = n (2n +3)\,. 
    \end{equation}
    Collecting both contributions, we can bound the total number of Clifford operations $N_C$ from $\GC$ as 
    \begin{equation}
        N_C\leq  \frac 43  k_{\max}n^4  + \OC\left(k_{\max} n^3\right)\,, 
    \end{equation}
    which completes the proof. 
\end{proof}

\subsection{Mapping optimal exact matchgate synthesis to SAT}\label{ap:SAT_mapping}
In Problem~\ref{prob:synthesis}, we formulated optimal exact matchgate synthesis as a decision problem. Here, we show how to map it to a SAT instance, following the approach of Ref.~\cite{gouzien2025provably}.

\medskip

Consider the finite set of generators  $\GC={\left\{\overline  T_q,\,\overline S_q\right\}}_{q=1}^n \cup {\left\{R_{q,q+1}^{xx}\left(\frac{\pi}{2}\right)\right\}}_{q=1}^{n-1}$,  where $\overline S$ and $\overline T$ act on single qubits, and $R^{xx}_{q,q+1}$ on pairs of nearest-neighbor qubits on a line. The number of generators is $M:=  |\GC|  = 3n-1$. 
For any given target matchgate $U\in \mathbb{SPIN}(2n)$, we work at the level of its associated Pauli transfer matrix $Q=\Phi(U\otimes U^*)\in\SO(2n)$. Hence, we will use the corresponding set of generators $\tilde \GC$ in $\SO(2n)$ (as defined in Eq.~\eqref{eq:SO_gates}). 
The entries of all the generators in $\tilde \GC$ lie in the ring 
\begin{equation}
    \mathbb{D}\left[\sqrt{2}\right] := \left\{ a + b\sqrt{2} \;\middle|\; a, b \in \mathbb{D}\right\}\,,
\end{equation}
where $\mathbb{D} := \left\{ \frac{m}{2^l} \;\middle|\; m \in \mathbb{Z},\ l \in \mathbb{N} \right\}$ is the ring of dyadic rationals, i.e. rational numbers whose denominator is a power of 2. 
We consider target matrices $Q\in \SO(2n)$ with coefficients in $\mathbb{D}\left[\sqrt{2}\right]$, which always admit an exact decomposition into a finite sequence of elements from $\tilde \GC$, as guaranteed by Theorem~\ref{th:exact_synthesis}. This fact allows us to formulate the optimal synthesis problem as a SAT instance. 

Following Ref.~\cite{gouzien2025provably}, a straightforward way to encode the synthesis problem into SAT is to introduce Boolean variables $x_ {ij}$, where $i$ denotes the layer of the circuit and $j$ indexes the $M$ generators. Then, for a given maximum depth $d$, we represent the synthesized transfer matrix as
\begin{equation}
    Q = \prod_{i = 1}^d \left[ \sum_{j = 1} ^M x_{ij}G^{(j)} \right]\,,
    \label{eq:unitary_boolean}
\end{equation}
where $G^{(j)}\in \tilde \GC$. We can ensure the selection of exactly one generator per layer, by enforcing $x_{ij} = 1$ for exactly one value of $j$ at each depth $i$. This corresponds to allowing, at each layer, the choice of applying any possible gate from the gateset. In turn, for a sufficient depth $d$ this approach allows to generate any possible combination of gates.
At the Boolean level, the choice condition can be imposed with the following one-hot constraint,
\begin{equation}\label{eq-ap:single_gate_per_layer_bool}
    \left[ \bigvee_j x_{i,j} \right] \land \left[ \bigwedge_{j< j'} \left( \lnot x_{i,j} \lor \lnot x_{i,j'} \right) \right]\,,
\end{equation}
where the first clause imposes that there is at least one $j$ for which $x_{ij}=1$, and the second clause that there is no more than one $j$ for which $x_{ij}=1$.
Later on, we will relax this condition to allow for multiple commuting gates per layer. 

Our goal is to turn Eq.~\eqref{eq:unitary_boolean} into a SAT expression. However, encoding it directly into Boolean formulas is inefficient, as it would involve a sum over all $M^d$ possible gate sequences, causing the number of Boolean clauses to grow exponentially with the depth $d$. 
To mitigate this, we introduce the intermediate matrix variables
\begin{equation}
    W_1 = G_1 = \sum_{j = 1} ^M x_{1j}G^{(j)} , \quad W_2 = W_1 G_2,\quad\dots,   \quad W_d = W_{d-1}G_d \,. 
\end{equation}
Here, $G_i = \sum_j x_{ij}G^{(j)}$ represents the generator selected at depth $i$. By introducing the intermediate matrices $W_i$, we break the full product in Eq.~\eqref{eq:unitary_boolean} into $d$ matrix multiplications. In doing so, we reduce the SAT encoding size from exponential to polynomial in $d$ (we analyze the exact scaling of variables and clauses in the following).

\subsubsection{Depth = 1 }
For $d = 1$, we have $W_1 = G_1 = Q$, and the decision problem reduces to determining if some generator equals $Q$. We first express both $Q$ and the generators $ G^{(j)}$ in terms of their matrix entries in $\mathbb D\left[\sqrt 2\right]$,
\begin{equation}\label{eq-ap:dyadic_entries}
     Q_{\alpha \beta} = u_{\alpha \beta} + v_{\alpha \beta} \sqrt{2}\,, \quad G^{(j)}_{\alpha \beta} =   a_{\alpha \beta}^{(j)} +  {b_{\alpha \beta}^{(j)}}{\sqrt{2}}\,.
\end{equation}
Equating the matrix elements yields~\footnote{Notice that, in practice, it might be convenient to scale the target matrix and the generators by suitable powers of $\sqrt{2}$, to work with numbers from the ring $\mathbb{Z}\left[\sqrt{2}\right]$~\cite{gouzien2025provably}.}
\begin{equation}\label{eq-ap:matrix_equality}
     Q_{\alpha \beta} = \sum_{j = 1}^M x_{1,j}   G^{(j)}_{\alpha \beta}\quad\iff \quad
     u_{\alpha \beta} = \sum_{j = 1}^M x_{1,j} a_{\alpha \beta}^{(j)}\,,  \quad v_{\alpha \beta} =  \sum_{j = 1}^M x_{1,j} b_{\alpha \beta}^{(j)}\,.
\end{equation}
Since $ u_{\alpha \beta}$, $a_{\alpha \beta}^{(j)}$,  $v_{\alpha \beta},b_{\alpha \beta}^{(j)} \in \mathbb D$, each of them can be expressed as a fraction $m/2^l$, for $m\in\mathbb Z$ and $l \in \mathbb N$. Hence, they can be encoded using $l$ bits, for $l$ the largest among their exponents of $2$ in the denominator. We denote the $t$-th bit of $u_{\alpha \beta}$ as $u_{\alpha \beta}[t]$. For each bit position $t$, we introduce the selector sets
\begin{equation}
    S_t^{a}(\alpha,\beta) := \left\{ j \in [1,M] \; \big| \; a_{\alpha \beta}^{(j)}[t] = 1 \right\}, \quad S_t^{b}(\alpha,\beta)  := \left\{ j \in [1,M] \; \big| \; b_{\alpha \beta}^{(j)}[t] = 1 \right\}\,.
\end{equation}
These sets specify, for each matrix entry $(\alpha,\beta)$ and bit position $t$, 
which generators contribute to that particular bit in the dyadic coefficients in Eq.~\eqref{eq-ap:dyadic_entries}. We can now enforce matrix equality~\eqref{eq-ap:matrix_equality} into Conjunctive Normal Form (CNF), with the following clauses for each $(\alpha,\beta)$ and $t$,
\begin{equation}
    \left( \lnot u_{\alpha\beta}[t] \,\lor \bigvee_{j \in S_t^a(\alpha,\beta)} x_{1,j} \right)\,,\quad \left( \lnot v_{\alpha\beta}[t] \,\lor \bigvee_{j \in S_t^b(\alpha,\beta)} x_{1,j} \right)\,.
\end{equation}
These clauses impose that if $u_{\alpha\beta}[t]=1$, then at least one $x_{1,j}\in S_t^a(\alpha,\beta)$  must be equal to $1$ (and analogously for the case $v_{\alpha\beta}[t]=1$). To ensure that when $u_{\alpha\beta}[t]=0$ or $v_{\alpha\beta}[t]=0$ no generator with index in $S_t^a(\alpha,\beta)$ or $S_t^b(\alpha,\beta)$ is selected, we need one additional clause for each $j$ in the selector sets,
\begin{equation}
    \left( \lnot x_{1,j} \lor u_{\alpha\beta}[t] \right)\,,\quad \left( \lnot x_{1,j} \lor v_{\alpha\beta}[t] \right)\,.
\end{equation}
In this way, the SAT encoding enforces that the bits 
of the target matrix entry $Q_{\alpha\beta}$ coincide with the bits of the selected generator. This construction introduces a total of $|S_t^a(\alpha,\beta)|+|S_t^b(\alpha,\beta)| + 2$ clauses to the SAT formula per matrix entry $(\alpha,\beta)$ and output bit $t$. 

\subsubsection{Depth $>$ 1}
We now generalize the SAT encoding to $d > 1$, by presenting the $d=2$ case. The extension to general $d$ follows straightforwardly. 
For depth $d=2$, we have $W_2 = W_1 G_2 = Q$. Thus, 
\begin{equation}
\begin{split}
    Q_{\alpha \beta } =  W_{2_{\alpha \beta }} = \sum_{\gamma = 1}^{2n} W_{1_{\alpha \gamma}} {(G_2)}_{\gamma \beta} = \sum_{\gamma = 1}^{2n} W_{1_{\alpha \gamma}} \sum_{j' = 1 }^M x_{2j'}G^{(j')}_{\gamma \beta} =   \sum_{\gamma = 1}^{2n} {\left(\sum_{j = 1 }^M x_{1j}G^{(j)}\right)}_{\alpha \gamma}{\left( \sum_{j' = 1 }^M x_{2j'}G^{(j')} \right)}_{\gamma\beta }   \\ = \sum_{\gamma = 1}^{2n} \sum_{j = 1 }^M \sum_{j' = 1 }^M  x_{1j} x_{2j'}   \left(a_{\alpha \gamma}^{(j)} + {b_{\alpha \gamma}^{(j)}}{\sqrt 2} \right)\left(a_{\gamma \beta}^{(j')} + {b_{\gamma \beta}^{(j')}}{\sqrt 2} \right)\,.
\end{split}
\end{equation}
To simplify the notation, we write
\begin{equation}
     a_{\alpha \gamma}^{(W_1)} :=  \sum_{j = 1}^M x_{1j} a_{\alpha \gamma}^{(j)}\,, \quad b_{\alpha \gamma}^{(W_1)} := \sum_{j = 1}^M x_{1j} b_{\alpha \gamma}^{(j)}\,,
    \label{eq:def_multiplexer}
\end{equation}
and analogously for $a_{\gamma \beta}^{(G_2)}$ and $b_{\gamma \beta}^{(G_2)}$. This yields 
\begin{equation}
     W_{2_{\alpha \beta}} = \sum_{\gamma = 1}^{2n}\left(a_{\alpha \gamma}^{(W_1)} + {b_{\alpha \gamma}^{(W_1)}}{\sqrt 2} \right)\left(a_{\gamma \beta}^{(G_2)} + {b_{\gamma \beta}^{(G_2)}}{\sqrt 2} \right) = \sum_{\gamma = 1}^{2n}\left(a_{\alpha \gamma}^{(W_1)} a_{\gamma \beta}^{(G_2)} + 2{b_{\alpha \gamma}^{(W_1)} b_{\gamma \beta}^{(G_2)}} + \left( a_{\alpha \gamma}^{(W_1)} b_{\gamma \beta}^{(G_2)} + b_{\alpha \gamma}^{(W_1)} a_{\gamma \beta}^{(G_2)}\right){\sqrt 2} \right )\,.
\end{equation}
As in the previous section, our goal is to match the target Pauli transfer matrix $Q$ by equating the coefficients in $\mathbb D[\sqrt 2]$, i.e.,
\begin{equation}\label{eq-ap:matrix_equalities_d2}
     u_{\alpha \beta} = \sum_{\gamma = 1}^{2n}\left(a_{\alpha \gamma}^{(W_1)} a_{\gamma \beta}^{(G_2)} + 2{b_{\alpha \gamma}^{(W_1)} b_{\gamma \beta}^{(G_2)}}\right)\,, \quad v_{\alpha \beta} = \sum_{\gamma = 1}^{2n}\left(a_{\alpha \gamma}^{(W_1)} b_{\gamma \beta}^{(G_2)} + b_{\alpha \gamma}^{(W_1)} a_{\gamma \beta}^{(G_2)}\right)\,.
\end{equation}
The idea of the encoding is to have clauses that enforce Eq.~\eqref{eq-ap:matrix_equalities_d2}, and clauses that, for values of  $a_{\alpha \beta}^{(\cdot)}$ and $b_{\alpha \beta}^{(\cdot)}$ compatible with Eq.~\eqref{eq-ap:matrix_equalities_d2}, assign correct values to the binary variables $x_{ij}$ according to Eq.~\eqref{eq:def_multiplexer}. Again, we consider that $a_{\alpha \beta}^{(\cdot)}$ and $b_{\alpha \beta}^{(\cdot)}$ can each be represented using $k$ bits. The latter clauses are then completely analogous to those of the $d=1$ case. Namely, for each matrix entry $(\alpha,\beta)$ and each bit position $t$, we have
\begin{equation}
    \left( \lnot a_{\alpha\beta}^{(W_1)}[t]\, \lor \bigvee_{j \in S_t^a(\alpha,\beta)} x_{1,j} \right)\,,\quad \left( \lnot  b_{\alpha\beta}^{(W_1)}[t] \,\lor \bigvee_{j \in S_t^b(\alpha,\beta)} x_{1,j} \right)\,,
\end{equation}
and similarly for $a_{\alpha\beta}^{(G_2)}[t]$ and $b_{\alpha\beta}^{(G_2)}[t]$. In addition, for each $j\in S_t^a(\alpha,\beta)$ and $j\in S_t^b(\alpha,\beta)$, we also have
\begin{equation}
    \left( \lnot x_{1,j} \lor a_{\alpha\beta}^{(W_1)}[t] \right)\,,\quad \left( \lnot x_{1,j} \lor b_{\alpha\beta}^{(W_1)}[t] \right)\,,
\end{equation}
and the corresponding clauses for $a_{\alpha\beta}^{(G_2)}[t]$ and $b_{\alpha\beta}^{(G_2)}[t]$. These clauses do not yet enforce the target equality $W_2=Q$. Rather, they only define, entry-wise, the bits of the generator $G_i$ chosen at layer $i$. For depths $d>1$, the condition $W_d=Q$ is imposed by `bit-blasting' the matrix products $W_i=W_{i-1}G_i$ entry-wise, and then matching the resulting bits of $W_d$ to the fixed bits of $Q$, as in  Eq.~\eqref{eq-ap:matrix_equalities_d2}. 

The multi-bit multiplication and addition operations needed to compute the matrix product are decomposed into Boolean AND, XOR, and carry operations admitting a polynomial-size CNF encoding. 
We illustrate this procedure by showing how to encode multiplication as a $3$-SAT instance in the following example.

\begin{example}{\textbf{Two-bit multiplication.}}
Let $A=A_1A_0$ and $B=B_1B_0$ be two-bit numbers, i.e.,
\begin{equation}
A = A_0 + 2A_1,\qquad B = B_0 + 2B_1\,,
\end{equation}
with $A_i,B_j\in\{0,1\}$. Define the partial products
\begin{equation}
p_{ij} \;:=\; A_i \land B_j \qquad i,j\in\{0,1\}\,.
\end{equation}

First, we expand and group the product $P=A\cdot B$
\begin{equation}
P=(A_0+2A_1)(B_0+2B_1)
  = A_0B_0 + 2(A_1B_0 + A_0B_1) + 4A_1B_1\,,
\end{equation}
or, in terms of partial products,
\begin{equation}
P \;=\; p_{00} \;+\; 2(p_{10}+p_{01}) \;+\; 4p_{11}\,.
\end{equation}
This shows which AND operations contribute to each output bit: $p_{00}$ at $2^0$, $p_{10}$ and $p_{01}$ at $2^1$, $p_{11}$ at $2^2$. Then, one can convert additions into XOR and  carry bits. Adding two bits $x+y$ can be written as
\begin{equation}
x+y \;=\; (x\oplus y) + 2(x\land y)\,.
\end{equation}
By introducing carry variables $c_1,c_2\in\{0,1\}$,  the output bits $P=P_3P_2P_1P_0$ can be expressed as 
\begin{equation}
    P_0 = p_{00}\,, \quad  P_1 = p_{10}\oplus p_{01}\,,
\qquad
c_1 = p_{10}\land p_{01}\,,
\qquad P_2 = p_{11}\oplus c_1\,,
\qquad
c_2 = p_{11}\land c_1\,,
\qquad
P_3 = c_2\,.
\end{equation}

After this, we can use a CNF encoding. For instance, each AND and XOR operations can be translated into CNF. As an example, consider
$s = x\oplus y$, which is equivalent to
\begin{equation}
(\lnot x \lor \lnot y \lor \lnot s)\ \land \
(\lnot x \lor y \lor s)\ \land \
(x \lor \lnot y \lor s)\ \land \
(x \lor y \lor \lnot s)\,.
\end{equation}
\end{example}

Although we illustrate the procedure with a two-bit example, the same construction generalizes to arbitrary bit lengths. Multi-bit multiplication can always be decomposed into Boolean AND, XOR, and carry operations, each of which admits a polynomial-size CNF encoding. Similarly, multi-bit addition can be expressed as a $3$-CNF encoding ($3$-SAT). Hence, matrix multiplication can always be bit-blasted, and then matrix equality to the target $Q$ is enforced.

With these ingredients in place, we can now characterize the total complexity of mapping the optimal exact matchgate synthesis problem to a SAT instance.
\begin{itemize}
    \item Number of variables: 
    \begin{itemize}
        \item Generator selectors, $x_{ij}$. We have $M$ generators and depth $d$, so  $\mathcal O (dM)$.
        \item Intermediate matrix variables, $W_i$. We have $d-1$ intermediate gates, we have matrices with $4n^2$ entries, and each entry needs $2l$ bits of precision. In total,  $\mathcal O (d n^2l)$.
        \item Variables resulting from the multiplication (intermediate and carry bits). We need $\mathcal O (d n^3 l)$. 
    \end{itemize}
    In total, the number of variables scales as $\mathcal O (d n^3 l)$.
    \item Number of clauses: 
    \begin{itemize}
        \item Multiplexers: $a^{(i)}_{\alpha \beta}: \mathcal O (n^2 M l d)$. We have $M$ possibilities, we take into account the total size of the matrix ($4n^2$), each entry is represented by $2l$ bits, and we have $d$ layers. 
        \item Bitwise multiplication: $\mathcal O(dn^3k^2)$. $l^2$ bits per product of variables, $2n$ for the sum in the matrix multiplication, $4n^2$ entries in each matrix, and one for each layer (total depth $d$). 
        \item Addition: $\mathcal O (n^3l)$
        \item Matching the values of the final target $Q_{\alpha \beta } $: $\mathcal O (n^2 l)$. We have $4n^2$ matrix entries, each one represented by $2l$ bits. 
    \end{itemize}
\end{itemize}
Considering that the number of gates $M$ and the bits needed to represent the matrix entries $l$ both scale linearly with $n$, we have that the total number of clauses scales as $\mathcal O (d n^5)$, whereas the number of variables scales as $\mathcal O (dn^4)$. 

\subsubsection{Allowing multiple commuting gates per layer}\label{ap:parallel_layers}

We next adapt the SAT clauses previously introduced to enable the possibility that a single layer may contain several commuting matchgates, applied in parallel on disjoint qubits.  In this setting, a layer in the standard representation of $\SO(2n)$ is now equivalent to a product of commuting Givens rotations acting on disjoint pairs of rows.  Here we explain how to extend the SAT encoding so as to handle such layers with parallelized gates. 
Defining the filter
\begin{equation}
    F_q(j) = \begin{cases}
        1\quad {\rm if}\;G^{(j)} \text{ acts on qubit $q$}\\
        0\quad {\rm otherwise}\,,
    \end{cases}
\end{equation}
we can enforce the construction via the following clauses, for each layer $i$,
\begin{align}\label{eq-ap:parallel_gates_bool}
    \left[\bigvee_{j} x_{i,j}\right] \wedge \bigwedge_{q=1}^n
    \left[\bigwedge_{\substack{j<k\\F_q(j)=F_q(k)=1}} \left(\lnot x_{i,j}\;\vee\;\lnot x_{i,k}\right)\right].
\end{align}
Indeed, for each layer $i$, the first boolean constraint enforces that at least one gate is applied, while the second one enforces that at most one gate acts on each qubit $q$.
Once Eq.~\eqref{eq-ap:single_gate_per_layer_bool} is replaced with Eq.~\eqref{eq-ap:parallel_gates_bool}, we can proceed using the same logic as outlined for the single gate per layer case.

\subsection{Exact preparation of Gaussian states}\label{app:gaussian_state_prep}
In the previous section, we developed a SAT-based encoding for exact matchgate synthesis. Building on that framework, we now discuss how to exactly prepare Gaussian states, using the reduced number of constraints involved in the state-preparation task.

A Gaussian state can always be expressed as the zero state ${|0\rangle}^{\otimes n}$ evolved by some matchgate unitary $U_g$, that is, $|\psi_g\rangle = U_g{|0\rangle}^{\otimes n}$. The information extracted from a Gaussian state via expectation value measurements can be compressed into the covariance matrix, which is given by
\begin{equation}
    \Gamma_{\alpha \beta} = \frac i 2\langle \psi_g |[c_\alpha, c_\beta]|\psi_g\rangle\,.
\end{equation}
In particular, consider the isomorphism 
\begin{equation}\label{eq:U_g_isomorphism}
    U_g^\dagger c_\alpha U_g = \sum_{\beta = 1}^{2n} Q_{\alpha \beta} c_\beta\,,
\end{equation}
where $Q\in \mathbb{SO}(2n)$ and $c_\alpha$ are Majorana operators. 
Then, the covariance matrix that corresponds to $|\psi_g\rangle = U_g|0\rangle^{\otimes n} $ can be rewritten as a transformation of the initial vacuum covariance $\Gamma_0$ (the covariance matrix of the $|0\rangle^{\otimes n}$ state): 
\begin{equation}
\begin{split}
    \Gamma_{\alpha \beta}
    &= \frac{i}{2}\,
       \langle 0|^{\otimes n} U_g^\dagger [c_\alpha, c_\beta] U_g \ket{0}^{\otimes n} \\
    &= \frac{i}{2}\,
       \langle 0|^{\otimes n} [U_g^\dagger c_\alpha U_g,\, U_g^\dagger c_\beta U_g] \ket{0}^{\otimes n} \\
    &= \frac{i}{2}\,
       \langle 0|^{\otimes n}
         \left[ \sum_{a=1}^{2n} Q_{\alpha a} c_a,\;
                \sum_{b=1}^{2n} Q_{\beta b} c_b \right]
       \ket{0}^{\otimes n} \\
    &= \sum_{a,b=1}^{2n} Q_{\alpha a} Q_{\beta b}
       \frac{i}{2} \langle 0|^{\otimes n} [c_a, c_b] \ket{0}^{\otimes n} \\
    &= \sum_{a,b=1}^{2n} Q_{\alpha a} {(\Gamma_0)}_{ab} Q_{\beta b}
     = \sum_{a,b=1}^{2n} Q_{\alpha a} {(\Gamma_0)}_{ab} {(Q^T)}_{b\beta}\,.
\end{split}
\end{equation} 
In terms of matrix multiplication, we can re-express the previous equation as 
\begin{equation}
    \Gamma = Q\Gamma_0 Q^T\,.
    \label{eq:covariance_transformation}
\end{equation}
Covariance matrices satisfy $\Gamma ^2 = -\id $ and have eigenvalues $\pm i$. Also, the zero state can be expressed as $|0\rangle \langle 0|^{\otimes n} = {\left(\frac{\id + Z}{2}\right)}^{\otimes n}$. This implies that its associated covariance matrix has support only on the Pauli strings containing exclusively products of $Z$. These are represented by Majorana pairs of type $c_{2k-1}c_{2k}$, which is reflected in the entries of $\Gamma_0$,
\begin{equation}
    \Gamma_0 = \bigoplus_{j= 1}^n i Y = \bigoplus_{j= 1}^n \begin{pmatrix}
        0 & 1 \\ -1 & 0 
    \end{pmatrix}\,.
    \label{eq:cov_matrix_zero}
\end{equation}

Let us define the $n$-dimensional complex $(+i)$ eigenspace of $\Gamma$ and $\Gamma_0$,
\begin{equation}
\begin{split}
    V_+& := \text{Eig}_{+i}(\Gamma) \subset \mathbb C^{2n}\,,\\
    V_{0,+}& := \text{Eig}_{+i}(\Gamma_0) \subset \mathbb C^{2n}\,.
\end{split}
\end{equation}
These eigenspaces can be interpreted as $2n\times n$ matrices whose $n$ columns are the $+i$ eigenvectors of the corresponding covariance matrix. For instance, the $+i$ eigenvectors of $\Gamma_0$ are 
\begin{equation}
    f_k := e_{2k-1}+ i  e_{2k}\in \mathbb C^{2n}\,,
\end{equation}
where $e_j$ are the standard basis vectors and $k \in \{1,2,\ldots, n\}$. The $(+i)$-eigenspace of $\Gamma_0$ is the span of such vectors:
\begin{equation}
    V_{0,+} := \mathrm{span}_\mathbb{C}\{f_1,\dots,f_n\}\,.
\end{equation}

As we show in the following Lemma, the subspace $V_+$ uniquely determines the state $\Gamma$.
\begin{lemma}
There is a one-to-one correspondence between the following two objects:

\begin{itemize}
    \item[(1)] Real linear maps $\Gamma : \mathbb{R}^{2n}\to \mathbb{R}^{2n}$ such that $\Gamma^2=-\id$.
    \item[(2)] $n$-dimensional complex subspaces $V_+ \subset \mathbb{C}^{2n}$ satisfying
    \begin{equation} 
        \mathbb{C}^{2n} = V_+ \oplus \overline{V_+}\,.
    \end{equation}
\end{itemize}

The correspondence is given by:
\begin{equation}
\Gamma \longmapsto V_+ := \Eig_{+i}(\Gamma),
\qquad\text{and}\qquad
V_+ \longmapsto \Gamma := i(P_+ - P_-)\,,
\label{eq:correspondence_gaussian}
\end{equation}
where $P_+$ and $P_-$ are the projections onto $V_+$ and $\overline{V_+}$, respectively.
\end{lemma}
This result is discussed in several standard textbooks, see for example Ref.~\cite{nomizu1969foundations}.

\begin{proof}
We construct the correspondence in both directions and then show that these constructions are mutually inverse.

We first show the first direction of the correspondence, from $\Gamma$ to $V_+$.  If $\Gamma$ is real and satisfies $\Gamma^2=-\id$, then its eigenvalues (over $\mathbb{C}$) are $\pm i$.  
Complex conjugation maps $\Eig_{+i}(\Gamma)$ isomorphically onto $\Eig_{-i}(\Gamma)$, so both have the same dimension.  
Since $\Gamma$ is diagonalizable over $\mathbb{C}$, we obtain the direct sum decomposition
\begin{equation}
\mathbb{C}^{2n} = \Eig_{+i}(\Gamma) \oplus \Eig_{-i}(\Gamma)\,,
\end{equation}
and hence $\dim(\Eig_{+i}(\Gamma))=n$.  
Let us define $V_+ := \Eig_{+i}(\Gamma)$, then we can write $\mathbb{C}^{2n} = V_+ \oplus \overline{V_+}$.

Now, we show the correspondence from $V_+$ to $\Gamma$. Assume $V_+\subset \mathbb{C}^{2n}$ has dimension $n$ and satisfies  
$\mathbb{C}^{2n} = V_+ \oplus \overline{V_+}$.  
Let $P_+$ and $P_-$ denote the projections onto $V_+$ and $\overline{V_+}$.  
Define
\begin{equation}
\Gamma := i(P_+ - P_-)\,.
\end{equation}
Using that $P_+^2=P_+$ and $P_-^2=P_-$ together with  $P_+P_-=0$ and $P_++P_-=\id$, one can obtain
\begin{equation}
\Gamma^2 = i^2 {(P_+ - P_-)}^2 = - (P_+ + P_-) = -\id\,.
\end{equation}
Moreover, for $v\in V_+$ we have $P_+ v = v$ and $P_- v = 0$, hence
\begin{equation}
\Gamma v = i(P_+ - P_-)v = i v\,,
\end{equation}
so $V_+ \subset \Eig_{+i}(\Gamma)$.  
Since both spaces have dimension $n$, they coincide:
\begin{equation}
\Eig_{+i}(\Gamma) = V_+\,.
\end{equation}

Finally, it remains to show that both constructions in Eq.~\eqref{eq:correspondence_gaussian} are inverses. Starting with $\Gamma$ and defining $V_+= \Eig_{+i}(\Gamma)$, then reconstructing $\Gamma$ via $\Gamma' = i(P_+ - P_-)$ yields $\Gamma' = \Gamma$, since both are diagonalizable with the same eigenspaces $V_+$ and $\overline{V_+}$ and eigenvalues $\pm i$. Conversely, starting with $V_+$ and constructing $\Gamma = i(P_+ - P_-)$, one obtains $\Eig_{+i}(\Gamma) = V_+$ by the computation above. Therefore, the assignments are mutually inverse, establishing the bijection.
\end{proof}

Having established this one-to-one correspondence, to prepare a given state $\Gamma$ it suffices to verify that
\begin{equation}
    V_+ \;=\; Q V_{0,+}
    \;=\;
    \bigl[\,q_1 + i q_2 \;\;\; q_3 + i q_4 \;\;\; \cdots \;\;\; q_{2n-1} + i q_{2n}\,\bigr]\,,
\end{equation}
Importantly, since $V_+$ is a $2n \times n$ matrix, this formulation reduces the number of algebraic constraints by a factor of two compared to full matchgate synthesis (which requires $2n \times 2n$ entry equalities). The unitary $U_g$ associated with any $Q$ satisfying this relation prepares $|\psi_g\rangle$ up to a global sign.

\subsection{Details on the $XX$-diagonalizing circuit}\label{ap:compliation_XX_diagonalizing}
As a nontrivial benchmark for our SAT-based synthesis algorithm we consider the explicit diagonalization circuits for integrable spin chains constructed in Ref.~\cite{verstraete2009quantum}. In that work, the authors design a unitary $U_{\mathrm{dis}}$ that maps the XY Hamiltonian to a free-fermion model via
\begin{equation}
  H_{\mathrm{XY}} \;=\; U_{\mathrm{dis}} \,\tilde H\, U_{\mathrm{dis}}^\dagger\,,
\end{equation}
where $\tilde H$ is a non-interacting Hamiltonian that can be chosen, without loss of generality, as a sum of decoupled single-qubit $Z$ terms. The circuit is obtained by composing a fermionic Fourier transform with a Bogoliubov rotation. For small system sizes the resulting circuit can be written explicitly in terms of a small set of local gates, see Figures 1 and 2 of Ref.~\cite{verstraete2009quantum}.

Here we focus on the XX limit of the XY chain, i.e.\@ when no anisotropy is present, with vanishing transverse field, for system sizes $n=4$ and $n=8$. In this regime (corresponding to setting $\gamma=\lambda=0$ in Eq.~(5) in Ref.~\cite{verstraete2009quantum}). The gates involved in $U_{\mathrm{dis}}$ are: the fermionic swap gate $U_{\mathrm{fSWAP}}$, the two-mode Fourier gates $F_k$ and the Bogoliubov transformation $B_k$. In the computational basis, these gates are given by: 
\begin{equation}
U_{\mathrm{fSWAP}} =
    \begin{pmatrix}
      1 & 0 & 0 & 0 \\
      0 & 0 & 1 & 0 \\
      0 & 1 & 0 & 0 \\
      0 & 0 & 0 & -1
    \end{pmatrix}, 
    \quad F_k =
    \begin{pmatrix}
      1 & 0 & 0 & 0 \\
      0 & \frac{1}{\sqrt{2}} & \frac{\alpha(k)}{\sqrt{2}} & 0 \\
      0 & \frac{1}{\sqrt{2}} & -\frac{\alpha(k)}{\sqrt{2}}  & 0 \\
      0 & 0 & 0 & -\alpha(k)
    \end{pmatrix}, \quad B_k =
    \begin{pmatrix}
      \cos\theta_k & 0 & 0 &  i\sin\theta_k \\
      0 & 1 & 0 & 0 \\
      0 & 0 & 1  & 0 \\
      i\sin\theta_k  & 0 & 0 &  \cos\theta_k
    \end{pmatrix},
\end{equation}

where $\alpha(k)=e^{i2\pi k/n}$ and $\theta_k=\frac{2\pi
k}{n}$. For $n=4$ and $n=8$ the phases in $\alpha(k)$ are $4$ and $8$-th roots of unity, so their real and imaginary parts belong to the set $\{0,\pm 1,\pm 1/\sqrt 2\}$. When mapping the fermionic circuit $U_{\mathrm{dis}}$ to the standard representation $Q_{\mathrm{dis}}\in\SO(2n)$ acting on Majorana operators, these phases enter only through $\cos$ and $\sin$ of multiples of $\pi/4$. Consequently, every entry of $Q_{\mathrm{dis}}$ lies in the dyadic ring $\mathbb D[\sqrt 2]$ of Eq.~\eqref{eq:D_sqrt2_ring}.

Because of Theorem~\ref{th:exact_synthesis}, we know that we can find a matchgate sequence that exactly implements $Q_{\mathrm{dis}}$. Therefore, we feed this target to the SAT solver, as detailed in Sec.~\ref{sec:SAT}. For each candidate solution with depth $d$, we use the parallel-layer encoding of Appendix~\ref{ap:parallel_layers} to describe all brickwork matchgate circuits consistent with the nearest-neighbor geometry.  We solve the resulting instances with \texttt{kissat}~\cite{biere2024cadical} when aiming strictly for depth optimality, and with \texttt{open-wbo}~\cite{martins2014open} for $T$-count minimization as a MAX-SAT objective.

\subsection{Comparison between SAT-based and Gaussian-elimination exact compilation}

In this section, we compare the solutions to the exact synthesis problem found by the algorithm used to prove Theorem~\ref{th:exact_synthesis}, and the SAT-based compiler. Given the very nature of the latter, it will necessarily result in a shallower (in fact depth-optimal) circuit, at the cost of a higher computational complexity.
Taking as an example the $XX$-diagonalizing circuit studied in the main text, in Fig.~\ref{fig:gauss-sat-comparison} we show the compiled circuits for $n=4$. We can appreciate how the depth-optimal (and therein $\overline{T}$-optimal) SAT solution has less than half the depth of the circuit found by Gaussian elimination. Interestingly, we find that the two solutions have the same $\overline{T}$-count, likely signaling a profound structure of the target circuit. For $n=8$, given the large depth of the solution found by Gaussian elimination, we avoid drawing it. Instead, we report that its depth is $d=113$ (to be compared with $d=25$ for the SAT-based compiler), while its $\overline{T}$-count is $34$. Surprisingly, the latter is again the same as the one found by SAT.

\begin{figure}
    \centering
    \includegraphics[width=\linewidth]{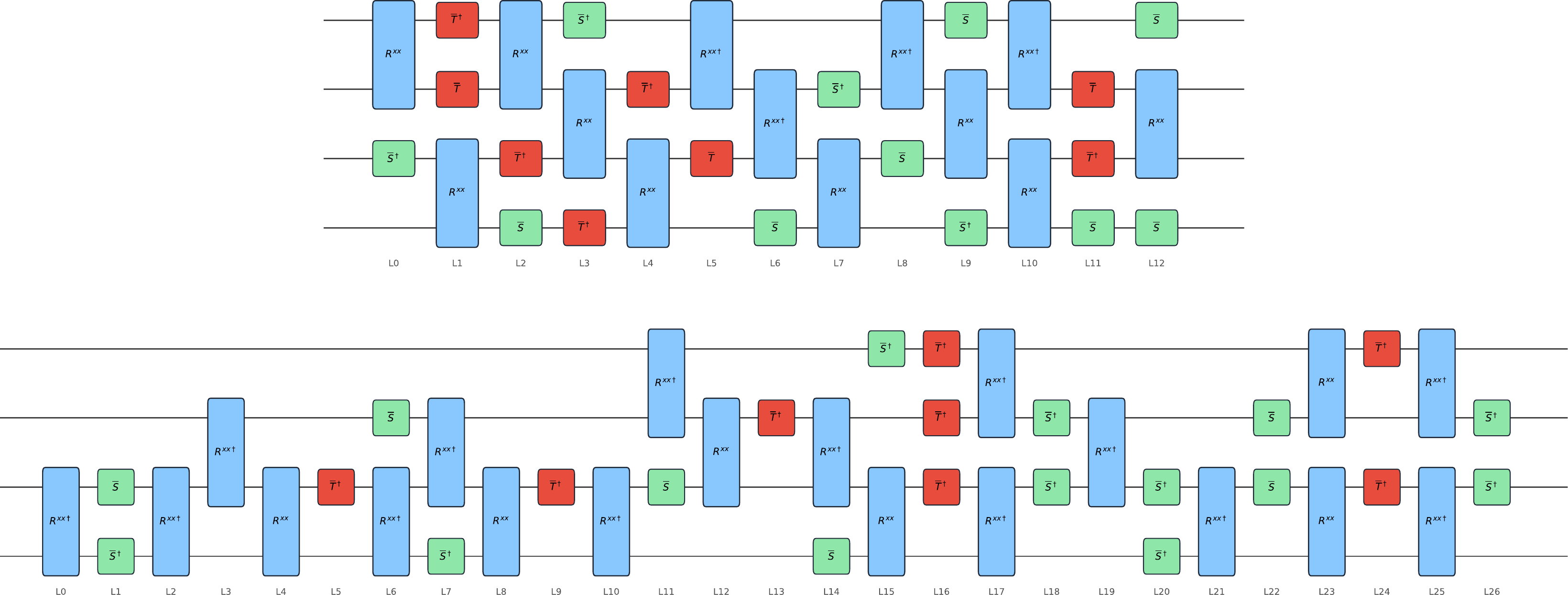}
    \caption{\textbf{Comparison between SAT-based and Gaussian-elimination exact synthesis of the $XX$-diagonalizing circuit on $n=4$ qubits.} Top: the solution found by the SAT-based compiler, with an optimal depth $d=13$ and a $\overline{T}$-count of $8$. Bottom: the decomposition found via Gaussian elimination, whose depth is $d=27$ and has the same $\overline{T}$-count.}
    \label{fig:gauss-sat-comparison}
\end{figure}

\end{document}